\documentclass[english,12pt]{article}

\usepackage[utf8]{inputenc}
\usepackage{geometry}
\usepackage{setspace}
\usepackage[authoryear]{natbib}
\usepackage{microtype}
\usepackage{units}
\usepackage{csquotes}
\usepackage{chngcntr}  
\usepackage{appendix}  
\usepackage{float}
\usepackage{xcolor}
\definecolor{darkblue}{rgb}{0,0,0.5} 
\usepackage{subcaption}
\usepackage{graphicx}
\usepackage{rotating}

\usepackage{amssymb,amsmath,amsfonts,amsthm,xcolor,graphicx,bm}
\usepackage{mathpazo}
\usepackage{setspace}
\usepackage{tikz}
\usepackage[affil-it]{authblk}
\usepackage{babel}
\usepackage{datetime}
\usepackage{booktabs,comment}
\usepackage[T1]{fontenc}

\usepackage{geometry}
\usepackage[linewidth=1pt]{mdframed}
\usepackage{dcolumn}
\usepackage{multirow}
\usepackage[font=footnotesize,position=top,labelfont={it},labelsep=period]{caption}
\usepackage{rotating}
\usepackage[shortlabels]{enumitem}

\newcommand{\okina}{%
  \raisebox{1.25ex}{\rotatebox[origin=c]{180}{,}}%
}

\DeclareSymbolFont{operators}   {OT1}{lmr} {m}{n}
\DeclareSymbolFont{letters}     {OML}{cmm} {m}{it}
\DeclareSymbolFont{symbols}     {OMS}{cmsy}{m}{n}
\SetSymbolFont{operators}{bold}{OT1}{cmr}{b}{n}
\SetSymbolFont{letters}{bold}{OML}{cmm}{b}{it}
\SetSymbolFont{symbols}{bold}{OMS}{cmsy}{b}{n}

\theoremstyle{plain}
\newtheorem{prop}{\protect\propositionname}
\theoremstyle{remark}

\newtheorem{corr}{{Corollary}}
\theoremstyle{plain}

\theoremstyle{plain}
\newtheorem*{defn*}{\protect\definitionname}

\newtheorem*{asn*}{Assumption}

\newtheorem*{lem*}{\protect\lemmaname}
\providecommand{\notename}{Note}
\providecommand{\propositionname}{Proposition}
\providecommand{\propositionname}{Lemma}
\providecommand{\definitionname}{Definition}
\providecommand{\lemmaname}{Lemma}
\providecommand{\theoremname}{Theorem}

\usepackage[unicode=true,pdftex, pdfauthor={Ashley C. Craig, Thomas Lloyd, and Dylan T. Moore},pdftitle={Do Distributional Concerns Justify Lower Environmental Taxes?}]{hyperref}

\newcommand\blfootnote[1]{%
  \begingroup
  \renewcommand\thefootnote{}\footnote{#1}%
  \addtocounter{footnote}{-1}%
  \endgroup
}

\usepackage{hyperref} 
\hypersetup{
    pdfdisplaydoctitle,
    colorlinks=true,    
    citecolor=darkblue, 
    linkcolor=darkblue, 
    urlcolor=darkblue   
}

\makeatletter

\newcommand{\lyxdeleted}[3]{}

\makeatother

\title{Do Distributional Concerns Justify\\Lower Environmental Taxes?}

\author{Ashley C. Craig}

\vspace{-0.1em}
\affil{Australian National University}

\author[2]{Thomas Lloyd}
\affil[2]{University of Michigan}

\author[3]{Dylan T. Moore}
\affil[3]{University of Hawai\okina i at Mānoa}

\date{ \vspace{-0.5em} December 2025 \vspace{-0.2em}}

\begin{document}

\maketitle
\thispagestyle{empty}
\begin{abstract}

    How should taxes on externality-generating activities be adjusted if they are regressive? In our model, the government raises revenue using distortionary income and commodity taxes. If more or less productive people have identical tastes for externality-generating consumption, the government optimally imposes a Pigouvian tax equal to the marginal damage from the externality. This is true regardless of whether the tax is regressive. But, if regressivity reflects different preferences of people with different incomes rather than solely income effects, the optimal tax differs from the Pigouvian benchmark. We derive sufficient statistics for optimal policy, and use them to study carbon taxation in the United States. Our empirical results suggest an optimal carbon tax that is remarkably close to the Pigouvian level, but with higher carbon taxes for very high-income households if this is feasible. When we allow for heterogeneity in preferences at each income level as well as across the income distribution, our optimal tax schedules are further attenuated toward the Pigouvian benchmark.
    
\end{abstract}

\vspace{0.6em}
\noindent \textbf{JEL codes}: H21, H23, Q52, Q54 \\
\noindent \textbf{Keywords}: optimal taxation, externalities, corrective taxation, income taxation, climate change, carbon tax, commodity taxation, redistribution\blfootnote{\scriptsize \hskip -18pt \textbf{Declarations of interest}: None.}
\blfootnote{\scriptsize \hskip -18pt \textbf{Acknowledgements}: We thank Kai Song (Australian National University) for excellent research assistance. We are grateful to Charlie Brown, James Hines, Matthew Lilley, Matthew Shapiro, Itai Sher, Damián Vergara, and Celine Wasem for comments and suggestions, as well as seminar participants at the University of Michigan, U. Mass. Amherst, the Australian National University (ANU), the 2025 IIPF Annual Congress, and the 2025 National Tax Association Annual Meeting. Financial support from the University of Michigan and the ANU is gratefully acknowledged. This research was approved by the institutional review board at the Australian National University (H/2024/1175) and exempted at the University of Michigan (HUM00261013).}

\pagebreak
\setcounter{page}{1}
\section{Introduction}

\vspace{-0.7em}

Human-generated carbon emissions are widely accepted to be causing large social costs, both now and in the future \citep{IPCC_AR6_SYR_2023}. Economists view this as a classic case for government intervention, because there is an externality: Consumers and producers do not bear the full cost of the emissions they generate. Private decision-making therefore leads to greater emissions than is socially optimal.

Most economists also agree that the solution is to impose a ``Pigouvian'' tax on goods that generate negative externalities---in this case, goods whose production or consumption generates carbon emissions \citep{EconomistsStatement2019}. By setting this tax equal to the marginal damage to others from emissions, consumers and producers are forced to pay the true social costs of their choices \citep{Pigou1920,sandmo1975optimal}. This Pigouvian tax would be optimal if non-distortionary lump sum compensation of winners and losers were possible.

In practice, carbon taxation has proved unpopular. One reason for this is that carbon taxes are regressive \citep{HassettMathurMetcalf2009,GraingerKolstad2010,MathurMorris2014,Metcalf2009,WilliamsGordonBurtrawCarboneMorgenstern2015,pizer2019distributional}, and people are concerned about the distributional impacts of adopting them \citep{DietzAtkinson2010,Dechezlepretre2025}. As we confirm in our data, a uniform tax on all carbon-generating activities would lead poorer households to pay more as a share of their incomes (or total expenditures) than richer households. This raises the question of whether corrective taxes should be lowered relative to the Pigouvian benchmark if they are regressive (or raised if they were progressive). We answer that question. Specifically, we study how the optimal carbon tax changes if the government is concerned about inequality, but can only compensate households using efficiency-reducing income and commodity taxes.

Distributional impacts of policies such as a carbon tax cannot readily be dismissed if redistribution is costly. The imposition of the tax produces winners and losers. Compensating the losers by taxing the winners (through adjustments to the income tax or other tax bases) could lead to a greater loss in efficiency than the gain from correcting the externality in the first place. In other cases, the combined package of carbon taxation and compensation could lead to net efficiency gains. Only in knife-edge cases, such as with strong restrictions on household utility functions, is the carbon tax optimally set equal to the Pigouvian benchmark \citep{Kaplow2012}. 

We shed light on optimal externality correction by deriving sufficient statistics formulas that apply when governments can raise revenue with a flexible but distortionary income tax combined with a Pigouvian tax. Using carbon emissions as a specific example, we then implement our approach for electricity, natural gas, heating fuels, and gasoline consumption in the United States. Our results suggest that the optimal carbon tax on these goods is remarkably close to the Pigouvian prescription.

Our focus on whether the carbon tax should be changed for distributional reasons is distinct from questions of how much to compensate households, which others discuss in detail \citep{GraingerKolstad2010,GoulderHafsteadKimLong2019_JPubE}. Our framework does incorporate flexible compensation, but the pattern of that compensation has no direct connection to the question of whether the carbon tax should differ from the Pigouvian level. 

Heterogeneity in preferences (broadly defined) plays a pivotal role in our analysis. If higher-income households have a relative preference for carbon intensity, so that they would choose a more carbon-intensive consumption bundle at any level of disposable income, then the optimal tax on carbon is greater than the marginal damage. If lower-income households are more carbon-intensive, the opposite is true. Preference heterogeneity can even reverse the sign of the optimal tax and make a carbon subsidy optimal. 

Preference heterogeneity is important because it determines the relative cost of redistribution via income vs. carbon taxes. If the carbon tax is set at the Pigouvian level, households allocate their after-tax income efficiently because the externality is exactly internalized. Starting from this level, raising the carbon tax lowers labor supply, while reducing income taxes to compensate increases labor supply. If all households have identical preferences for carbon intensity, the two labor supply effects cancel out. But if higher productivity households have a relative preference for carbon intensity in the sense that they would emit more carbon at any given level of income, there is a net increase in labor supply. In this case, raising the carbon tax redistributes more efficiently than income taxation. By similar logic, if lower productivity types have more carbon-intensive preferences, lowering the carbon tax below the Pigouvian level increases efficiency.

Our model highlights the need to specifically isolate differences in \textit{preferences} between people with different levels of income. This is distinct from income effects, which can also drive differences in carbon intensity along the income distribution.\footnote{The relevant distinction for our analysis is whether differences in choices by people with different incomes simply reflect differences in the amount they have to spend (income effects), or factors that are independent of income, which would drive different choices even conditional on income. The latter category includes true ``preferences'', but also habits or constraints that arise due to factors such as where one lives.} Indeed, it is clearly possible for any given cross-sectional relationship between income and carbon intensity to be explained exclusively by income effects, by preference heterogeneity, or by some combination of the two. The fact that poorer households are more carbon-intensive, and carbon taxes correspondingly regressive, does not therefore suffice to tell us whether the carbon tax should deviate from the Pigouvian benchmark.

Put more formally in the language of optimal taxation, information about carbon intensity is only useful for redistribution if carbon intensity is a ``tag'' in the sense of \citet{Akerlof1978_tagging}. This is the case if it provides direct information about a household's type beyond what is revealed by an individual's choice of income. For example, if carbon intensity were an immutable characteristic that co-varied with income, then taxing carbon-intensive goods would allow some redistribution without efficiency costs.\footnote{\citet{MankiwWeinzierl2010} study the example of taxing an individual's height, which is close to this ideal. See also \citet{Parsons1996ImperfectTagging,Kaplow2006OptimalIncomeTransfers}; and \citet{CremerGahvariLozachmeur2010}.} In turn, this would provide a reason other than externality correction to tax carbon-intensive consumption. Conversely, if all variation in carbon intensity comes from income effects, then taxing carbon-intensive goods has no advantage over income taxation for the purpose of redistribution, so there is no reason to deviate from the Pigouvian level of carbon taxation.

Assuming that differences in preferences do justify deviating from the Pigouvian level of taxation of the externality, distributional concerns lead the government to trade off externality correction and redistribution. As a result, our model highlights that the optimal tax depends on how responsive households are to the carbon tax and income tax. If the income tax is highly distortionary, it can be optimal to use the carbon tax to achieve distributional aims, even if there is a large cost of setting a tax that is different from the Pigouvian level. Using the carbon tax to redistribute is also attractive if doing so has little behavioral impact on carbon intensity. Alternatively, if changing the carbon tax sharply changes carbon emissions, using the carbon tax for redistribution is less desirable.

Unlike the classic Pigouvian case without distributional concerns, the optimal externality correction is not necessarily the same at all levels of consumption. 
To the extent possible, using means-tested vouchers or non-linear electricity and gas pricing, the optimal carbon tax schedule may be progressive or regressive. There are many real-life examples of non-linear pricing and taxation schemes such as electricity tiered pricing, congestion taxes, and emission-based vehicle registration fees or subsidies \citep[see e.g.,][]{Zhang2015ChinaTieredTariff,ACEA2020CO2VehicleTaxes}. Whether such non-linearity is optimal depends on the rate at which the preference for carbon intensity changes with income in different parts of the distribution. It does not depend on whether the carbon tax is regressive or progressive in the sense that poorer or richer households pay more as a share of their incomes.

Building on previous work by \citet{ferey2024sufficient} on commodity taxation, our model provides a convenient way to identify heterogeneity in preferences for people with different status quo incomes. To do so, we compare two statistics. First, we measure the cross-sectional relationship between income and consumption of carbon-intensive goods. This captures both income effects and preference heterogeneity. Second, we measure income effects by quantifying the causal impact of higher income on a given household's consumption of carbon-intensive goods. Finally, we subtract income effects from the cross-sectional relationship to yield the component that arises solely from preferences.

We implement this theoretical framework for the United States using survey evidence from 46,918 households in the Consumer Expenditure Surveys (CEX) and 1,448 respondents from a custom online survey. Our analysis focuses on the four expenditure categories that account for the largest share of carbon emissions aggregated into a single \enquote{dirty good}: gasoline, natural gas, electricity, and heating fuels. In our survey, we ask respondents for information about their income, overall expenditure, and how they split that expenditure over these and other categories. We then ask how their expenditure on each of these categories would change if they had more to spend.

Our first step is to measure the cross-sectional relationship between household taxable income and consumption of carbon-intensive goods. While higher-income households consume a greater (absolute) quantity of these goods, they allocate a smaller share of their total consumption to carbon-intensive goods, reflecting an increasing but concave relationship between income and carbon-intensive good consumption.

Next, we measure respondents' causal consumption response to a small and permanent increase in household income across the income distribution. Surprisingly, the results indicate that a pattern of income effects in carbon intensity closely matches the cross-sectional variation in consumption in the Consumer Expenditure Survey. We note that this is not at all mechanical or obvious. Indeed, these causal consumption responses are within-individual effects, which are compared to variation across individuals. The result suggests that cross-sectional variation in consumption of these carbon-intensive goods is largely explained by income effects, and not by differences in preferences. 

The interpretation of these survey responses hinges on the reliability of stated preferences. Growing evidence suggests that stated preferences perform well in general \citep{list-survey,fusterzafar}. Yet, reliability inevitably depends on the context and methodology. In our setting, we have a way of verifying the reliability of our survey results. Because we measure status quo allocations across different types of goods in both our survey and in the diary-based Consumer Expenditure Survey, we are able to verify that at least the cross-sectional relationship between income and carbon-intensive good consumption is very similar across both sources. This lends credibility to our survey instrument and sample. We also note that we calibrate the hypothetical income increase to be familiar and relevant for respondents---comparable in size to a small permanent cut in average tax rates or typical lump-sum stimulus checks. 

Our empirical results allow us to estimate the optimal tax on carbon-intensive goods. On average, differences in preferences justify only small deviations from the standard Pigouvian prescription of setting the tax equal to marginal damage. Our estimates suggest that the Pigouvian level of the tax is about 40 cents per dollar of expenditure (assuming a Social Cost of Carbon of \$200 per ton). Taking into account differences in preferences, we find that the optimal linear tax would be just 0.2 cents lower than this. The very small magnitude of this adjustment is robust to a wide range of potential parameterizations, reflecting our core empirical result that the causal and cross-sectional relationships between income and carbon intensity match closely on average.

To the extent that taxes on externality-generating goods can be made non-linear---which has some precedents---there are somewhat larger deviations from the Pigouvian prescription. For a non-linear tax, we compare the cross-sectional and causal relationships between income and consumption choices at each point in the income distribution. For example, consider households with incomes around \$100,000. We find that the estimated causal consumption response is slightly larger than the slope of the cross-sectional relationship between income and carbon intensity. This implies that those earning around \$100,000 have a decreasing preference for carbon-intensive goods as we move up the income distribution. In turn, this justifies a carbon tax that is around 1.5 cents per dollar of expenditure lower than the Pigouvian level. By contrast, households with incomes closer to \$250,000 would face carbon taxes of 43 cents, which is 7.5 percent higher than the Pigouvian level.

Our final step is to extend the model to allow for multidimensional heterogeneity in the sense that carbon intensity is allowed to vary across households with the same level of income. The multidimensional case is more complex because adjustments to the income tax can no longer exactly compensate all households for changes to the carbon tax. Furthermore, individuals with the same taxable income may differ in their responsiveness to taxation and their preferences for the carbon-intensive good. We show how to derive an empirically tractable optimal tax condition, and show that the lessons from the unidimensional case carry over qualitatively. A particularly important result is that variation in the causal effect of higher income on consumption choices tends to dampen any deviation of the optimal tax from its Pigouvian level. For this reason, we find that the optimal carbon tax is pushed even closer to the Pigouvian benchmark.

Our paper provides an empirical and theoretical framework that could be applied more broadly to environmental taxes, or any other tax designed to correct an externality. Consistent with the principle of additivity \citep{sandmo1975optimal,Kopczuk2003}, the total tax on any externality-generating activity is a combination of a ``Pigouvian'' correction and a distributional component. We provide a way to measure the distributional component and show how it depends on causal and cross-sectional sufficient statistics. The optimal tax on the externality may deviate positively or negatively from the Pigouvian prescription. In some cases, the total tax could even flip sign relative to the externality. In our case, we find that distributional concerns have only a minor impact on the optimal carbon tax. We show in Appendix sections \ref{app:empirical_pareto_efficient_tax_gasoline_only} and \ref{app:empirical_pareto_efficient_tax_electricity_only} that this is also true for the individual components of our carbon-intensive aggregate good: electricity, natural gas and heating fuels, and gasoline. However, applying our framework to other taxes may yield different results.

\smallskip
\noindent \textit{Connections in the Literature}\nopagebreak[4]

\vspace{-0.4em}
Externalities have long been recognized as providing a motivation for government intervention. Corrective taxation as a solution dates back to \citet{Pigou1920}. Economists have studied wide varieties of externalities in practice, including environmental taxes \citep{weitzman1974prices,andersson2019carbon}, occupation choice \citep{Lockwood2017}, congestion \citep{walters1961congestion}, labor market distortions \citep{CraigJMP}, control of a pandemic \citep{CraigHines2020NTJ}, rent-seeking \citep{RothschildScheuer2016}, and positional externalities \citep{AronssonJohanssonStenman2008}.

A subset of this literature incorporates distributional concerns. \citet{Kopczuk2003} shows how the optimal tax on an externality can be written as the sum of an externality correction component and a separate term capturing other objectives, building on \citet{sandmo1975optimal}. \citet{Kaplow2012} studies a special case without preference heterogeneity, showing that the optimal tax is equal to marginal damage in this case. \citet{PirttilaTuomala1997} allow for interactions with labor supply.  \citet{JacobsVanDerPloeg2019JEEM} focus on linear taxation. \citet{CremerGahvariLadoux1998}
and \citet{CremerGahvariLadoux2003} consider optimal taxation and externality correction with a small number of types, demonstrating the importance of homogeneity and separability in preferences, highlighting the potential for the optimal tax on the externality-generating good to have the opposite sign from the externality. \citet{AhlvikLiskiMakimattila2024} use a mechanism design approach to study the interaction of income taxation and externality correction, focusing on different types of equity considerations and differences in cost burdens. They empirically estimate optimal deviations from Pigouvian taxation, but with much stronger restrictions on individual heterogeneity than ours. \citet{JacobsdeMooij2015} ask whether the Pigouvian correction should be adjusted for the marginal cost of public funds. \citet{micheletto2008ext} studies the correction of externalities that are not ``atmospheric'' so that externality correction is imperfectly targeted.

Methodologically, the most closely related work is by \citet{allcott2019regressive} and \citet{ferey2024sufficient}. These papers focus primarily on ``internality'' correction, and savings taxation when the government values equality and taxes are distortionary. The authors highlight the core insight that causal income effects can be subtracted from cross-sectional variation in consumption patterns to establish the degree to which preferences vary across income levels. Our contributions are to: (i) apply these ideas to shed light on how externality-correcting taxes should change with preference heterogeneity, including multidimensional heterogeneity; and (ii) study carbon emissions by households, which is one of the most important examples of an externality. These papers and ours build on previous work \citep{Saez2002,PikettySaez2013,DiamondSpinnewijn2011,GauthierHenriet2018}; and connect to concurrent theoretical analysis by \citet{DoligalskiEtAl2025}.

Empirically, recent studies highlight significant heterogeneity in consumer preferences for carbon-intensive goods vs. cleaner alternatives. For example, \cite{lyubich2024role} decomposes the unobserved heterogeneity in US household carbon consumption and finds, by and large, a greater role for person effects (up to 50\% of the variance decomposition) relative to place effects (up to 23\%). Similarly, \cite{dorsey2024unequal} provide evidence of across-income preference heterogeneity in solar panel uptake.  

Our focus is on heterogeneity between households and the importance of the distortionary costs of redistributive taxation. We do not incorporate general equilibrium effects or dynamics, which can provide important additional motives for deviating from standard Pigouvian taxation. This has been considered in other work \citep{CremerGahvari2001JPuE,CremerGahvariLadoux2010,douenne2025,Barrage2020,Bierbrauer2024,FriedNovanPeterman2024}. Similarly, we do not consider multinational issues \citep[see][]{AronssonBlomquist2003} or explicitly model multiple generations \citep{FriedNovanPeterman2018}. Our work complements \citet{FrickeFuestSachs2025}, who study the optimal rebate when carbon emissions are capped.

\smallskip
\noindent \textit{Roadmap for the Paper}\nopagebreak[4]

\vspace{-0.4em}
In Section \ref{sec:model}, we introduce a model of optimal taxation with flexible tax instruments, an externality-generating good, and heterogeneity across but not within levels of income. We show how the optimal taxes on income and the externality and income interact, and derive empirically estimable sufficient statistics to characterize optimal taxation. Section \ref{sec:multi} extends our results to multidimensional heterogeneity. Section \ref{sec:empirics} then applies these results to carbon taxation in the United States.  Section \ref{sec:conclusion} concludes.

\vspace{-1.5em}

\section{Externality Correction With Taste Heterogeneity}\label{sec:model}

\vspace{-0.5em}

Consider a population of individuals indexed by type $w\in W$, where $w$ has differentiable cumulative distribution $F$ and $W$ is compact. Individual $w$ chooses taxable income $z\in\left[0,\overline{z}\right]{}$. Income has differentiable cumulative distribution $H_z(z)$ and density $h_z(z)$.  Income is subject to nonlinear income taxation with tax schedule ${T}_{z}\left(z\right)$. She allocates her after-tax income, $z-{T}_z(z)$ between a numeraire consumption good, $c$, and an externality-generating good, $x$. A tax, ${T}_x(x)$ must be paid on consumption of the externality-generating good. We constrain ${T}_{z}\left(z\right)$ and ${T}_x(x)$ to be twice-differentiable.

We focus in this section on the case in which heterogeneity is unidimensional in the sense that $w\in \mathbb{R}$ and utility is $u(c,x,z\mid w)$. This implies that there is no variation in preferences conditional on $w$. However, individuals with different $w$ may have different preferences over the two consumption goods. Specifically, if two individuals with different $w$ were forced to have the same $z$ and $z-T(z)$ they may prefer different consumption bundles. In Section \ref{sec:multi}, we extend our results to allow for multidimensional heterogeneity.

\textbf{Separable Benchmark.} As an important benchmark, we will consider the case, as in \citet{AtkinsonStiglitz1976}, where preferences over consumption goods are weakly separable from labor supply, and homogeneous across agents, including those of different types. In this special case, it is well-known that we obtain the standard
Pigouvian result that it is optimal to tax the externality-generating
good at a constant rate equal to the marginal external damage
associated with consuming the good \citep{Kaplow2012}.

\vspace{-1.4em}
\subsection{Individual Maximization}
\vspace{-0.4em}

Each agent maximizes her utility subject to her budget constraint and the tax system set by the social planner, producing indirect utility
$v(w)$. 

\vspace{-2.5em}
\begin{align}
v(w) & =\max_{z,x,c}u(c,x,z\mid w)\label{eq:agentProb}\\
 & \text{s.t. }\quad c+x+{T}_{x}(x)=z-{T}_{z}(z)\nonumber 
\end{align}

\vspace{-1em}

\noindent We assume that $u(c,x,z\mid w)$ is twice continuously differentiable
with bounded first derivatives, increasing in $c$ and $x$, and decreasing
in $z$. We further assume that $u(c,x,z\mid w)$ is strictly concave
in $c$, $x$, and $z$.

We will often find it useful to approach \eqref{eq:agentProb} as
a two stage maximization problem. First, a type $w$ chooses their
income, $z(w)$. Second, they decide how to split the result disposable
income between the two consumption goods, leading to choices $c(w)$
and $x(w)$. Specifically, a type $w$ agent chooses taxable income:
\vspace{-0.4em}\begin{equation}\vspace{-0.4em}
    z\left(w\right)\equiv\arg\max_{z}\left\{ u\left(z-T_{z}\left(z\right)-x\left(w;z\right)-T_{x}\left(x\left(w;z\right)\right),x\left(w;z\right),z\mid w\right)\right\} 
\end{equation}
where: 
\vspace{-0.4em}\begin{equation}\vspace{-0.4em}
x\left(w;z\right)\equiv\arg\max_{x}\left\{ u\left(z-{T}_{z}\left(z\right)-x-{T}_{x}\left(x\right),x,z\mid w\right)\right\} 
\end{equation}
is the amount of the externality-generating good a type $w$ agent
would choose if his or her taxable income were $z$. The amount of
the externality-generating good that a type $w$ agent actually chooses
is $x\left(w\right)\equiv x\left(w;z\left(w\right)\right)$. 

This approach allows us to write first-order conditions characterizing a type $w$ agent's choices as a function of both tax schedules:
\vspace{-0.4em}\begin{equation}\vspace{-0.6em}
1-\mathcal{T}_{z}'\left(z\right)-\frac{\partial x\left(w;z\right)}{\partial z}\left(1+\mathcal{T}_{x}'\left(x\right)\right)=-\frac{u_{z}+\frac{\partial x\left(w;z\right)}{\partial z}u_{x}}{u_{c}}\label{eq:agentFOC1}
\end{equation}

\vspace{-1.5em}\begin{equation}\vspace{-0.2em}
\text{and}\qquad 1+\mathcal{T}_{x}'\left(x\right)=-\frac{u_{x}}{u_{c}}.\label{eq:agentFOC2}
\end{equation}
For the sake of brevity we have suppressed the dependence of the marginal utility (e.g. $u_c$) terms on type ($w$) and/or on income ($z$). These two conditions can be combined to yield a reduced first order condition for the optimal choice of income: $1-\mathcal{T}_{z}'\left(z\right)=-\frac{u_{z}}{u_{c}}$, but this masks important intuition as we highlight below.

Equation \ref{eq:agentFOC1} highlights something important: an agent's anticipation of their own behavior at stage two enters their first stage decision: For each additional dollar of income they consider earning, they evaluate how it would change their disposable income and how much they would spend on good $x$, given their preferences and the tax rate. In this way,
the agent can be thought of as choosing their level of taxable income
subject to an \emph{effective net of income tax rate }of $1-\mathcal{T}_{z}'\left(z\right)-\frac{\partial x\left(w;z\right)}{\partial z}\left(1+\mathcal{T}_{x}'\left(x\right)\right)$.
Thus, when a rational agent sees their purchasing
power decline due to commodity taxes, this reduces their marginal benefit of earning income. The size of this reduction depends on what fraction of their marginal dollar of disposable income will be spent on the dirty good. Critically for the analysis below, this marginal spending behavior may differ across people.

\vspace{-1.2em}
\subsection{Social Welfare Maximization}
\vspace{-0.3em}

A welfarist social planner chooses twice-differentiable tax functions,
$\mathcal{T}_{z}$ and $\mathcal{T}_{x}$. It does so to jointly to
maximize social welfare, $\bm{\mathcal{W}}$, while raising enough
revenue to cover an exogenous revenue requirement, $\mathcal{R}$.
\begin{align}
\max\bm{\mathcal{W}} & =\int\gamma(w)v(w)dF-\mathcal{D}(\overline{x})\\
 & \text{s.t. }\quad\mathcal{R}=\int\left({T}_{z}(z(w))+{T}_{x}(x(w))\right)dF(w)\\
 & \text{and }\quad\overline{x}=\int x(w)dF\nonumber 
\end{align}
Here, $\gamma(w)$ are Pareto weights; and $\mathcal{D}(\overline{x})$
is a separable atmospheric externality, which is an increasing function
of aggregate consumption of the externality-generating good.

This formulation is equivalent to having the harm of aggregate consumption
of the externality-generating good entering individuals' utility
in a separable way. To see the equivalence between the two ways of
defining social welfare, consider $\Tilde{u}\equiv u(c,x,z\mid w)-\widetilde{x}(\bar{x};w)$,
then $\Tilde{v}(w)=v(w)-\widetilde{x}(\bar{x};w)$ and $\mathcal{D}'(\overline{x})=\int\gamma(w)\widetilde{x}'(\bar{x};w)dF$.
Our formulas are also unaffected if we define welfare using a social
welfare function, $\mathcal{W}\equiv\int W\left(v\left(w\right)\right)dF -\mathcal{D}(\overline{x})$,
rather than Pareto weights. However, the distinction between the two
approaches would matter for non-marginal tax changes.

\vspace{-1.2em}
\subsection{Regularity Assumptions}
\vspace{-0.3em}

For all tax schedules that we consider, we assume that the optimal
choices of $c$, $x$, and $z$ are globally unique for every individual,
with the relevant second-order conditions holding strictly. We also
assume there is a smooth mapping from $w$ to $c(w)$, $x(w)$, and
$z(w)$. While these assumptions are standard and are required to
ensure that choices move smoothly as the tax system changes, we note
that they imply restrictions on the set of tax schedules under consideration.
Finally, we assume that the mapping from $w$ to $z$ is strictly
monotonic, as is the mapping from $z$ to $x$. 

\vspace{-1.2em}
\subsection{Characterizing Tax Reforms}
\vspace{-0.3em}

\label{subsec:optnonlin}

We characterize the welfare effects of reforming the tax system using
a formalization of the perturbation method of \citet{saez2001using}
via Gateaux derivatives, an approach discussed more thoroughly in
\citet{golosov2014variational}, \citet{jacquet2021optimal}, and
\citet{gerritsen2024optimal}. Given an initial tax system with the
tax schedules, $\mathcal{T}_{x}(x)$ and $\mathcal{T}_{z}(z)$, we ask how welfare and revenue change when we perturb them in the
direction of $\tau_{x}(x)$ and $\tau_{z}(z)$ respectively. Letting
$\kappa\in\mathbb{R}$ index the magnitude of the perturbation, we
can define the reformed tax schedules as follows: 
\begin{equation}
T_{z}\left(z;\kappa\right)\equiv\mathcal{T}_{z}\left(z\right)+\kappa\tau_{z}\left(z\right)\qquad\text{and}\qquad T_{x}\left(x;\kappa\right)\equiv\mathcal{T}_{x}\left(x\right)+\kappa\tau_{x}\left(x\right).\label{eq:reformed-schedules}
\end{equation}
All our results about efficient and optimal tax policy will be based
on examining the first-order effect of such perturbations on the social
planner's Lagrangian,
\[
\left.\frac{\partial\mathcal{L}}{\partial\kappa}\right|_{\kappa=0}=\left.\frac{\partial\mathcal{W}}{\partial\kappa}\right|_{\kappa=0}+\frac{1}{\lambda}\left.\frac{\partial\mathcal{R}}{\partial\kappa}\right|_{\kappa=0},
\]
starting from the initial tax system. For any initial tax system,
and proposed set of reforms, if $\left.\frac{\partial\mathcal{L}}{\partial\kappa}\right|_{\kappa=0}>0$,
then movement in that direction is welfare-improving. At the optimal
tax system, no such improvements can be possible and thus, at the
optimum, we must have:
\vspace{-0.3em}\begin{equation}
\left.\frac{\partial\mathcal{L}}{\partial\kappa}\right|_{\kappa=0}=0\quad\text{for all }\tau_{z}(\cdot)\text{ and }\tau_{x}(\cdot).\label{eq:optimalityCond}
\end{equation}

\vspace{-1.6em}
\subsection{Distribution Neutral Reforms}\label{sec:distneut}
\vspace{-0.3em}

To begin our investigation of the optimal tax system, consider a special
type of tax reform: one which is \emph{distribution neutral}. 

Let $\hat{w}\left(z\right)$ be the type of agent that chooses income
$z$. In other words, it is the inverse the taxable income choice
function, $z\left(w\right)$. This allows us to define $\hat{x}\left(z\right)\equiv x\left(\hat{w}\left(z\right);z\right)$,
the dirty good consumption of an agent with taxable income $z$. 

A distribution neutral tax reform is one where for a given change
in the dirty good tax liability of agents with taxable income $z$,
$\tau_{x}\left(\hat{x}\left(z\right)\right)$, we also change their
income tax liability by a precisely offsetting amount, such that if
the agent can still afford their original bundle. That is to say,
the income tax reform is defined by:
\vspace{-0.5em}\begin{equation}\vspace{-0.5em}
 \tau_{z}\left(z\right)=-\tau_{x}\left(\hat{x}\left(z\right)\right).   
\end{equation}
Notice, this implies marginal income tax rate changes of:
\vspace{-0.5em}\begin{equation}\vspace{-0.5em}
\tau_{z}'\left(z\right)=-\hat{x}'\left(z\right)\tau_{x}'\left(\hat{x}\left(z\right)\right),
\end{equation}
where $\hat{x}'\left(z\right)$ is the slope of the observed cross-sectional relationship
between dirty good consumption and taxable income.

Distribution neutral reforms are of special interest because of the relative simplicity of their impact on social welfare. For any such reform, there is no first-order impact on agent utility, nor is there a mechanical impact on revenue. All effects of the reform arise because behavioral responses affect the externality and government revenue. 

This also implies that for any distribution neutral reform, if
$\left.\frac{\partial\mathcal{L}}{\partial\kappa}\right|_{\kappa=0}>0$,
then a reform in that direction is a Pareto improvement: Positive
gains in government revenue net of externality costs can be obtained without
changes to individual agent utility, which allows a small transfer to be made that makes all agents strictly better off. Put another way, for any Pareto
efficient tax system, it must be the case
\vspace{-0.4em}\begin{equation}\vspace{-0.5em}
\left.\frac{\partial\mathcal{L}}{\partial\kappa}\right|_{\kappa=0}=0\quad\text{for all }\tau_{z}(\cdot)\text{ and }\tau_{x}(\cdot)\text{ satisfying }\tau_{z}\left(z\right)=-\tau_{x}\left(\hat{x}\left(z\right)\right)\label{eq:efficiencyCond}
\end{equation}

\vspace{-1.2em}
\subsection{Behavioral Responses}\label{sec:behav}
\vspace{-0.3em}

By construction, any distribution-neutral reform produces a change in relative prices
without a change in their purchasing power. Thus, only compensated behavioral responses will be relevant to the analysis. To establish what kinds of responses will occur, we can revisit the agent's first-order conditions. 

\underline{Consumption of the Dirty Good (Direct).} The reform to $T_x$ changes the stage 2 first-order condition of an agent with taxable income $z$ in the following way:
\vspace{-0.3em}\begin{equation}\vspace{-0.3em}
1+\mathcal{T}_{x}'\left(\hat{x}\left(z\right)\right)+\kappa\tau_{x}'\left(\hat{x}\left(z\right)\right)=-\frac{u_{x}}{u_{c}}.
\end{equation}
Holding constant the agent's choice of $z$, this generates a behavioral
response of
\vspace{-0.3em}\begin{equation}\vspace{-0.3em}
\frac{\partial\hat{x}(z)}{\partial\kappa}\bigg|_{z,\,\kappa=0}=-\frac{\varepsilon_{x\mid z}(z)\,\hat{x}(z)}{1+\mathcal{T}_{x}'(\hat{x}(z))}\tau_{x}'(\hat{x}(z))\label{eq:DN_dirtygood_ownprice}
\end{equation}
where:
\vspace{-0.3em}\begin{equation}
\varepsilon_{x\mid z}(z)\equiv-\frac{1+\mathcal{T}_{x}'(\hat{x}(z))}{\hat{x}(z)}\frac{\partial\hat{x}(z)}{\partial\mathcal{T}_{x}'}\bigg|_{z}^{c}
\end{equation}
is the compensated elasticity of spending on $x$ with respect to its tax-inclusive price. 

\underline{Income.} Turning our attention to the stage 1 first-order condition, we have
the following perturbed first-order condition:
\vspace{-0.4em}\begin{equation}\vspace{-0.3em}\label{eq:incfoc1}
1-\mathcal{T}_{z}'\left(z\right)+\underbrace{\kappa\hat{x}'\left(z\right)\tau_{x}'\left(\hat{x}\left(z\right)\right)}_{\text{Reform to }T_z}-x_{\text{inc}}^{\prime}(z)(1+\mathcal{T}_{x}'\left(x\right)+\underbrace{\kappa\tau_{x}'\left(\hat{x}\left(z\right)\right)}_{\text{Reform to }T_x})=-\frac{u_{z}+x_{\text{inc}}^{\prime}(z)u_{x}}{u_{c}}.
\end{equation}
where $x_{\text{inc}}^{\prime}(z)\equiv\frac{\partial x\left(\hat{w}\left(z\right);z\right)}{\partial z}$
is the causal effect of increasing the taxable income of a type $\hat{w}\left(z\right)$
agent on their choice of dirty good consumption.

Notice, in the two-stage formulation of the agent's problem, that this first-order condition depends on the reforms to both tax schedules. This is because, when responding to the reform
at stage 1, the agent anticipates how the reform to the tax on the dirty good will impact their marginal return to earning income, given that some
of their income will be spent on that good. After rearranging the left hand side of the equation, the reformed effective marginal income
tax rate can be rewritten as:
\vspace{-0.3em}\begin{equation}\vspace{-0.3em}\label{eq:incfoc2}
1-\mathcal{T}_{z}'\left(z\right)-x_{\text{inc}}^{\prime}(z)\left(1+\mathcal{T}_{x}'\left(x\right)\right)+\kappa\left(\hat{x}'\left(z\right)-x_{\text{inc}}^{\prime}(z)\right)\tau_{x}'\left(\hat{x}\left(z\right)\right).
\end{equation}
This formulation highlights that heterogeneity in preferences leads to a distinction between the causal effect of additional income on dirty good consumption, and the cross-sectional relationship between income and dirty good consumption. Following \cite{ferey2024sufficient}, we let $x_{\text{het}}^{\prime}(z)$ be the component of the cross-sectional relationship that is driven by heterogeneity in tastes:
\begin{equation}\label{eq:decomp}
x_{\text{het}}^{\prime}(z)\equiv\hat{x}'\left(z\right)-x_{\text{inc}}^{\prime}(z). 
\end{equation}
As equation \ref{eq:incfoc1} shows, the change in the effective marginal income tax rate induced by the reform reform is directly proportional to this preference heterogeneity effect. As a result, the behavioral response of taxable income induced by the reform is:
\vspace{-0.3em}\begin{equation}\vspace{-0.3em}
\frac{\partial z}{\partial\kappa}\bigg|_{\kappa=0}=\frac{z\varepsilon_{z}(z)}{1-\mathcal{T}_{z}'(z)}x_{\text{het}}^{\prime}(z)\tau_{x}'(\hat{x}(z))\label{eq:DN_inc_ownprice}
\end{equation}
where:
\vspace{-0.3em}\begin{equation}
\varepsilon_{z}(z)\equiv\frac{1-\mathcal{T}_{z}'(z)}{z}\,\frac{\partial z}{\partial\big(1-\mathcal{T}_{z}'(z)\big)}\bigg|^c_{z}
\end{equation}
is the compensated elasticity of taxable income with respect to the
net-of-income-tax rate. Note that formal definitions of these elasticities are included in the proof of Proposition \ref{prop:uni-1}.

\underline{Consumption of the Dirty Good (Indirect).} To complete our enumeration of behavioral responses, we consider the subsequent impact of the induced change in income on consumption of $x$. For an agent earning income $z$, this is determined by the
causal effect of income on the dirty good,
\vspace{-0.3em}\begin{equation}\vspace{-0.3em}
\frac{\partial\hat{x}(z)}{\partial z}\bigg|_{\kappa=0}\frac{\partial z}{\partial\kappa}\bigg|_{\kappa=0}=x_{\text{inc}}^{\prime}(z)\frac{z\varepsilon_{z}(z)}{1-\mathcal{T}_{z}'(z)}x_{\text{het}}^{\prime}(z)\tau_{x}'(\hat{x}(z)),\label{eq:DN_dirtygood_incEff}
\end{equation}
resulting in a total effect of the reform on dirty good consumption of: 
\vspace{-0.3em}\begin{equation}\vspace{-0.3em}
\frac{\text{d}x}{\text{d}\kappa}\bigg|_{\kappa=0}=\frac{\partial\hat{x}(z)}{\partial\kappa}\bigg|_{z,\,\kappa=0}+\frac{\partial x}{\partial z}\bigg|_{\kappa=0}\frac{\partial z}{\partial\kappa}\bigg|_{\kappa=0}.
\end{equation}

\subsection{A Key Sufficient Statistic: The ``Taste Elasticity''}
\vspace{-0.4em}

The formulation above highlights the key role of heterogeneity in preferences. Specifically, our optimal tax conditions will feature an important additional sufficient statistic: the ``taste elasticity'',  $\eta_{x,z}^{\text{Taste}}\left(z\right)$. This is the elasticity of consumption of $x$ with respect to income that occurs only through heterogeneity in agent tastes:
\vspace{-0.3em}\begin{equation}\vspace{-0.3em}
\eta_{x,z}^{\text{Taste}}\left(z\right)\equiv\frac{zx_{\text{het}}^{\prime}(z)}{\hat{x}\left(z\right)}.\label{eq:tasteHet}
\end{equation}
The taste elasticity captures any differences in consumption of $x$ across income that occur not due to the causal effect of  higher income, but due to differences in the preferences of agents who locate at different parts of the income distribution.

In the \citet{AtkinsonStiglitz1976} benchmark in which preferences are identical across agents and separable from labor supply, $\eta_{x,z}^{\text{Taste}}\left(z\right)=0$. In this special case, both the income and indirect consumption effects (equations \ref{eq:DN_inc_ownprice} and \ref{eq:DN_dirtygood_incEff}) are zero. The logic of the two-stage decision process above sheds light on why. Recall that a distribution neutral reform uses the income tax to exactly offset the financial impact of the altered tax on the dirty good. If all agents have the same preferences, this also means that every agent’s evaluation of every counterfactual income choice is also unchanged, because all agents would make the same consumption decisions if they had the same income level. This implies that the reform induces no change to any agent's choice of income.

More generally, differences in tastes across types can lead to a positive or negative value of $\eta_{x,z}^{\text{Taste}}\left(z\right)$, depending on whether higher or lower types have a relative preference for the dirty good. This implies potential labor supply impacts from the distribution neutral reform. To see why, recall the two stage decision process. An agent of type $w_0$ and income $z_0$ considers whether to increase her income slightly to $z_1$. If $\eta_{x,z}^{\text{Taste}}\left(z\right)<0$, she would consume less of $x$ given income $z_1$ than type $w_1$ agents (who do have income $z_1$). Now consider a reform that increases $T_x(\hat{x}(z_1))$ while lowering $T_z(z_1)$ just enough to leave the utility of type $w_1$ unchanged at $z_1$. This same reform must raise the utility that $w_0$ would receive at $z_1$, making that higher income more attractive. Conversely, if $\eta_{x,z}^{\text{Taste}}\left(z\right)<0$, the same distribution neutral reform reduces the incentive to increase one's income.

\vspace{-1.5em}
\subsection{Testing Local Pareto Efficiency: A Specific Reform}
\vspace{-0.3em}

To characterize the set of Pareto efficient tax systems, consider a specific distribution neutral tax reform where we only alter the marginal
commodity tax rate for agents earning taxable income $z$, increasing
it:
\vspace{-0.4em}\begin{equation}\vspace{-0.3em}
  \tau_{x}'(\hat{x}(z'))=\begin{cases}
1 & \text{if }z'=z\\
0 & \text{otherwise}
\end{cases}.  
\end{equation}
Distribution neutrality then requires an offsetting reduction of those same agents' marginal income tax rates. Specifically, $\tau_{z}'(z') = -\hat{x}'(z')$ if $z'=z$ and $\tau_{z}'(z') = 0$ otherwise.

\vspace{-1.6em}
\subsection{Welfare Effects of the Reform}\label{sec:welfeff}
\vspace{-0.3em}

We obtain the welfare effect of a reform by multiplying each behavioral response (equations \ref{eq:DN_dirtygood_ownprice},
\ref{eq:DN_inc_ownprice}, and \ref{eq:DN_dirtygood_incEff}) by
its marginal impact on social welfare. For changes in consumption of the dirty good (equations \ref{eq:DN_dirtygood_ownprice}
and \ref{eq:DN_dirtygood_incEff}), the marginal social value is the difference between the marginal tax rate on that good and its marginal social damage, i.e. $\mathcal{T}_{x}'(\hat{x}(z))-\frac{\mathcal{D}'(\bar{x})}{\lambda}$. For changes in income, the marginal social value
is simply the income tax rate $\mathcal{T}_{z}'(z)$. In total, the impact on social welfare is the sum of three effects:
\begin{equation}\label{eq:changewelf}
\Delta\bm{\mathcal{W}}(z) = \Delta\bm{\mathcal{W}}^{z}(z) + \Delta\bm{\mathcal{W}}^{x\mid z}(z) + \Delta\bm{\mathcal{W}}^{z\rightarrow x}(z)
\end{equation}
where the components are defined as follows.

\vspace{-2.2em}
\begin{align}
    \Delta\bm{\mathcal{W}}^{z}(z) &= \mathcal{T}_{z}'(z)\frac{z\varepsilon_{z}(z)}{1-\mathcal{T}_{z}'(z)}x_{\text{het}}^{\prime}(z)\\
    \Delta\bm{\mathcal{W}}^{x\mid z}(z) &= -\left(\mathcal{T}_{x}'(\hat{x}(z))-\frac{\mathcal{D}'(\bar{x})(z)}{\lambda}\right)\frac{\varepsilon_{x\mid z}(z)\,\hat{x}(z)}{1+\mathcal{T}_{x}'(\hat{x}(z))}\\
    \Delta\bm{\mathcal{W}}^{z\rightarrow x}(z) &= \left(\mathcal{T}_{x}'(\hat{x}(z))-\frac{\mathcal{D}'(\bar{x})}{\lambda}\right)x_{\text{inc}}^{\prime}(z)\frac{z\varepsilon_{z}(z)}{1-\mathcal{T}_{z}'(z)}x_{\text{het}}^{\prime}(z)
\end{align}
\vspace{-1.8em}

The first of these effects, $\Delta\bm{\mathcal{W}}^{z}$, captures the impact on government revenue from taxable income changes. The second, $\Delta\bm{\mathcal{W}}^{x\mid z}$, is the uninternalized component of the externality caused by changes in consumption of the externality-generating good, which are directly caused by the tax reform. Finally, $\Delta\bm{\mathcal{W}}^{z\rightarrow x}$ is the effect of further changes in consumption of $x$ that are induced indirectly by changes in taxable income.

\vspace{-1.5em}
\subsection{A Neccessary Condition for Pareto Efficiency}
\vspace{-0.3em}
Setting the total impact on welfare ($\Delta\bm{\mathcal{W}}$) equal to zero yields an optimal tax condition at each income level $z$, which simply requires that there be no net gain to distribution neutral tax reforms at any given level of income.
   \vspace{-0.3em}\begin{equation}\vspace{-0.3em}
\frac{\mathcal{T}_{x}^{\prime}\left(\hat{x}\left(z\right)\right)-\frac{\mathcal{D}^\prime(\bar{x})}{\lambda}}{1+\mathcal{T}_{x}^{\prime}\left(\hat{x}\left(z\right)\right)}=\eta_{x,z}^{\text{Taste}}\left(z\right)\frac{\varepsilon_{z}\left(z\right)}{\varepsilon_{x|z}\left(z\right)}\left(\frac{\mathcal{T}_{z}^{\prime}\left(z\right)+\left[\mathcal{T}_{x}^{\prime}\left(\hat{x}\left(z\right)\right)-\frac{\mathcal{D}^\prime(\bar{x})}{\lambda}\right]x^\prime_{\text{inc}}(z)}{1-\mathcal{T}_{z}^{\prime}\left(z\right)}\right)\label{eq:ParetoTaxSN}
\end{equation}

In our proof in Appendix \ref{app:proofs}, we derive equation \ref{eq:ParetoTaxSN} more formally, considering the broader set of distribution neutral reforms. Specifically, the result follows
from applying the fundamental lemma of the calculus of variations.

Our final step is to rearrange equation \ref{eq:ParetoTaxSN}
to isolate the wedge between the marginal tax rate on $x$ and its marginal damage, i.e. $\mathcal{T}_{x}'(\hat{x}(z))-\frac{\mathcal{D}'(\bar{x})}{\lambda}$. This yields Proposition \ref{prop:uni-1}. 
\begin{prop}
\label{prop:uni-1} If the income and commodity taxes are optimal
and are allowed to be non-linear, they must satisfy equation \ref{eq:ParetoTaxSN_r}:
\vspace{-0.3em}\begin{equation}\vspace{-0.3em}
\mathcal{T}_{x}^{\prime}\left(\hat{x}\left(z\right)\right)-\frac{\mathcal{D}^{\prime}(\bar{x})}{\lambda}=\frac{\text{RR}(z)\mathcal{T}_{z}^{\prime}\left(z\right)}{1-\text{RR}(z)x_{\text{inc}}^{\prime}(z)}\label{eq:ParetoTaxSN_r}
\end{equation}
where:
\vspace{-0.3em}\begin{equation}
\text{RR}(z)\equiv\eta_{x,z}^{\text{Taste}}\left(z\right)\frac{\varepsilon_{z}\left(z\right)}{\varepsilon_{x|z}\left(z\right)}\frac{1+\mathcal{T}_{x}^{\prime}\left(\hat{x}\left(z\right)\right)}{1-\mathcal{T}_{z}^{\prime}\left(z\right)}\label{eq:relresponse}
\end{equation}
measures relative responsiveness to changes in the carbon versus the
income tax rates. If this condition is not satisfied, there exists
a simultaneous reform to both taxes that yields a Pareto gain.
\end{prop}

From Equation \ref{eq:ParetoTaxSN_r}, we can see that $\eta_{x,z}^{\text{Taste}}\left(z\right)$ is a sufficient statistic that determines the sign and scales the size of any deviation from setting the marginal tax on $x$ equal to the marginal damage, $\mathcal{D}^\prime(\overline{x})$. If there is no taste heterogeneity, i.e., $\eta_{x,z}^{\text{Taste}}\left(z\right)=0$, the optimal corrective tax is the standard Pigouvian one: $\mathcal{T}_{x}^{\prime}\left(\hat{x}\left(z\right)\right)=\mathcal{D}^\prime(\bar{x})/\lambda$. In this case, there is no ``tagging'' role for the carbon tax, because tastes do not change along the income distribution. Only if $\eta_{x,z}^{\text{Taste}}\left(z\right)\neq 0$ does the optimal carbon tax differ from the Pigouvian benchmark. We can also see that such deviations may vary across different levels of income, which could motivate means-tested rebates or subsidies for externality-producing goods.

Unlike the standard Pigouvian tax, the optimal tax depends on the magnitude of behavioral responses. The stronger the demand response (i.e.
the larger is $\varepsilon_{x|z}\left(z\right)$), the closer the optimal tax will be to the Pigouvian benchmark. This happens because the planner's ability to successfully exploit
tagging and use deviations from the Pigouvian level of the tax to redistribute depends on how responsive consumption of the tag is to taxation. Conversely, a higher value of the elasticity if income favors a larger deviation from Pigouvian taxation:  If the income tax is highly distortionary, using other instruments such as the carbon tax can be an efficient alternative. For a similar reason, we see that the deviation from Pigouvian taxation is proportional to
the level of the marginal tax rate on income: A low income tax rate implies a small marginal distortion from the income tax in
the first place.

The optimal tax could even have the opposite sign of the marginal damage, at least in theory. In this case, a good causing a negative externality is optimally subsidized, or a good causing a positive externality is taxed. This occurs 
if $\eta_{x,z}^{\text{Taste}}\left(z\right)$ is negative and the magnitude of the optimal deviation is large. That could happen because consumption of $x$ is
a very effective tag (i.e. $\eta_{x,z}^{\text{Taste}}\left(z\right)$
is itself large), because the response to the tax is small ($\varepsilon_{x|z}\left(z\right)$
is small), the Pareto parameter of taxable income ($1-H_{z}\left(z\right)/\left[zh_{z}\left(z\right)\right]$) at $z$ is small,
and/or the social marginal utility is small at incomes above $z$.

The denominator of equation \ref{eq:ParetoTaxSN_r} amplifies positive
deviations from the Pigouvian level, and dampens negative deviations.
Suppose that the taste elasticity is positive, motivating a carbon
tax greater than the Pigouvian level. Raising the carbon tax and lowering
the income tax produces efficiency gains. Incomes increase. If $x_{\text{inc}}^{\prime}(z)>0$,
consumption of the carbon-intensive good increases as well. There
is an uninternalized social \textit{benefit} to that to the extent
that $\mathcal{T}_{x}^{\prime}\left(\hat{x}\left(z\right)\right)>\frac{\mathcal{D}^{\prime}(\bar{x})}{\lambda}$,
amplifying the original gains further. Conversely, suppose that the
taste elasticity is negative. Then we realize efficiency gains as
we lower the carbon tax below the Pigouvian level, the resulting rise
in incomes and carbon-intensive consumption produces an uninternalized
social \textit{cost} because $\mathcal{T}_{x}^{\prime}\left(\hat{x}\left(z\right)\right)<\frac{\mathcal{D}^{\prime}(\bar{x})}{\lambda}$.
This makes the efficiency gains from lowering the carbon tax less
attractive than otherwise.

A different way of seeing this is that using the carbon tax for redistribution
is more beneficial if the income tax is more distortionary. The net
wedge from the income tax is: 
\vspace{-0.3em}\begin{equation}\vspace{-0.3em}
\mathcal{T}_{z}^{\prime}\left(z\right)+\left[\mathcal{T}_{x}^{\prime}\left(\hat{x}\left(z\right)\right)-\frac{\mathcal{D}^{\prime}(\bar{x})}{\lambda}\right]x_{\text{inc}}^{\prime}(z).
\end{equation}
It is clear from this expression that raising $\mathcal{T}_{x}^{\prime}$
increases that wedge, making the income tax more distortionary, and
thus relying on the carbon tax for redistribution more beneficial.
Conversely, progressively lowering $\mathcal{T}_{x}^{\prime}$ reduces
the wedge from income taxation. Mechanically, this means that positive
deviations from the Pigouvian level of the carbon tax are amplified,
and the negative deviations are dampened.

\vspace{-1.5em}
\subsection{Level of the Optimal Taxes}
\vspace{-0.3em}

A very similar approach can be used to derive conditions characterizing the optimal levels of the externality and income taxes. The fact that a reform to each tax schedule is not by itself distribution neutral means that they involve trade-offs between individuals, and therefore marginal social welfare weights. We follow \citet{Diamond1975} and \citet{ferey2024sufficient} in defining augmented welfare weights:
\begin{align}
    g(w) &= \gamma(w) \frac{u_c(c(w),x(w),z(w)\mid w)}{\lambda} + \left[\mathcal{T}_{x}^{\prime}\left(x\left(w\right)\right)-\frac{\mathcal{D}^{\prime}\left(\bar{x}\right)}{\lambda}\right]\left.\frac{\partial \hat{x}(w)}{\partial I}\right|_{z=z(w)}\notag
    \\ &+ \left\{\mathcal{T}_{z}^{\prime}(z(w))+\left[\mathcal{T}_{x(w)}^{\prime}(x\left(w\right))-\frac{\mathcal{D}^{\prime}\left(\bar{x}\right)}{\lambda}\right]x^\prime_{\text{inc}}(z(w))\right\}
    \frac{\partial z(w)}{\partial I},
\end{align}
where $\frac{\partial {x}(w)}{\partial I}$ and $\frac{\partial z(w)}{\partial I}$ are the effects of additional disposable income on consumption of $x$ and income; and $\lambda$ is the multiplier on the government's budget constraint. These augmented weights incorporate both the direct social value of increasing the consumption of an individual of type $w$, and indirect impacts on social welfare from income effects.

The resulting optimal tax conditions for the two schedules are stated in Proposition~\ref{prop:optlevels}. We defer the derivation of these conditions to Appendix \ref{app:proofs}.

\begin{prop}\label{prop:optlevels}
    The optimal tax on the externality-producing good must satisfy Equation \ref{eq:optTax1}:
\vspace{-0.3em}\begin{equation}\vspace{-0.3em}\label{eq:optTax1}
\frac{\mathcal{T}_{x}^{\prime}\left(\hat{x}\left(z\right)\right)-\frac{\mathcal{D}^\prime(\bar{x})}{\lambda}}{1+\mathcal{T}_{x}^{\prime}\left(\hat{x}\left(z\right)\right)}=\eta_{x,z}^{\text{Taste}}\left(z\right)\frac{1}{\varepsilon_{x|z}\left(z\right)}\frac{1-H_{z}\left(z\right)}{zh_{z}\left(z\right)}\left(1-\bar{g}_{+}\left(z\right)\right),
\end{equation}
and the optimal tax on income must satisfy Equation \ref{eq:optTaxZ}:
\vspace{-0.3em}\begin{equation}\vspace{-0.3em}
\frac{\mathcal{T}_{z}^{\prime}\left(z\right)+\left[\mathcal{T}_{x}^{\prime}\left(x\left(\hat{w}\left(z\right)\right)\right)-\frac{\mathcal{D}^{\prime}\left(\bar{x}\right)}{\lambda}\right]x^\prime_{\text{inc}}(z)}{1-\mathcal{T}_{z}^{\prime}\left(z\right)}=\frac{1}{\varepsilon_{z}\left(z\right)}\frac{1-H_{z}\left(z\right)}{zh_{z}\left(z\right)}\left(1-\bar{g}_{+}\left(z\right)\right), \label{eq:optTaxZ}
\end{equation}
where $\bar{g}_{+}\left(z\right)$ is the average marginal social welfare weight of agents with income above $z$.
\end{prop}

\vspace{-0.3em}
Equation \ref{eq:optTax1} closely resembles the standard ``ABC'' expression for the optimal non-linear income tax for a redistributive social planner \citep{Mirrlees1971}. Part of the intuition is therefore familiar. A higher marginal tax rate at $\hat{x}\left(z\right)$ is optimal if: redistribution to agents with incomes below $z$ is highly valued (i.e. if the average welfare weight above $z$, $\bar{g}_{+}\left(z\right)$, is relatively small); people do not respond much to the carbon tax (i.e. $\varepsilon_{x\mid z}\left(z\right)$ is small); and/or a large amount of revenue can be raised relative to the amount of income distorted by a local increase in the marginal tax rate ($1-H_z(z)$ large relative to $zh_z(z)$).

There are two key differences from the standard optimal tax condition. The first is that there is a Pigouvian correction, which additively adjusts the marginal tax to take into account the marginal damage on others of consuming $x$, $\mathcal{D}^\prime(\bar{x})$. The second is the presence of the taste elasticity, $\eta_{x,z}^{\text{Taste}}\left(z\right)$, which captures the same intuition as above.

Equation \ref{eq:optTaxZ} is almost identical to the usual ``ABC'' formula. A higher marginal tax rate is optimal if redistribution is valued ($\bar{g}_{+}\left(z\right)$ is large), people do not respond much to the income tax ($\varepsilon_{z}\left(z\right)$ small), or a large amount of revenue can be raised relative to the amount of income distorted by a local increase in the marginal tax rate ($1-H_z(z)$ large relative to $zh_z(z)$). The only difference from the standard intuition is the presence of $\left[\mathcal{T}_{x}^{\prime}\left(x\left(\hat{w}\left(z\right)\right)\right)-\mathcal{D}^{\prime}\left(\bar{x}\right)/\lambda\right]x^\prime_{\text{inc}}(z)$ on the left hand side. As before, this term captures the intuition that marginal expenditure on the taxed good, $x$ is socially beneficial if the carbon tax is set higher than the Pigouvian level; or socially costly otherwise.

\vspace{-1.5em}
\subsection{Restricting to Linear Taxation}
\vspace{-0.4em}
 
Some version of a non-linear tax on carbon may well be feasible in practice. Indeed,
there are many real-life examples of non-linear pricing and taxation
schemes for electricity, road usage, and vehicle registration fees.
However, this may not be true in all settings or for all externality-generating
goods. We therefore also consider the optimal linear tax on the externality-generating
good, and the extent to which this deviates from the Pigouvian prescription
of setting it to the level of marginal damage.

If the tax on the externality-generating good must be linear ($\mathcal{T}_{x}^{\prime}\left(\hat{x}\left(z\right)\right)=t_{x}$), the Pareto efficiency condition for the linear tax rate is
similar to the condition for the non-linear corrective tax (Equation
\ref{eq:ParetoTaxSN_r}). However, it averages over all levels of
income $z$ and consumption of $x$. Specifically, the optimal linear tax
is characterized by Proposition 2.
\begin{prop}
\label{prop:lin} If the income and commodity taxes are optimal and
the corrective tax is restricted to being linear, they must satisfy
equation \ref{eq:ParetoTaxSL}: \vspace{-0.2em}
\begin{equation}
\vspace{-0.2em}t_{x}-\frac{\mathcal{D}^{\prime}(\bar{x})}{\lambda}=\frac{\mathbb{E}_{z}\left[\text{RR}\left(z\right)\mathcal{T}_{z}^{\prime}\left(z\right)\right]}{1-\mathbb{E}_{z}\left[\text{RR}\left(z\right)x_{\text{inc}}^{\prime}(z)\right]}\label{eq:ParetoTaxSL}
\end{equation}
where $\mathbb{E}_{z}$ is the expectation over the income distribution,
and: \vspace{-0.2em}
\begin{equation}\vspace{-0.3em}
\vspace{-0.2em}\text{RR}\left(z\right)\equiv\eta_{x,z}^{\text{Taste}}\left(z\right)\frac{\varepsilon_{z}\left(z\right)\hat{x}\left(z\right)}{E_{z}\left[\varepsilon_{x|z}\left(z\right)\hat{x}\left(z\right)\right]}\left(\frac{1+t_{x}}{1-\mathcal{T}_{z}^{\prime}\left(z\right)}\right).
\end{equation}
and $x_{\text{inc}}^{\prime}(z)=\frac{\partial x\left(\hat{w}\left(z\right);z\right)}{\partial z}$
is the causal effect of income on consumption of $x$. If this condition
is not satisfied, there exists a simultaneous reform to both taxes
that yields a Pareto gain. 
\end{prop}

\vspace{-0.3em}
The derivation is straightforward. For a linear dirty goods tax, the
reforms we can consider are restricted to $\tau_{x}\left(\hat{x}\left(z\right)\right)\equiv\hat{x}\left(z\right)$.
Plugging this into the equation for the welfare impact (equation \ref{eq:changewelf})
and setting it to zero gives us equation \ref{eq:ParetoTaxSL}.

\vspace{-1.5em}
\section{Multidimensional Heterogeneity}
\vspace{-0.5em}

\label{sec:multi}

We now extend our analysis to allow for multidimensional heterogeneity.
In Section \ref{sec:model}, we assumed that heterogeneity was unidimensional
and that all agents with the same taxable income $z$ would make identical
consumption choices. In practice, households with the same income
may differ in their preferences for carbon-intensive goods due to
factors such as geographic location, household composition, or idiosyncratic
tastes. Here, we allow for such within-income heterogeneity.

Formally, we assume each agent has some type $\left(w,\theta\right)\in\mathbb{R}_{++}\times\Theta$
and that type $\left(w,\theta\right)$ agents make choices to maximize
some utility function $u\left(c,x,z;w,\theta\right)$. To simplify
the discussion, we restrict our attention here to the case of a linear
tax.

\vspace{-1.2em}
\subsection{Vertically Neutral Reforms}
\vspace{-0.3em}

In the multidimensional case, it is not generally possible to construct
a Pareto gain via a joint perturbation to the income and corrective
taxes, because some individuals at any given level of income may be
disproportionately intense consumers of the externality-generating
good. This type of within-income heterogeneity has been studied empirically by \citet{pizer2019distributional}, \citet{CroninFullertonSexton2019},
and \citet{RauschMetcalfReilly2011}. \citet{sallee2022} discusses
the potential for true Pareto gains in light of this heterogeneity.

It is still possible to construct a reform in the same vein
as Section \ref{sec:model}, which ensures that there is no mechanical
redistribution between people with different levels of income. This
allows us to derive a similar necessary condition for the tax system
to be optimal. Formally, we consider a joint, infinitesimal
reform that is distribution neutral on average at each income
level $z$: what we call a \emph{vertically neutral} tax reform. Any pair of reform functions $\tau_{x}$ and $\tau_{z}$
is vertically neutral if, for each income level $z$, we have:
\vspace{-0.3em}\begin{equation}\vspace{-0.3em}
\tau_{z}(z)=-\mathbb{E}[\tau_{x}(x(w,\theta))|z(w,\theta)=z].
\end{equation}
In the linear tax case, this is simply $\tau_{z}(z)=-\bar{x}\left(z\right)$, where $\bar{x}\left(z\right)\equiv\mathbb{E}[x(w,\theta)|z(w,\theta)=z]$.
The implied reform to marginal income tax rates is then:

\vspace{-2.3em}
\begin{align*}
\tau_{z}'(z) & =-\bar{x}'\left(z\right),\\
 & =-\frac{\text{d}\mathbb{E}[x(w,\theta)|z(w,\theta)=z]}{\text{d}z}.
\end{align*}

\vspace{-2em}
\subsection{Effects of Vertically Neutral Reforms}
\vspace{-0.3em}

The effects of a vertically neutral reform with multidimensional heterogeneity largely parallel the unidimensional case. However, there are new effects that stem from hetereogenity in the impact of a reform on different agents with the same income level.

Let $\hat{x}\left(z,\theta\right)\equiv x\left(\hat{w}\left(z;\theta\right),\theta\right)$
be the dirty good consumption of a type $\theta$ agent that has chosen
taxable income $z$, with $\hat{w}\left(z;\theta\right)$ being the
inverse of $z\left(w,\theta\right)$, holding $\theta$ constant.
Alternatively, we can express the consumption level of this agent as:
\vspace{-0.3em}\begin{equation}\vspace{-0.3em}
\hat{x}\left(z,\theta\right)=x\left(\hat{w}\left(z;\theta\right),\theta;z\right)    
\end{equation}
where $x\left(w,\theta;z\right)$ is the amount of $x$ that type $\left(w,\theta\right)$ would consume if they had taxable income $z$ (whether or not this is the level of income they actual choose).

As in Section \ref{sec:model}, a vertically neutral reform induces three types of compensated behavioral responses. First, a reform which increases the linear commodity tax rate generates the following first-order compensated response of dirty good consumption:
\vspace{-0.3em}\begin{equation}\vspace{-0.3em}
\frac{\partial\hat{x}(z,\theta)}{\partial\kappa}\bigg|_{z,\,\kappa=0}=-\frac{\varepsilon_{x\mid z}(z)\hat{x}(z,\theta)}{1+t_{x}},\label{eq:eq:DN_dirtygood_ownprice_het}
\end{equation}
for agents of type $\theta$ with taxable income $z$.

Next, we turn to the effect of the vertically neutral reform on the first stage of the agent's problem, in which they choose which level of income to earn. A type $\theta$ agent with income $z$ sees their effective marginal income tax rate change by $-\kappa x_{\text{het}}^{\prime}(z;\theta)$, where:
\vspace{-0.3em}\begin{equation}\vspace{-0.3em}
x_{\text{het}}^{\prime}(z;\theta)\equiv\bar{x}'\left(z\right)-x_{\text{inc}}^{\prime}(z;\theta).
\end{equation}
The interpretation of $x_{\text{het}}^{\prime}(z;\theta)$ is very similar to the unidimensional case. Fixing some income level $z$, it is the difference between the cross-sectional relationship between income and the average level of dirty good consumption, and the causal effect of additional income on consumption of $x$ for agents of type $\theta$.

This change in the incentive to earn income leads to a potentially type-specific compensated taxable income response of:
\vspace{-0.3em}\begin{equation}\vspace{-0.3em}
\frac{\partial z\left(\hat{w}\left(z;\theta\right),\theta\right)}{\partial\kappa}\bigg|_{\kappa=0}^{c}=\frac{z\varepsilon_{z}(z;\theta)}{1-\mathcal{T}_{z}'(z)}x_{\text{het}}^{\prime}(z;\theta)\label{eq:DN_inc_ownprice_het}
\end{equation}
where $\varepsilon_{z}(z;\theta)$ is the elasticity of taxable income for an agent of type $(w,\theta)$. Also as before, this change in income induces a change in dirty good consumption:
\vspace{-0.3em}\begin{equation}\vspace{-0.3em}
\frac{\partial\hat{x}(z,\theta)}{\partial z}\bigg|_{\kappa=0}\frac{\partial z\left(\hat{w}\left(z;\theta\right),\theta\right)}{\partial\kappa}\bigg|_{\kappa=0}^{c}=x_{\text{inc}}^{\prime}(z;\theta)\frac{z\varepsilon_{z}(z;\theta)}{1-\mathcal{T}_{z}'(z)}x_{\text{het}}^{\prime}(z;\theta).\label{eq:DN_dirtygood_incEff_het}
\end{equation}

Unlike a distribution neutral reform in the unidimensional case, there is redistribution between agents from this reform, because it is only \emph{vertically} neutral. Unlike the distribution neutral case, some agents are left unable to afford their original bundle, while others can now afford strictly more than before.

To be more precise, the impact of our vertically neutral reform on the virtual income of a type $\theta$ agent with taxable income $z$ is $\left(\hat{x}(z,\theta)-\bar{x}\left(z\right)\right)\kappa$. The impact on social welfare is equal to this change in income, scaled by the marginal social welfare weight of the affected agent, $g(z,\theta)$. These welfare weights are defined analogously to the unidimensional case:
\begin{align*}
g\left(z,\theta\right)&\equiv\gamma(\hat{w}(z,\theta);\theta)\frac{u_{c}(\hat{c}(z,\theta),\hat{x}(z,\theta),z;\hat{w}(z,\theta),\theta)}{\lambda}
+\left[\mathcal{T}_{x}^{\prime}\left(\hat{x}(z,\theta)\right)-\frac{\mathcal{D}^{\prime}\left(\bar{x}\right)}{\lambda}\right]\left.\frac{\partial\hat{x}(z,\theta)}{\partial I}\right|_{z}\\
&+\left\{ \mathcal{T}_{z}^{\prime}(z)+\left[\mathcal{T}_{x(w)}^{\prime}(\hat{x}(z,\theta))-\frac{\mathcal{D}^{\prime}\left(\bar{x}\right)}{\lambda}\right]x_{\text{inc}}^{\prime}(z, \theta)\right\}\frac{\partial z(\hat{w}(z,\theta),\theta)}{\partial I},
\end{align*}
Integrating across all levels of income and all agent types within
each level of income, we get the following expression: 

\vspace{-2em}
\begin{align}
\left.\frac{\partial\mathcal{L}}{\partial\kappa}\right|_{\kappa=0} & =\underbrace{\int\mathbb{E}\left[(x(z;\theta)-\bar{x}(z)) g(z;\theta)\mid z\right]h_{z}(z)\,dz}_{\text{Welfare effect from mechanical redistribution}}\label{eq:multidimWelfareEffectLinear}\\
&+\underbrace{\int\mathbb{E}\left[\Delta\bm{\mathcal{W}}^{x\mid z}(z;\theta)+\Delta\bm{\mathcal{W}}^{z}(z;\theta)+\Delta\bm{\mathcal{W}}^{z\rightarrow x}(z;\theta)\mid z\right]h_{z}(z)\,dz}_{\text{Welfare effects from compensated behavioral responses}},\nonumber 
\end{align}
where the three behavioral impacts conceptually mirror the unidimensional case.

As before, $\Delta\bm{\mathcal{W}}^{z}$ captures the impact on government revenue from taxable income choices. The second term, $\Delta\bm{\mathcal{W}}^{x\mid z}$, is the impact of changes in consumption of $x$ that are directly caused by the tax reform. Finally, $\Delta\bm{\mathcal{W}}^{z\rightarrow x}$ is the effect of further changes in consumption of $x$ that are induced indirectly by changes in income.
\begin{align}
    \Delta\bm{\mathcal{W}}^{z}(z;\theta)&=\mathcal{T}_{z}'(z)\frac{\varepsilon_{z}(z;\theta)z}{1-\mathcal{T}_{z}'(z)} x'_{\text{het}}(z;\theta)\\
    \Delta\bm{\mathcal{W}}^{x\mid z}(z;\theta)&=-\left(t_{x}-\frac{\mathcal{D}'(\bar{x})}{\lambda}\right)\frac{\varepsilon_{x|z}(z;\theta)x(z;\theta)}{1+t_{x}}\\
    \Delta\bm{\mathcal{W}}^{z\rightarrow x}(z;\theta)&=\left(t_{x}-\frac{\mathcal{D}'(\bar{x})}{\lambda}\right) x'_{\text{inc}}(z;\theta)\frac{\varepsilon_{z}(z;\theta)z}{1-\mathcal{T}_{z}'(z)} x'_{\text{het}}(z;\theta)
\end{align}

\vspace{-1.4em}
\subsection{Heterogeneity in Behavioral Responses}\label{sec:hetbeh}
\vspace{-0.3em}

The presence of mechanical redistribution is not the only difference between equation \ref{eq:multidimWelfareEffectLinear} and its unidimensional counterpart (equation \ref{eq:changewelf}). Variation in the behavioral responses of different agents with the same level of income is also important.

To understand this, consider the income tax revenue that results from a marginal increase in the dirty goods tax
rate by $\text{d}t_{x}$, along with the compensating change to the income tax. This is $\mathbb{E}\left[\Delta\bm{\mathcal{W}}^{z}(z;\theta)\mid z\right]$, which is the expectation of the impact at $z$ across agents with different values of $\theta$. Recall that this behavioral response arises from an agent anticipating the impact of a change in the tax rate on $x$ on their incentive to earn additional income.  For an agent of type $\theta$ and income $z$, the net impact of the vertically neutral reform is to change the agent's effective marginal income tax rate by $x'_{\text{het}}(z;\theta)\text{d}t_{x}$. Clearly, this varies across agents with different types to the extent that there is heterogeneity in the causal income effect, $x'_{\text{inc}}(z;\theta)$. In addition, there may be heterogeneity in the elasticity of taxable income, $\varepsilon_{z}(z;\theta)$, across types with the same income. 

Heterogeneity in these behavioral responses means that we cannot simply replace each behavioral parameter with its mean value. Instead, we need to take into account the potential for $\varepsilon_{z}(z;\theta)$ and $x'_{\text{inc}}(z;\theta)$ to be correlated at a given level of income. This leads to the following equation for $\mathbb{E}\left[\Delta\bm{\mathcal{W}}^{z}(z;\theta)\mid z\right]$:
\begin{equation}\label{eq:multicomp1}
    \mathbb{E}\left[\Delta\bm{\mathcal{W}}^{z}(z;\theta)\mid z\right]=\mathcal{T}_{z}'(z)\cdot\frac{\bar{\varepsilon}_{z}(z)z}{1-\mathcal{T}_{z}'(z)}\cdot\bar{x}'_{\text{het}}(z)+\frac{\mathcal{T}_{z}'(z)z}{1-\mathcal{T}_{z}'(z)}\mathbb{C}\left[\varepsilon_{z}(z;\theta),x'_{\text{het}}(z;\theta)|z\right]
\end{equation}
where $\bar{x}'_{\text{het}}(z)$ is the average value of $x'_{\text{het}}(z;\theta)$ for agents at income level $z$ and $\mathbb{C}$ indicates a covariance. Equation \ref{eq:multicomp1} highlights that ignoring heterogeneity and using $\bar{x}'_{\text{het}}(z)$ to measure this effect may be valid, but only if $\varepsilon_{z}(z;\theta)$ and $x'_{\text{inc}}(z;\theta)$ are uncorrelated.

Turning to the expectation of the second effect, $\mathbb{E}\left[\Delta\bm{\mathcal{W}}^{z\rightarrow x}(z;\theta)\mid z\right]$, we can again rewrite this as a function of the average values of each statistic and their covariances:
\begin{align}\mathbb{E}\left[\Delta\mathbf{\mathcal{W}}^{z\rightarrow x}(z;\theta)\mid z\right] & =\left(\frac{t_{x}-\frac{\mathcal{D}^{\prime}(\bar{x})}{\lambda}}{1-\mathcal{T}_{z}^{\prime}(z)}\right)z\bar{\varepsilon}_{z}(z)\bar{x}_{\text{inc}}^{\prime}(z)\bar{x}_{\text{het}}^{\prime}(z)\\  & +\left(\frac{t_{x}-\frac{\mathcal{D}^{\prime}(\bar{x})}{\lambda}}{1-\mathcal{T}_{z}^{\prime}(z)}\right)z\bar{\varepsilon}_{z}(z)\mathbb{C}\left[x_{\text{inc}}^{\prime}(z;\theta),x_{\text{het}}^{\prime}(z;\theta)|z\right]\notag\\  & +\left(\frac{t_{x}-\frac{\mathcal{D}^{\prime}(\bar{x})}{\lambda}}{1-\mathcal{T}_{z}^{\prime}(z)}\right)z\;\mathbb{C}\left[\varepsilon_{z}(z;\theta),x_{\text{inc}}^{\prime}(z;\theta)x_{\text{het}}^{\prime}(z;\theta)|z\right].\notag\end{align}

The first term in this expression is simply the product of mean behavioral responses, which would also be present in a model without multidimensional heterogeneity. 

The second and third terms involve covariances. We can understand these as arising from the two steps that lead from the tax reform to indirect changes in consumption of $x$. First, the income response for an individual of type $(w,\theta)$ is proportional to $\varepsilon_{z}(z;\theta)\times x'_{\text{het}}(z;\theta)$. Second, fixing this income response, the income effect on consumption of $x$ is proportional to $x'_{\text{inc}}(z;\theta)$. Not only is there potential for these two components of the response to be correlated, but the definition of $x'_{\text{het}}(z;\theta)$ in terms of $x'_{\text{inc}}(z;\theta)$ means that there is a \textit{mechanical} relationship between them. In fact, this mechanical connection means that the first covariance collapses to the negative variance of causal income effects at $z$:
\begin{equation}\label{eq:multicomp2}
\mathbb{C}\left[x_{\text{inc}}^{\prime}(z;\theta),x_{\text{het}}^{\prime}(z;\theta)\mid z\right]= -\mathbb{V}\left[x_{\text{inc}}^{\prime}(z;\theta)\mid z\right].
\end{equation}
This is because agents with a larger $x'_{\text{het}}(z;\theta)$ experience a larger change in their effective marginal income tax rates, and also have a smaller consumption response to any given increase in their incomes. 

The second covariance makes a further adjustment to account for the possibility that the product of the causal income and heterogeneity effects $x_{\text{inc}}^{\prime}(z;\theta)x_{\text{het}}^{\prime}(z;\theta)$ may be correlated with the elasticity of taxable income, $\varepsilon_{z}(z;\theta)$. This highlights a useful and tractable special case: If we assume that $\varepsilon_{z}(z;\theta)$ is constant within income levels, then the third term is zero. However, the first covariance remains. This highlights that heterogeneity matters even without correlation between the various statistics. Only with the further restriction that there is no within-income heterogeneity at all, so that $\text{Var}(x'_{\text{inc}}(z;\theta)\mid z)=0$, does equation \ref{eq:multicomp2} become isomorphic to the unidimensional case. In this sense, multidimensional heterogeneity leads to robust attenuation of the indirect consumption effect.

\vspace{-1.4em}
\subsection{Optimal Linear Taxation: The Multidimensional Case}
\vspace{-0.4em}

Our final step is to put together the impacts on social welfare from a vertically neutral reform, and present an optimal tax condition, harnessing the fact that the welfare effect of any vertically neutral reform (equation
\ref{eq:multidimWelfareEffectLinear}) must equal zero at the optimum.

By construction, the reforms we have considered do not redistribute vertically. For simplicity, we also assume throughout the analysis below that welfare weights are uncorrelated with consumption of the externality-generating good conditional on income, at least on average. That is, we assume that:
\vspace{-0.3em}\begin{equation}\vspace{-0.3em}
\mathbb{E}\left[\mathbb{C}\left(g\left(w,\theta\right),x\left(w,\theta\right)|z\left(w,\theta\right)=z\right)\right]=0.
\end{equation}
In practice, the sign of this covariance depends on whether more intense consumers of carbon-intensive goods have higher or lower welfare weights conditional on income. 

Assuming that the effect on welfare of redistribution within income levels is negligible on average, rearranging equation
\ref{eq:multidimWelfareEffectLinear} yields the following characterization of the optimal linear tax on the externality-generating good.
\begin{prop}
\label{prop:multidim} Assume that welfare weights are uncorrelated
with consumption of the externality-generating good conditional on
income: 
\vspace{-0.3em}\begin{equation}\vspace{-0.3em}
\mathbb{E}\left[\mathbb{C}(g(z;\theta),x(z;\theta)\mid z)\right]=0.
\end{equation}
Then the optimal linear tax on the externality-generating good satisfies:
\vspace{-0.3em}\begin{equation}\vspace{-0.3em}
t_{x}-\frac{\mathcal{D}'(\bar{x})}{\lambda}=\frac{\mathbb{E}_{z}\left[\mathbb{E}_{\theta}\left[\text{RR}(z;\theta)\mathcal{T}_{z}'(z)\mid z\right]\right]}{1-\mathbb{E}_{z}\left[\mathbb{E}_{\theta}\left[\text{RR}(z;\theta) x'_{\text{inc}}(z;\theta)\mid z\right]\right]},\label{eq:opt_multidim_general}
\end{equation}
where: 
\vspace{-0.3em}\begin{equation}\vspace{-0.3em}
\text{RR}(z;\theta)\equiv\eta_{x,z}^{\text{Taste}}(z;\theta)\left(\frac{\varepsilon_{z}(z;\theta)\bar{x}(z)}{\mathbb{E}[\varepsilon_{x|z}(z;\theta) x(z;\theta)]}\right)\left(\frac{1+t_{x}}{1-\mathcal{T}_{z}'(z)}\right),
\end{equation}
and the type-specific taste elasticity is: 
\vspace{-0.3em}\begin{equation}\vspace{-0.3em}
\eta_{x,z}^{\text{Taste}}(z;\theta)\equiv\frac{z x'_{\text{het}}(z;\theta)}{\bar{x}(z)}=\frac{z(\bar{x}'(z)-x'_{\text{inc}}(z;\theta))}{\bar{x}(z)}.
\end{equation}
\end{prop}
Equation \eqref{eq:opt_multidim_general} closely parallels the optimal
linear tax condition from the unidimensional case (Proposition \ref{prop:lin}),
but with expectations taken over both income levels and $\theta$
types within each income level. The taste elasticity $\eta_{x,z}^{\text{Taste}}(z;\theta)$
now measures how a type $\theta$ individual's causal income effect
diverges from the overall cross-sectional relationship between income
and consumption.

\vspace{-1.3em}
\subsection{Decomposing the Multidimensional Optimal Tax Condition}
\vspace{-0.4em}

To guide our empirical implementation, we can decompose the expectations in equation \ref{eq:opt_multidim_general}
into average effects and within-income covariance terms. This follows the same logic as when we decomposed the individual behavioral responses in Section \ref{sec:hetbeh}. 

First, the numerator of the right-hand side of equation \ref{eq:opt_multidim_general} can be
written as $\mathcal{T}_{z}'(z)/[1-\mathcal{T}_{z}'(z)]\times \textbf{N}$ where:
\vspace{-0.3em}\begin{equation*}\vspace{-0.3em}
\textbf{N} =\mathbb{E}_{z}\left[\bar{\eta}_{x,z}^{\text{Taste}}(z)\frac{\bar{x}(z)\bar{\varepsilon}_{z}(z)}{\mathbb{E}[\varepsilon_{x|z}(z;\theta) x(z;\theta)]}\right] -\mathbb{E}_{z}\left[\frac{z\;\mathbb{C}(\varepsilon_{z}(z;\theta),x'_{\text{inc}}(z;\theta)\mid z)}{\mathbb{E}[\varepsilon_{x|z}(z;\theta) x(z;\theta)]}\right].
\end{equation*}
where $\bar{\eta}_{x,z}^{\text{Taste}}(z)$ and $\bar{\varepsilon}_{z}(z)$
are averages across $\theta$ types at income $z$.

The first term of this equation involves average values of each sufficient statistic at
each income level. To use this first term while ignoring the second would be equivalent to naively applying the unidimensional formula with the average values of each sufficient statistic. The second term is a correction involving the within-income covariance between
the elasticity of taxable income and the causal income effect on $x$.

Section \ref{sec:hetbeh} explains where this covariance term comes from. We now provide complementary intuition. Consider a population where, conditional on $z$, agents with higher $\varepsilon_{z}$ also have higher $x'_{\text{inc}}$ (and thus lower $\eta_{x,z}^{\text{Taste}}$). Such agents have two key characteristics. First, they respond strongly to income tax changes. Second, the tagging value of carbon is low for such agents. Thus, if the planner could condition taxes on $\theta$, they would set a smaller deviation from Pigouvian taxation for these agents. In reality, the planner can only condition on $z$. However, the optimal uniform policy must account for the fact that the most responsive agents are precisely those for whom tagging is least valuable.

Second, the denominator of the right-hand side of equation \ref{eq:opt_multidim_general} can be
written as $1-[1+t_{x}]/[1-\mathcal{T}_{z}'(z)]\times \textbf{D}$ where:
\vspace{-0.3em}\begin{equation*}\vspace{-0.3em}
\textbf{D}  = \mathbb{E}_{z}\left[\bar{\eta}_{x,z}^{\text{Taste}}(z)\frac{\bar{x}(z)\mathbb{E}[x'_{\text{inc}}(z;\theta)\varepsilon_{z}(z;\theta)|z]}{\mathbb{E}[\varepsilon_{x|z}(z;\theta) x(z;\theta)]}\right]
  -\mathbb{E}_{z}\left[\frac{z\;\mathbb{C}(x'_{\text{inc}}(z;\theta),x'_{\text{inc}}(z;\theta)\varepsilon_{z}(z;\theta)|z)}{\mathbb{E}[\varepsilon_{x|z}(z;\theta) x(z;\theta)]}\right].
\end{equation*}

As Section \ref{sec:hetbeh} explains, this covariance comes from correlation between the tagging value of the externality-generating good, and indirectly induced consumption of that good. Suppose that among agents at a given level of income, two things go together: (1) a high tagging value of $x$, as reflected by a low value of $x'_{\text{inc}}$; and (2) a small indirect response of $x$ because $x'_{\text{inc}}$ is small, $\varepsilon_{z}$ is small, or both. Then this covariance term will be positive, which increases the denominator and dampens any optimal deviation from the Pigouvian level of the optimal tax. Intuitively, the planner must take into account that agents for whom tagging is most valuable are precisely those for whom any efficiency gains from a reform lead indirectly to more counterproductive externality-generating consumption.

\vspace{-1.2em}
\subsection{A Tractable Case: Constant Elasticity of Taxable Income}
\vspace{-0.3em}

In our empirical implementation, we assume that the elasticity of taxable income does not vary with $\theta$ at each income level, so that
$\varepsilon_{z}(z;\theta)=\bar{\varepsilon}_{z}(z)$ for all $\theta$.
Proposition \ref{prop:multidim_tractable} presents the optimal tax condition for this empirically implementable case. In addition to the average value of each statistic, the condition features the variance of income effects, $x'_{\text{inc}}$. Holding all other empirical objects constant, this extra term always attenuates any deviation of the optimal tax on $x$ from the Pigouvian level. Only in the special case in which there is no within-income heterogeneity, so that $\text{Var}(x'_{\text{inc}}(z;\theta)\mid z)=0$, does equation \eqref{eq:opt_multidim_tractable} simplify to the unidimensional optimal linear tax condition in Proposition \ref{prop:lin}. 

\smallskip
\begin{corr}
\label{prop:multidim_tractable} If $\varepsilon_{z}(z;\theta)=\bar{\varepsilon}_{z}(z)$
for all $\theta$ at each income level $z$, the optimal linear tax
satisfies: 
\vspace{-0.3em}\begin{equation}\vspace{-0.5em}
t_{x}-\frac{\mathcal{D}'(\bar{x})}{\lambda}=\frac{\mathbb{E}_{z}\left[\overline{\text{RR}}(z)\mathcal{T}_{z}'(z)\right]}{1-\mathbb{E}_{z}\left[\overline{\text{RR}}(z)\bar{x}'_{\text{inc}}(z)\right]+\mathbb{E}_{z}\left[\widetilde{\text{RR}}(z)\mathbb{V}(x'_{\text{inc}}(z;\theta)\mid z)\right]},\label{eq:opt_multidim_tractable}
\end{equation}
where: 

\vspace{-2.5em}
\begin{align}
\overline{\text{RR}}(z) & \equiv\bar{\eta}_{x,z}^{\text{Taste}}(z) \left(\frac{\bar{\varepsilon}_{z}(z)\bar{x}(z)}{\mathbb{E}[\varepsilon_{x|z}(z;\theta) x(z;\theta)]}\right) \left(\frac{1+t_{x}}{1-\mathcal{T}_{z}'(z)}\right),\\
\widetilde{\text{RR}}(z) & \equiv\frac{\bar{\varepsilon}_{z}(z) z}{\mathbb{E}[\varepsilon_{x|z}(z;\theta) x(z;\theta)]}\frac{1+t_{x}}{1-\mathcal{T}_{z}'(z)},
\end{align}
and $\bar{x}'_{\text{inc}}(z)\equiv\mathbb{E}[x'_{\text{inc}}(z;\theta)\mid z]$
is the average causal income effect at $z$. 
\end{corr}

Corollary \ref{prop:multidim_tractable} identifies the sufficient statistics that are needed to characterize optimal policy with multidimensional heterogeneity, subject to our restriction that the elasticity of taxable income is constant. First, we require the average value of several statistics at each income level: the cross-sectional slope, $\bar{x}'(z)$; the causal income effect, $\bar{x}'_{\text{inc}}(z)$; the elasticity of taxable income, $\bar{\varepsilon}_{z}(z)$; and the weighted average demand elasticity, $\varepsilon_{x|z}(z;\theta)$. Second, we need to measure the within-income variance in income effects, $\text{Var}(x'_{\text{inc}}(z;\theta)\mid z)$.

\vspace{-1.5em}
\section{Empirical Application}\label{sec:empirics}

\vspace{-0.5em}
Equations \ref{eq:ParetoTaxSN} and \ref{eq:ParetoTaxSL} allow us to quantify the optimal carbon tax for the United States, taking into account distributional concerns.  For a given marginal social cost of carbon, we show how the Pareto-efficient carbon tax may deviate from the standard Pigouvian correction under both nonlinear and linear commodity tax systems, presenting estimates for a range of different assumptions regarding the compensated elasticities of taxable income and dirty good consumption. We initially consider the simpler unidimensional case, and then we incorporate multidimensional heterogeneity.

We start with our empirical quantification of taste heterogeneity, $\eta_{x,z}^{\text{Taste}}\left(z\right)$. This is the focus of Section \ref{subsec:decompositon2}. In Section \ref{sec:otherasn}, we then describe how we estimate and calibrate the other key objects which jointly determine the optimal carbon tax.

\vspace{-1.2em}
\subsection{Quantifying Taste Heterogeneity}\label{subsec:decompositon2}
\vspace{-0.3em}

To estimate the extent of taste heterogeneity, we use the decomposition of the cross-sectional relationship provided by equation \ref{eq:decomp}. This requires an estimate of $\hat{x}'(z)$, which we obtain from cross-sectional variation in the consumer expenditure survey; and an estimate of $x'_{\text{inc}}$ at each of many different points in the income distribution, which we obtain by running our own survey.

\vspace{-1em}
\subsubsection{Measuring Cross-Sectional Variation Using The CEX}\label{subsubsec:cex}

\vspace{-0.5em}

We first measure the cross-sectional relationship between taxable income and consumption of carbon-intensive goods, $\hat{x}'(z)$. We do this using the Consumer Expenditure Survey (CEX), which provides quarterly expenditure data from 2019 to 2023. Specifically, we use information from the public use microdata (PUMD) files of the Consumer Expenditure Surveys provided by the United States Bureau of Labour Statistics. We rely on the Interview survey, a rotating panel survey with approximately 6,000 observations per quarter containing detailed data on expenditures, income, and demographic characteristics of consumers in the United States. Each observation is at the household level.\footnote{See detailed definition of a \enquote{Consumer Unit} at: \href{https://www.bls.gov/cex/csxfaqs.htm\#qc1}{https://www.bls.gov/cex/csxfaqs.htm\#qc1}.} Our sample is composed of 113,108 quarterly observations across 46,918 households. Because of the rotating nature of the panel, households may be sampled from once up to 4 consecutive times (quarters) across the sample period (2019-2023). 

We compute expenditure on carbon-intensive goods as a percentage of total consumption for each quarterly household observation. We focus on two carbon-intensive expenditure categories: (i) an aggregate dirty good, and (ii) gasoline and other motor oils. Aggregate dirty good consumption is defined as the sum of electricity expenditure, gasoline and motor oil expenditure, all heating fuels expenditure, and natural gas expenditure. Appendix \ref{app:empirical_avg_exp_co2} shows the share of consumption accounted for by these and other categories.

 \begin{figure}[htbp]
  \centering
  \subfloat[Aggregate Dirty Good]{\includegraphics[width=0.88\textwidth, trim=1pt 1pt 1pt 1pt, clip]{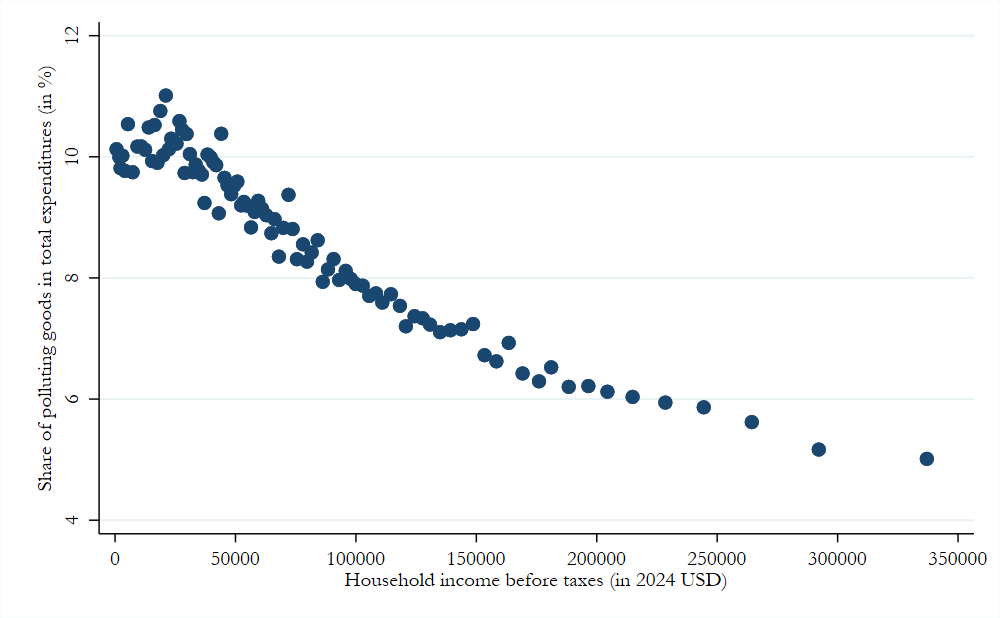}\label{fig:cross_panela_dirtygood}}
  \\
  \subfloat[Gasoline]{\includegraphics[width=0.88\textwidth, trim=1pt 1pt 1pt 1pt, clip]{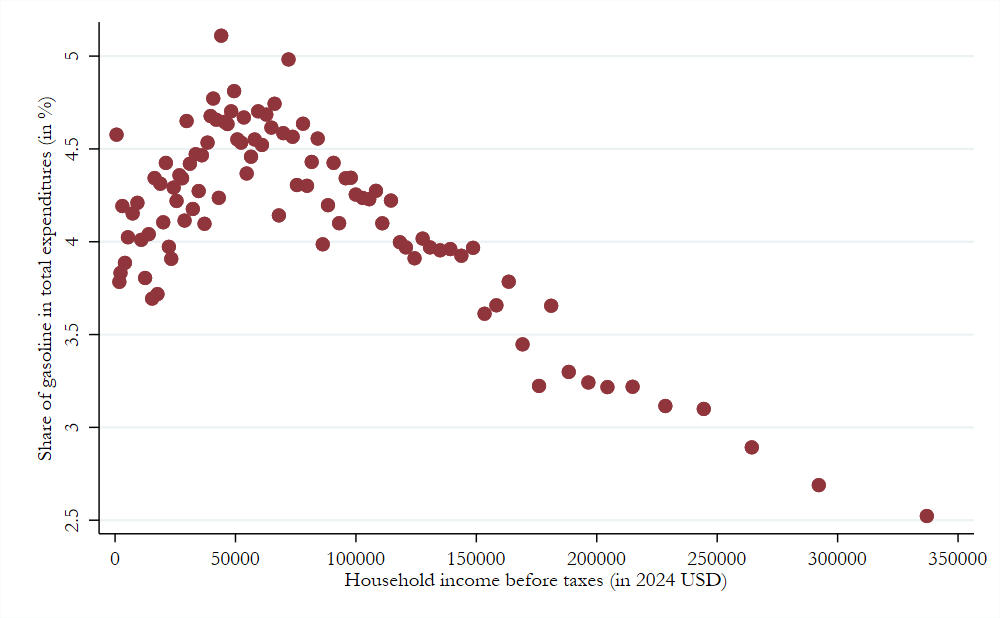}\label{fig:cross_panelb_gaso}}
  \caption{Expenditure shares of carbon-intensive goods across the income distribution}
       \begin{minipage}{\textwidth}
    \footnotesize \textbf{Note:} 
    These binscatter plots show (i) the average share (in \%) of expenditure on polluting goods as a fraction of total consumption, and (ii) the corresponding average household taxable income, for each income percentile using data from the 2018-2023 Consumer Expenditure Survey quarterly waves. The \enquote{dirty good} expenditure category includes gasoline, electricity, natural gas, and heating fuels. Household taxable incomes were converted to 2024 USD using the average annual Consumer Price Index in the US. \end{minipage}\label{fig:cross_section_expshare_income}
\end{figure}

To smooth the cross-sectional relationship between household taxable income and consumption of carbon-intensive goods, we assign each household to a year-specific income percentile (because taxable income is measured annually). We then average measures of consumption and income across quarterly observations of households within each income percentile bin, using CEX-provided survey weights to ensure the sample is representative of the US population. This binning procedure abstracts away from any within-income heterogeneity in dirty good consumption. In Section \ref{subsec:multidim_empirics}, we relax this restriction. 

Panel (a) of Figure \ref{fig:cross_section_expshare_income} shows the average \emph{dirty good} expenditure as a share (in percent) of total expenditures at each point in the income distribution. Panel (b) shows the same measure for \emph{gasoline} only.\footnote{Using total expenditure rather than annual income as the denominator is preferable when studying the distributional burden of excise taxes. As \cite{poterba1989lifetime} explains, this is because total expenditures are typically less volatile than income and thus serve as a better proxy for lifetime income. For our analyses, using total expenditures as the denominator also has the additional benefit of abstracting away from differences in the saving rate across incomes. The main takeaway from Figure \ref{fig:cross_section_expshare_income} remains the same if we compute shares of income rather than expenditure.} Poorer households spend a much larger share of their consumption budget on carbon-intensive goods than do higher-income households. For the aggregate dirty good displayed in panel (a), the share ranges from a high of 11\% at the bottom of the income distribution to a low of approximately 5\% at the very top. Focusing on gasoline only, displayed in panel (b), the expenditure share ranges from 4 to 5\% at the bottom of the income distribution to 2.5\% at the top.  

While consumption \emph{shares} of carbon-intensive goods are generally \emph{decreasing} with income (see Figure \ref{fig:cross_section_expshare_income}), it is important to note that \emph{levels} are \emph{increasing} with income. Figure \ref{fig:cross_section_cons_income} shows the cross-sectional relationship between quarterly consumption expenditure (in dollars) of carbon-intensive goods and household taxable income. The relationship is generally positive but flattens out at higher levels of household taxable income. 

\begin{figure}[htbp]
  \centering
  \subfloat[Aggregate Dirty Good]{\includegraphics[width=0.88\textwidth, trim=1pt 1pt 1pt 1pt, clip]{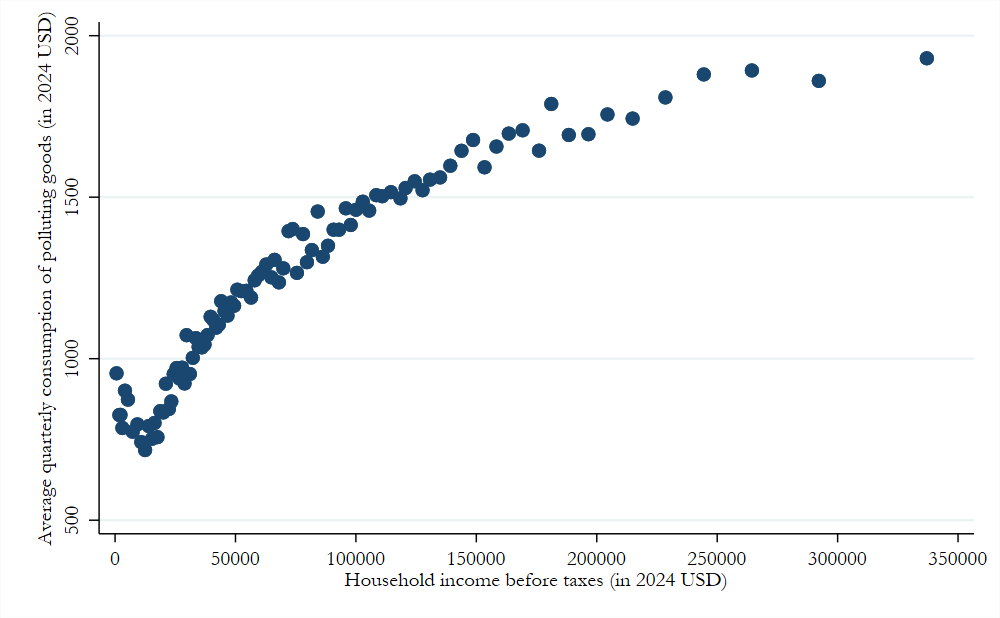}\label{fig:f1}}
  \\
  \subfloat[Gasoline]{\includegraphics[width=0.88\textwidth, trim=1pt 1pt 1pt 1pt, clip]{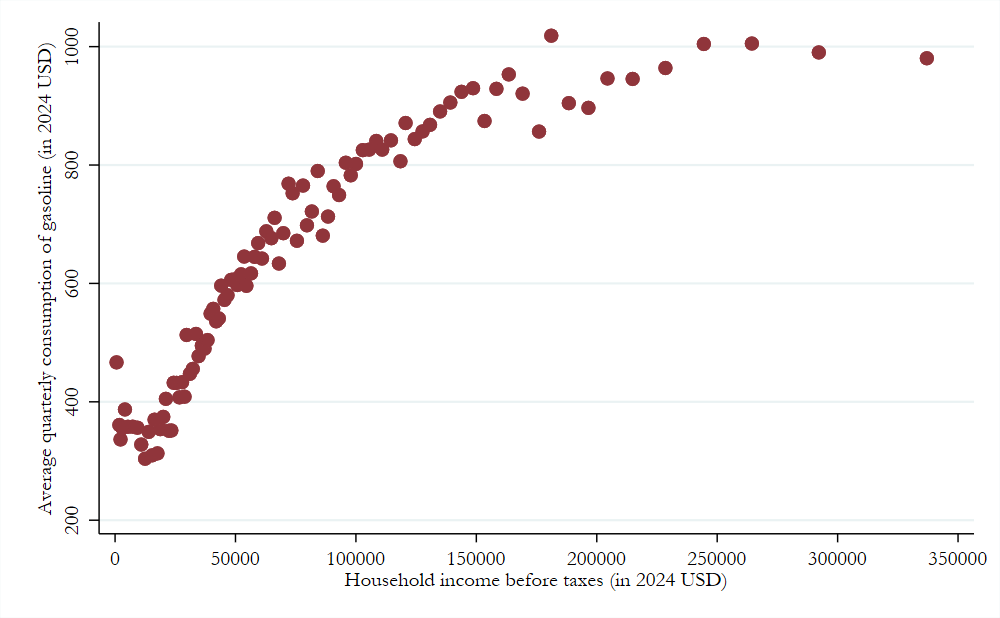}\label{fig:f2}}
  \caption{Quarterly consumption expenditure of carbon-intensive goods across the income distribution}
   \begin{minipage}{\textwidth}
    \footnotesize \textbf{Note:} 
    These binscatter plots show (i) the average quarterly expenditures on polluting goods, and (ii) the corresponding average household taxable income, for each income percentile using data from the 2018-2023 Consumer Expenditure Survey quarterly waves. The \enquote{dirty good} expenditure category includes gasoline, electricity, natural gas, and heating fuels. Household consumption expenditures and taxable incomes were converted to 2024 USD using the average annual Consumer Price Index in the US. \end{minipage}\label{fig:cross_section_cons_income}
\end{figure}

\vspace{-1em}

\subsubsection{Measuring Income Effects Using A Hypothetical Choice Survey}\label{subsubsec:survey}

\vspace{-0.6em}

Next, we measure the causal response of additional income, $x'_{\text{inc}}$. This requires a survey of our own, which measures how each household's consumption would respond to a small amount of additional income. We designed and ran the survey using \emph{Qualtrics}, and recruited participants on \emph{Prolific Academic}.\footnote{The survey and our analysis of the results were pre-registered on the Open Science Framework (OSF) platform in January, 2025: \url{https://osf.io/swqze}.} The sample is comprised of 1,448 US residents aged 25 to 54, and was designed to be representative of the US population along the targeted dimensions of age, gender, and race (see Appendix Table \ref{tab:demographics_comparison_survey_census2023}).\footnote{The survey was conducted in two waves, the first in February 2025, and the second in July 2025.}

The primary purpose of the survey was to elicit how each household would respond to an unexpected increase in household (after-tax) income. Respondents were asked to allocate a \$1,000 permanent increase in household income across an exhaustive list of expenditure categories (mirroring those from the CEX), and a residual savings bucket.\footnote{See Figure \ref{fig:q_causal_screenshot} for a screenshot of the online survey page.} This elicitation procedure provides us with each respondent's marginal propensity to consume (MPC) out of a permanent increase in household income, and its breakdown between carbon-intensive goods (electricity, natural gas, heating fuels, and gasoline) and other goods.  We also collected information about respondents' \textit{status quo} household income (both before- and after-tax), and the breakdown of these expenditures.

\underline{Response Quality.}{ } We took several steps to ensure the quality of respondents, including screening out respondents based on (i) past studies' rejection rates, and (ii) the number of hours spent on the platform. As stated in our pre-registration plan, we also filter out participants who fail attention checks, speedsters, extreme outliers, and respondents who provide inconsistent or nonsensical information on before- vs. after-tax income or marginal propensity to consume (see Appendix Table \ref{tab:sample_exclusions} for a breakdown of sample exclusions).

In addition, we have a powerful way to validate the representativeness of our survey sample. Figure \ref{fig:cross_section_survey_statusquo_expshare_income} plots the relationship between our respondents' reported consumption expenditures and their incomes. This relationship can be compared to that in the Consumer Expenditure Survey (Figure \ref{fig:cross_section_expshare_income}). The figures align remarkably closely, suggesting that our responses are of high quality and our respondents quite representative.

\underline{Marginal Consumption Patterns In The Survey.}{ } Figure \ref{fig:causal_cons_response_expshare_income} shows a binscatter plot of the relationship between the share of carbon-intensive goods in total marginal consumption, and household taxable income. We also include corresponding global polynomial fits (of order 2). Panel (a) shows the aggregate results for all of the carbon-intensive goods, while panel (b) shows the results for gasoline. In both cases, we see a declining marginal propensity to consume these goods as income rises.

\begin{figure}[htbp]
  \centering
  \subfloat[Aggregate Dirty Good]{\includegraphics[width=0.69\textwidth, trim=1pt 1pt 1pt 1pt, clip]{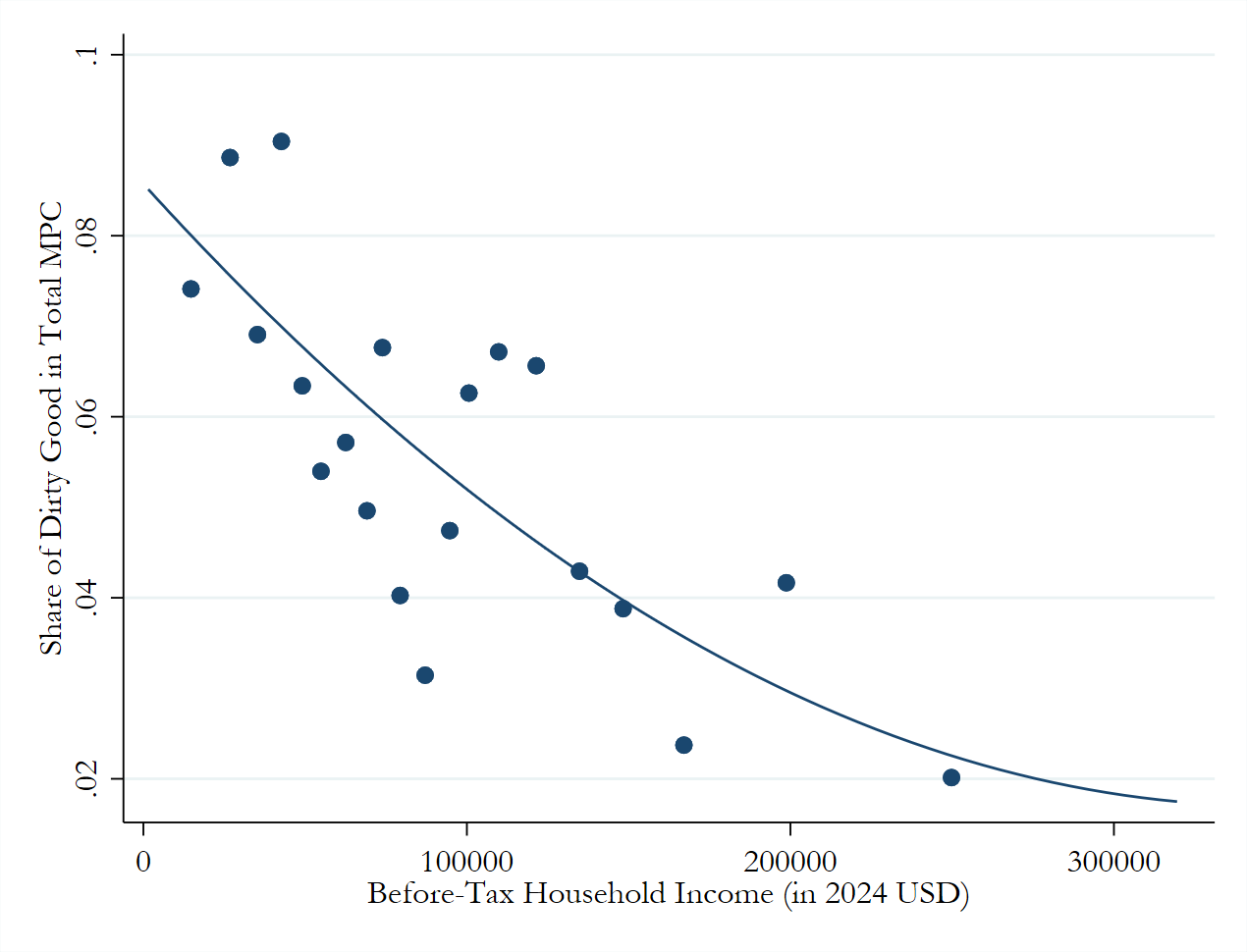}\label{fig:causal1}}
  \\[-0.5em]
  \subfloat[Gasoline]{\includegraphics[width=0.69\textwidth, trim=1pt 1pt 1pt 1pt, clip]{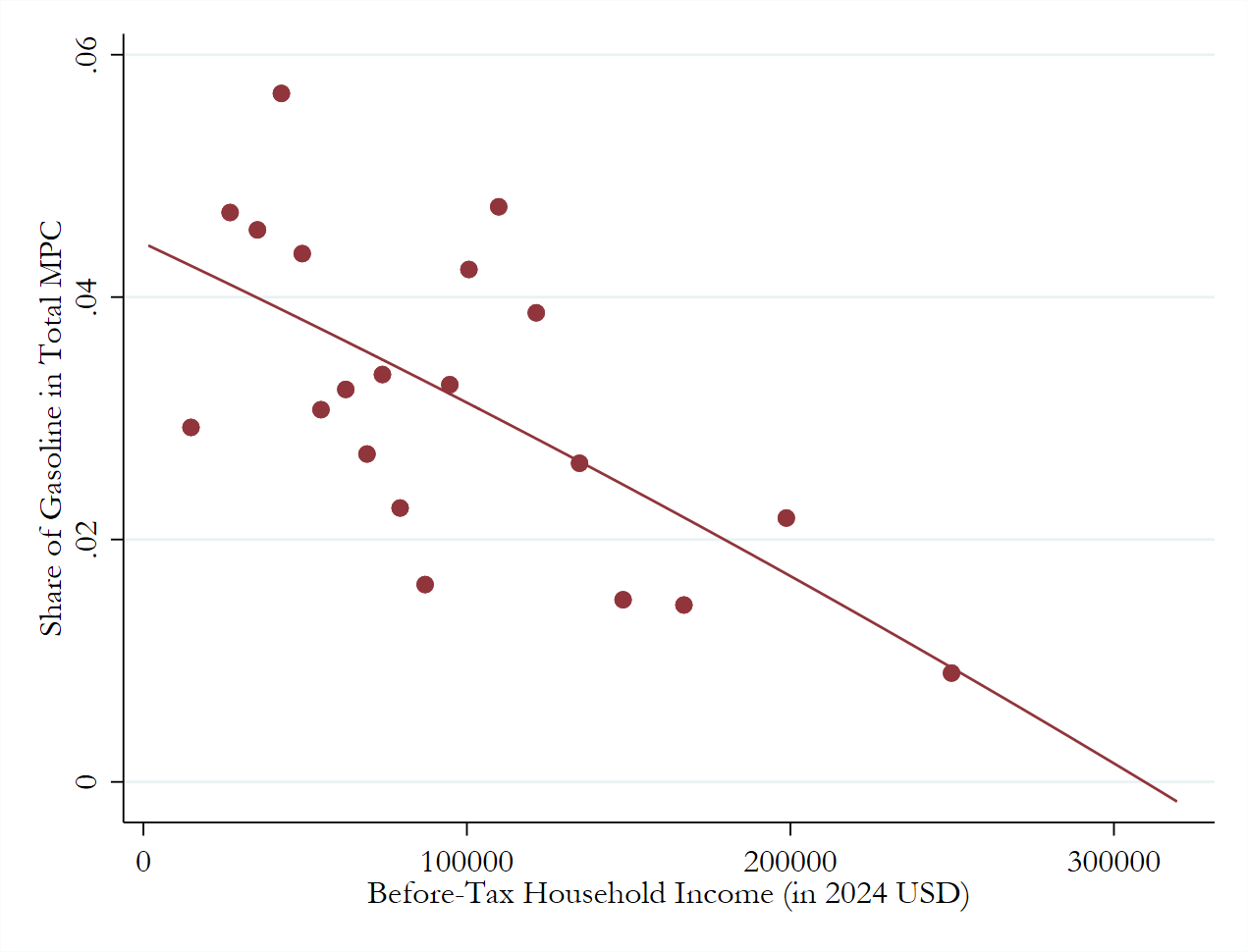}\label{fig:causal2}}
  \caption{Share of carbon-intensive goods in total MPC across the income distribution}
   \begin{minipage}{\textwidth}
    \footnotesize \textbf{Note:} 
    These binscatter plots show (i) the average share of expenditure on polluting goods as a fraction of total marginal propensity to consume (MPC) out of a small and unexpected \$1,000 increase in after-tax household income, and (ii) the corresponding average household taxable income, averaged across 20 bins of household taxable income (selected with \textit{binsreg} \citep{cattaneo2024binscatter}), and (iii) a global polynomial of degree 2 estimated to fit the relationship between these two variables, using data from a custom representative survey of the US population ($N=1,448$). The \enquote{dirty good} expenditure category includes gasoline, electricity, natural gas, and heating fuels. \end{minipage}\label{fig:causal_cons_response_expshare_income}
\end{figure}

\underline{Estimating Causal Income Effects.}{ } Our estimates for $x'_{\text{inc}}$ build on our survey respondents' marginal consumption choices in response to a hypothetical increase in their incomes (see Figure \ref{fig:causal_cons_response_expshare_income}). However, we need to make a few additional assumptions to estimate $x'_{\text{inc}}(z)$ based on the consumption response to a small permanent increase in after-tax income.

First, note that in our theoretical framework, the response of $x$ to a pure income shock (the quantity we observe in our custom survey) is given by: 
\vspace{-0.3em}\begin{equation}\vspace{-0.3em}
     \frac{d x\left(\hat{w}\left(z\right);z\right)}{d I}
= \eta_{x\mid z}(z)
+ \frac{\partial x\left(\hat{w}\left(z\right);z\right)}{\partial z}\,\eta_{z}(z),
 \end{equation}
 where $\eta_{z}$ and $\eta_{x\mid z}$ are income effects due to a reform that lowers the tax burden slightly, while leaving both marginal tax rates unchanged. Imposing weak separability between income and consumption choices allows us to make more progress here.\footnote{There is limited evidence that gasoline consumption is complementary with leisure \citep{WilliamsWest2007}, although this finding is not universal \citep{BeznoskaComp}. If this is true for overall carbon intensity, our estimates understate the optimal tax on carbon-intensive goods. See \citet{Kaplow2012} for further discussion.} 
In this case, the causal effect of \(z\) on \(x\) only occurs via the change in disposable income $I$, i.e. $\partial x /\partial z = \partial I /\partial z \times \left.\partial x/\partial I\right|_{z=z(w)} = (1- T'_z) \times \left.\partial x/\partial I\right|_{z=z(w)}$, and we can rewrite the income effect on consumption of the dirty good holding taxable income fixed as:
\vspace{-0.3em}\begin{equation}\vspace{-0.3em}
    \eta_{x \mid z}(z) 
= 
\frac{1}{1 - T'_z\bigl(z\bigr)} 
\,\frac{\partial x\left(\hat{w}\left(z\right);z\right)}{\partial z}.
\end{equation}
This means the total impact of a pure income shock on consumption of \(x\) can be written as
\vspace{-0.3em}\begin{equation}\vspace{-0.3em}
    \frac{d x\left(\hat{w}\left(z\right);z\right)}{dI}
=
\Bigl(\frac{1}{1 - T'_z\bigl(z\bigr)} + \eta_{z}(z)\Bigr)\,
\frac{\partial x\left(\hat{w}\left(z\right);z\right)}{\partial z}.
\end{equation}
If we further assume that the income effect on taxable income is negligible ($\eta_{z}\approx 0$), then 
\vspace{-0.3em}\begin{equation}\vspace{-0.3em}
        x'_{\text{inc}}(z) \equiv \frac{\partial x\left(\hat{w}\left(z\right);z\right)}{\partial z} = \left(1 - T'_z(z)\right) \frac{d x\left(\hat{w}\left(z\right);z\right)}{dI}.
\end{equation}

\vspace{-1.4em}
\subsubsection{Abstracting Away From Savings}\label{subsubsec:savings}

\vspace{-0.5em}

Our model is a static one, and thus abstracts away from the savings decision (i.e., from consumption across multiple time periods). However, in the data, households do save a fraction of their disposable income. Therefore, we make the following two adjustments to abstract away from the savings decision. First, when constructing $\hat{x}^\prime (z)$ with the CEX data, we use \emph{shares} (as in Figure~\ref{fig:cross_section_expshare_income}) rather than \emph{levels} at each income percentile, and multiply these shares by the level of after-tax income. This results in a \enquote{savings-adjusted} level of consumption of the dirty good $x$.\footnote{This approach to handling the savings decision effectively assumes that consumption shares remain constant over time, which is true with a utility function of the form:
$U_i = u^i\left(f_i(c_1, x_1) + \beta f_i(c_2, x_2), z\right)$ where $f_i(c_t, x_t)$ is homothetic within each period $t \in \{1,2\}$ and identical across periods. Note also that in order to have a direct dependence of $z$ on within-period consumption preferences, we need to allow the homothetic preference parameters to depend on $z$ (on top of their type dependence).} Second, we use \emph{shares} of total marginal expenditure (as depicted in Figure \ref{fig:causal_cons_response_expshare_income}), i.e. the additional amount spent on carbon-intensive goods as a fraction of the total additional amount spent, rather than the dollar amount spent on carbon-intensive goods divided by \$1,000.

\vspace{-1.2em}
\subsubsection{Combining Cross-Sectional and Causal Variation To Estimate Taste Heterogeneity}\label{subsubsec:combining}

\vspace{-0.5em}

Our final step is to smooth the data over the income distribution, and calculate $x'_{\text{het}}(z)$ as the difference between the fitted $\hat{x}'(z)$ and the fitted $x'_{\text{inc}}(z)$. We first convert the discrete distribution (at each percentile of taxable income) of (i) taxable income $z$, (ii) savings-adjusted consumption of the dirty good $x$, and (iii) shares of total MPCs using the global polynomials from Figure \ref{fig:causal_cons_response_expshare_income} rescaled by the net-of-tax (income) rate to obtain  $x'_{\text{inc}}(z)$ into a continuous representation by constructing a grid of equally log-spaced points. We then fit smoothing splines to the log-transformed variables, and convert them back to the original levels. Finally, we compute the derivative of the fitted cross-sectional relationship $\hat{x}(z)$ to recover $\hat{x}'(z)$ to which we can subtract the smoothed $x'_{\text{inc}}(z)$ to obtain $x'_{\text{het}}(z)$.

The results are depicted in Figure \ref{fig:decomposition_dirty_good}. For brevity, we restrict the remainder of our empirical analysis to the aggregate dirty good. However, we present results for each components Appendix \ref{app:empirical}, which are very similar in their implications.
\begin{figure}[!h]
    \centering
    \includegraphics[width=0.7\linewidth]{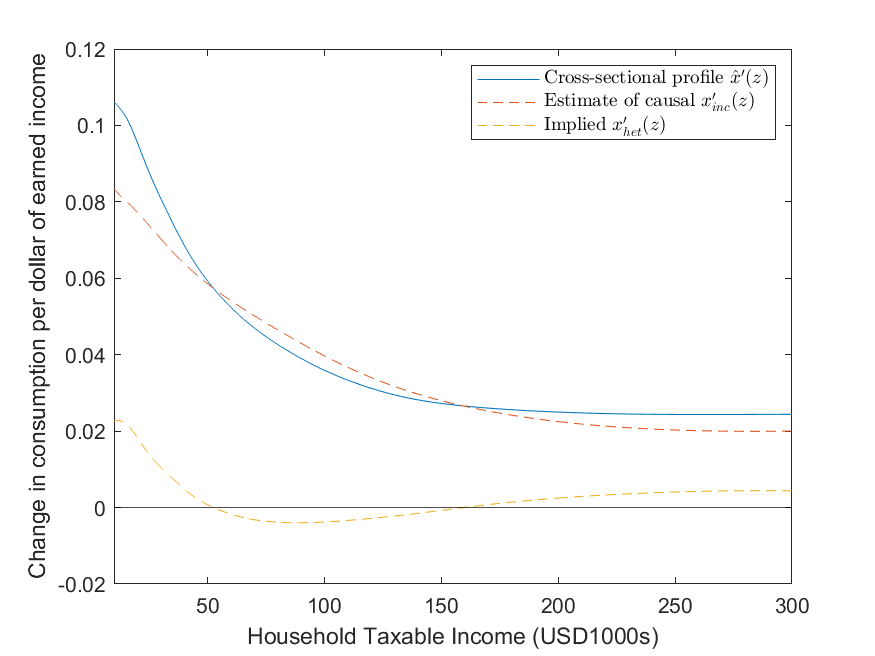}
    \caption{Decomposition of the cross-sectional profile $\hat x'(z)$ - Aggregate Dirty Good}
    \begin{minipage}{\textwidth}
    \footnotesize \textbf{Note:} 
    This figure shows how we decompose the cross-sectional relationship between income and carbon-intensive good consumption. The blue line is a smoothed estimate of the cross-sectional slope, $\hat{x}'(z)$ at each level of income. The dashed red line shows estimates of the causal effect of income on expenditure on carbon-intensive goods. The yellow dashed line subtracts the income effect from the cross-section to obtain the part of the relationship between $x$ and $z$ that comes from preference heterogeneity.\end{minipage}\vspace{-0.5em}
    \label{fig:decomposition_dirty_good}
\end{figure}

The blue curve in Figure \ref{fig:decomposition_dirty_good} depicts the cross-sectional profile $\hat{x}'(z)$. It is positive, decreasing with income, but quickly flattens out for household taxable incomes between \$100,000 and \$300,000. This illustrates the increasing and concave relationship between \emph{levels} of consumption of carbon-intensive goods and taxable income. The dashed orange curve shows our fitted estimate of $x'_{\text{inc}}(z)$, it is decreasing throughout the income distribution. The dashed yellow curve shows the implied (residual) $x'_{\text{het}}(z)$ which determines the sign of $\eta_{x,z}^{\text{Taste}}\left(z\right)$, and in turn the shape of our optimal and Pareto-efficient nonlinear tax schedules. Our estimates of $\hat{x}'(z)$ and $x'_{\text{inc}}(z)$ are remarkably close overall. This leads to estimates of $x'_{\text{het}}(z)$ that are close to zero, which implies that differences in taste contribute very little to variation in carbon intensity over the income distribution.

It is important to note that there was no reason \textit{a priori} to expect our estimate for the causal effect of income to closely match the cross-sectional profile across the income distribution. The first is computed as a slope using expenditures (i.e., levels) \emph{across} households from the large CEX surveys, while the second is an individual-level response  (i.e., a change) from our custom hypothetical survey. Their similarity is thus striking.

For direct application in our quantitative analysis, Figure \ref{fig:taste_heterogeneity_elasticity} translates this $x'_{\text{het}}(z)$ into the taste elasticity that appears in Equation \ref{eq:optTax1}. We see that $\eta_{x,z}^{\text{Taste}}\left(z\right)$ is positive for households with low incomes (below \$52,000), negative for households with taxable income between \$52,000 and \$160,000, and then being positive and increasing at the top of the income distribution. This schedule, together with the compensated elasticities of taxable income and dirty good consumption, the marginal damage, and the marginal income tax schedule, determines the shape of the Pareto efficient nonlinear tax schedule on consumption of the carbon-intensive good (see equation \ref{eq:ParetoTaxSN}).

\begin{figure}[htbp]
    \centering
    \includegraphics[width=0.7\linewidth]{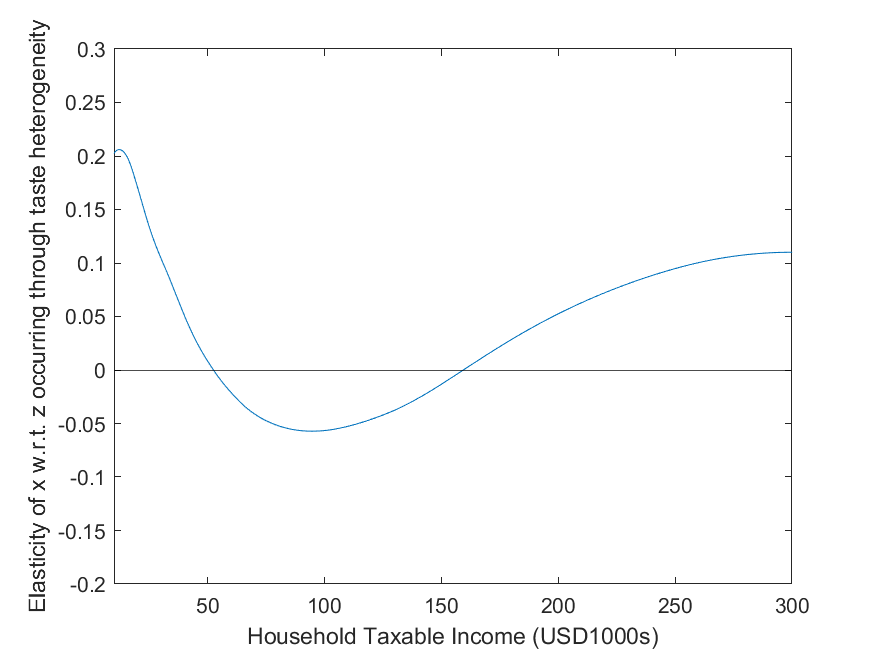}
    \caption{Elasticity of consumption of $x$ with respect to income via  taste heterogeneity only}
    \begin{minipage}{\textwidth}
    \footnotesize \textbf{Note:} 
    This figure plots the ``taste elasticity'', $\eta^{\text{Taste}}_{x,z}$. It is based on our estimates of $x^\prime_{\text{het}}(z)$, which are shown by the dashed yellow line in Figure \ref{fig:decomposition_dirty_good}, but converted into an elasticity.  \end{minipage}\label{fig:taste_heterogeneity_elasticity}
\end{figure}

\vspace{-1.4em}
\subsection{Estimating and Calibrating Other Quantities}\label{sec:otherasn}
\vspace{-0.5em}

\textbf{Elasticity of consumption of carbon-intensive goods.} We follow an empirical approach  similar to that of \cite{li2014gasoline}.
We construct a panel using the CEX, with variation over time and between states in gasoline excise taxes and prices. Then we study the consumption response of gasoline to tax changes from 1996 to 2021. We estimate the uncompensated elasticity to be $0.58$ (SE: $0.29$) in our preferred specification with month$\times$year fixed effects and household-level controls (see Appendix \ref{app:empirical_estimation_elas_gaso_consumption} for details). We thus use three values for the compensated elasticity in our tax simulations for sensitivity analysis: $0.5$ for our preferred \emph{benchmark} scenario, $0.25$ for the \emph{low} elasticity of dirty good consumption scenario, and $0.75$ for the \emph{high} elasticity of dirty good consumption scenario. 
\textbf{Elasticity of taxable income.} We calibrate the elasticity of taxable income $\varepsilon_{z}\left(z\right)$ for our \emph{benchmark} scenario to be constant at 0.33, which is the preferred estimate (on the intensive margin) from the meta-analysis of \cite{chetty2012bounds}. We also consider a value of 0.7 for our \emph{high} taxable income elasticity scenario.

\textbf{Marginal Damage From Carbon.} We follow \cite{hahn2024welfare} and assume a Social Cost of Carbon (SCC) of \$200 per ton. This is within the range of estimates provided by a recent report of the US
Environmental Protection Agency \citep{epa2023}. We assume that consumption of carbon-intensive goods generates an average of 2kg of $\text{CO}_2$ per dollar of final expenditure.\footnote{The aggregate dirty good has average CO$_2$ emissions of 2.24kg per dollar of final expenditure, calculated as the expenditure-weighted average of its three components for which emissions data are available: gasoline, electricity, and natural gas (see Appendix Table \ref{tab:co2_emissions_per_dollar_per_exp_cat} for a breakdown by component).} This results in a marginal damage of 40 cents per dollar of final expenditure, which implies a baseline Pigouvian tax of 40\% in our model\footnote{Coincidentally, this approximately corresponds to the gasoline tax in France (TICPE) in 2024 \citep{dila2024}; while in our data (see Appendix Section \ref{app:empirical_estimation_elas_gaso_consumption} for details on data sources) the average gasoline tax across US states was 18.5\% in 2021.}.  We note that our Pareto-efficient tax schedules are calibrated relative to this marginal damage estimate, which allows for flexible adjustments (of our results) based on alternative SCC values.

\textbf{Marginal income tax rate schedule.} Using the same approach as for the consumption shares described above, we compute the average before- and after-tax household income for each income percentile using data from the 2018 to 2023 CEX quarterly waves. We convert these values to 2024 USD using the average annual Consumer Price Index (CPI). We then smooth the relationship between after-tax and before-tax income, fit a smoothing spline and compute its derivative to recover the marginal tax rates on taxable income. The blue curve in Figure \ref{fig:status_quo_mtr} shows the resulting marginal income tax rates across the income distribution used in our simulations.

\vspace{-1.2em}
\subsection{Empirically Optimal Tax Rates}\label{sec:opttax}
\vspace{-0.3em}

Figure \ref{fig:pareto_tax_schedules_combined} plots the marginal tax rates on the carbon-intensive good that are implied by our Pareto efficiency conditions. Panels a and c present the nonlinear case (equation \ref{eq:ParetoTaxSN}). Panels b and d present the case where the tax is restricted to be linear (equation \ref{eq:ParetoTaxSL}). Panels a and b assume an elasticity of taxable income $\varepsilon_{z} = 0.33$, which is most consistent with \citeauthor{chetty2012bounds}'s \citeyear{chetty2012bounds} meta-analysis. Panels c and d assume $\varepsilon_{z} = 0.7$. In each panel, we present estimates for a range of values for the compensated net-of-tax rate elasticity of spending on the
 carbon-intensive good. 

\vspace{-1.2em}
\subsubsection{Optimal Non-linear Taxation of Carbon}
\vspace{-0.3em}

We focus first on the nonlinear case, assuming that non-linear taxation of these particular commodities is feasible. Panel (a) of Figure \ref{fig:pareto_tax_schedules_combined} presents our results using our baseline value of $\varepsilon_{z}$. There is a mild U-shaped pattern of the Pareto-efficient carbon tax over the income distribution, which is centered approximately on the level of marginal damage from these goods. Optimal marginal tax rates start near the level of marginal damage for very low-income households, fall below it for low- and middle-income households, and rise above the Pigouvian level for the highest-income households. This pattern is somewhat amplified when consumption of carbon-intensive goods is assumed to be very inelastic (i.e., $\varepsilon_{x|z}$ is very low).

When the elasticity of taxable income is higher ($\varepsilon_{z} = 0.7$), panel (c) shows that there are larger deviations from Pigouvian taxation, although they are still centered around the level of marginal damage. This is especially the case when $\varepsilon_{x|z}$ is low. This combination of parameters implies that income taxation is highly distortionary relative to adjusting the carbon tax to redistribute, providing that some part of the covariance of income and carbon intensity is explained by differences in preferences ($\eta^{\text{Taste}}_{x,z}(z)\neq 0$). The most important implication of these more extreme parameter values is that the optimal carbon tax is meangfully higher than the Pigouvian level for the highest-income individuals. For middle-income households, the optimal carbon tax is lower than the Pigouvian level, but it is never comes close to being Pareto efficient to subsidize carbon consumption.

\begin{sidewaysfigure}
\centering
\begin{subfigure}{.4\textwidth}
  \centering
  \caption{Nonlinear: $\varepsilon_{z}\left(z\right)=0.33$}
  \includegraphics[width=\linewidth, trim=0pt 0pt 0pt 20pt, clip]{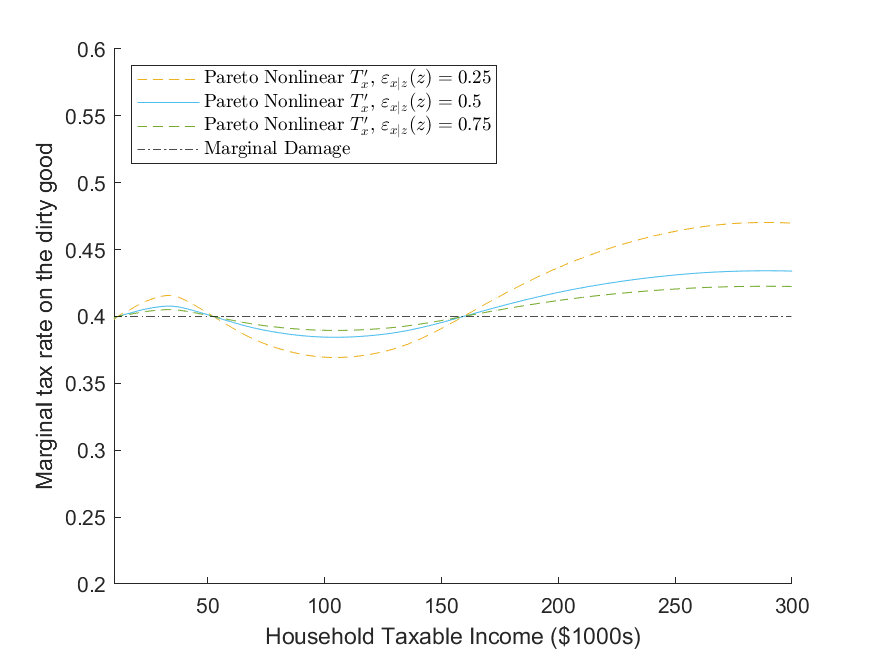}
  \label{fig:eti0point33nonlinear}
\end{subfigure}%
\begin{subfigure}{.4\textwidth}
  \centering
  \caption{Linear: $\varepsilon_{z}\left(z\right)=0.33$}
  \includegraphics[width=\linewidth, trim=0pt 0pt 0pt 20pt, clip]{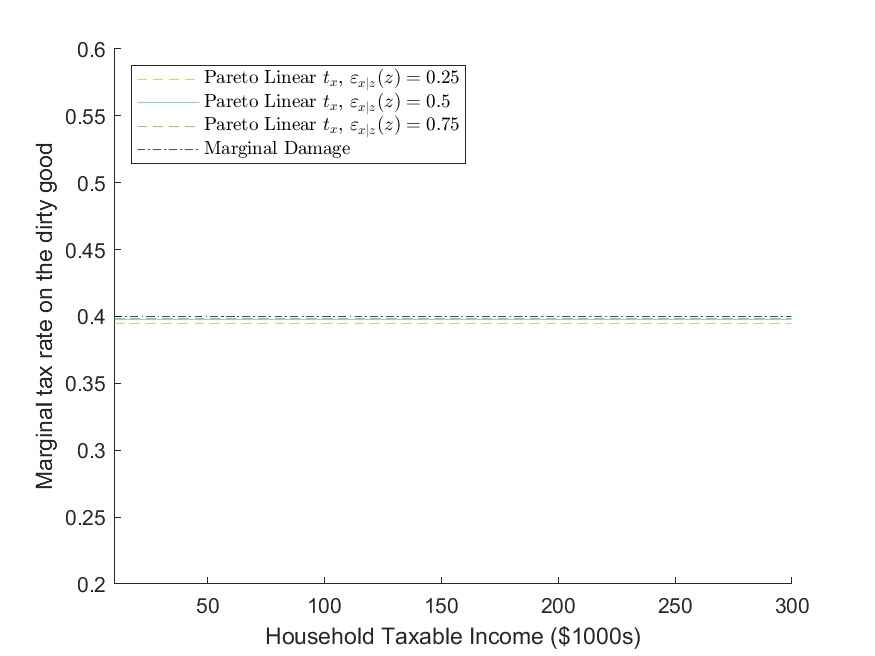}
  \label{fig:eti0point33lin}
\end{subfigure}
\begin{subfigure}{.4\textwidth}
  \centering
  \caption{Nonlinear: $\varepsilon_{z}\left(z\right)=0.7$}
  \includegraphics[width=\linewidth, trim=0pt 0pt 0pt 20pt, clip]{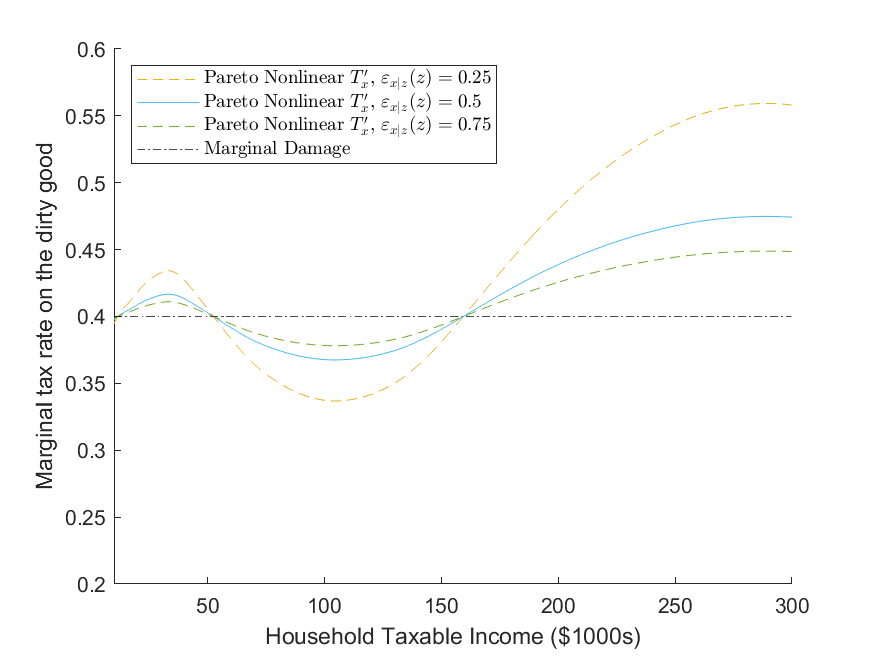}
  \label{fig:eti0point7nonlinear}
\end{subfigure}%
\begin{subfigure}{.4\textwidth}
  \centering
  \caption{Linear: $\varepsilon_{z}\left(z\right)=0.7$}
  
  \includegraphics[width=\linewidth, trim=0pt 0pt 0pt 20pt, clip]{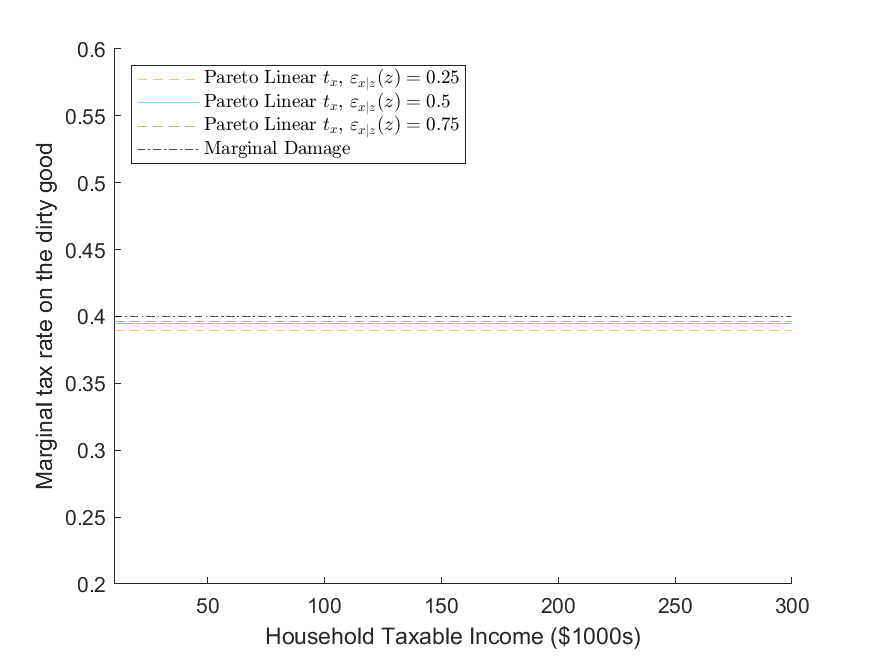}
  \label{fig:eti0point7lin}
\end{subfigure}

\vspace{-0.5em}
\caption{Dirty Good Tax Schedule: Marginal Rates implied by Pareto Efficient formulas}
\label{fig:pareto_tax_schedules_combined}
    \begin{minipage}{\textwidth}
    \footnotesize \textbf{Note:} 
    This figure plots the Pareto efficient carbon tax for different values of the elasticities of taxable income and consumption of the carbon-intensive good. Three cases are shown in each panel, along with the benchmark of the assumed value of marginal damage. Panels (a) and (c) show nonlinear marginal tax schedules, while panels (b) and (d) restrict the carbon tax to be linear.
    \end{minipage}
\end{sidewaysfigure}

\vspace{-1.2em}
\subsubsection{Optimal Linear Taxation of Carbon}
\vspace{-0.5em}

Nonlinear taxation of these commodities may prove logistically or politically challenging. In this case, the carbon tax must be restricted to be linear. Panels (b) and (d) show that the Pareto-efficient linear tax is extremely close to the Pigouvian level for all parameter values. In our benchmark scenario with $\varepsilon_{z} = 0.33$, and $\varepsilon_{x|z} = 0.5$, the Pareto-efficient linear tax is equal to 99.4\% of the value of marginal damage. In this sense, distributional concerns have very little bearing at all on the optimal carbon tax if it is required to be linear.

This stark result comes from the fact that the Pareto-efficient linear tax is a weighted average of the quantities determining the marginal tax rates at each income level in the nonlinear case (see equation \ref{eq:ParetoTaxSL}). Intuitively, departures from the Pigouvian level of the carbon tax will depend on a weighted average tagging benefit of carbon. As Figure \ref{fig:taste_heterogeneity_elasticity} shows, $\eta_{x,z}^{\text{Taste}}\left(z\right)$ is negative for the portion of the income distribution with the largest share of households, but positive at high levels of income. Integrating over the distribution, the Pareto-efficient linear tax is therefore below, but very close to, the Pigouvian level.

\vspace{-1.2em}
\subsubsection{Optimal Linear Taxation of Carbon with Multidimensional Heterogeneity}\label{subsec:multidim_empirics}
\vspace{-0.5em}

To empirically quantify the optimal carbon tax with multidimensional heterogeneity, we exploit our individual-level survey measure of the causal effect of income on consumption of the carbon-intensive good. We first group survey respondents whose household taxable income ranges from \$600 to \$325,000, into equally-sized deciles of taxable income; and then compute the conditional variance of $\partial x/\partial z$ in each income decile. Figure \ref{fig:cond_var_causal_effect_of_income_deciles} shows these conditional variances together with bootstrapped confidence intervals. 

\begin{figure}[htbp]
    \centering
    \includegraphics[width=0.7\linewidth, trim=1pt 1pt 1pt 1pt, clip]{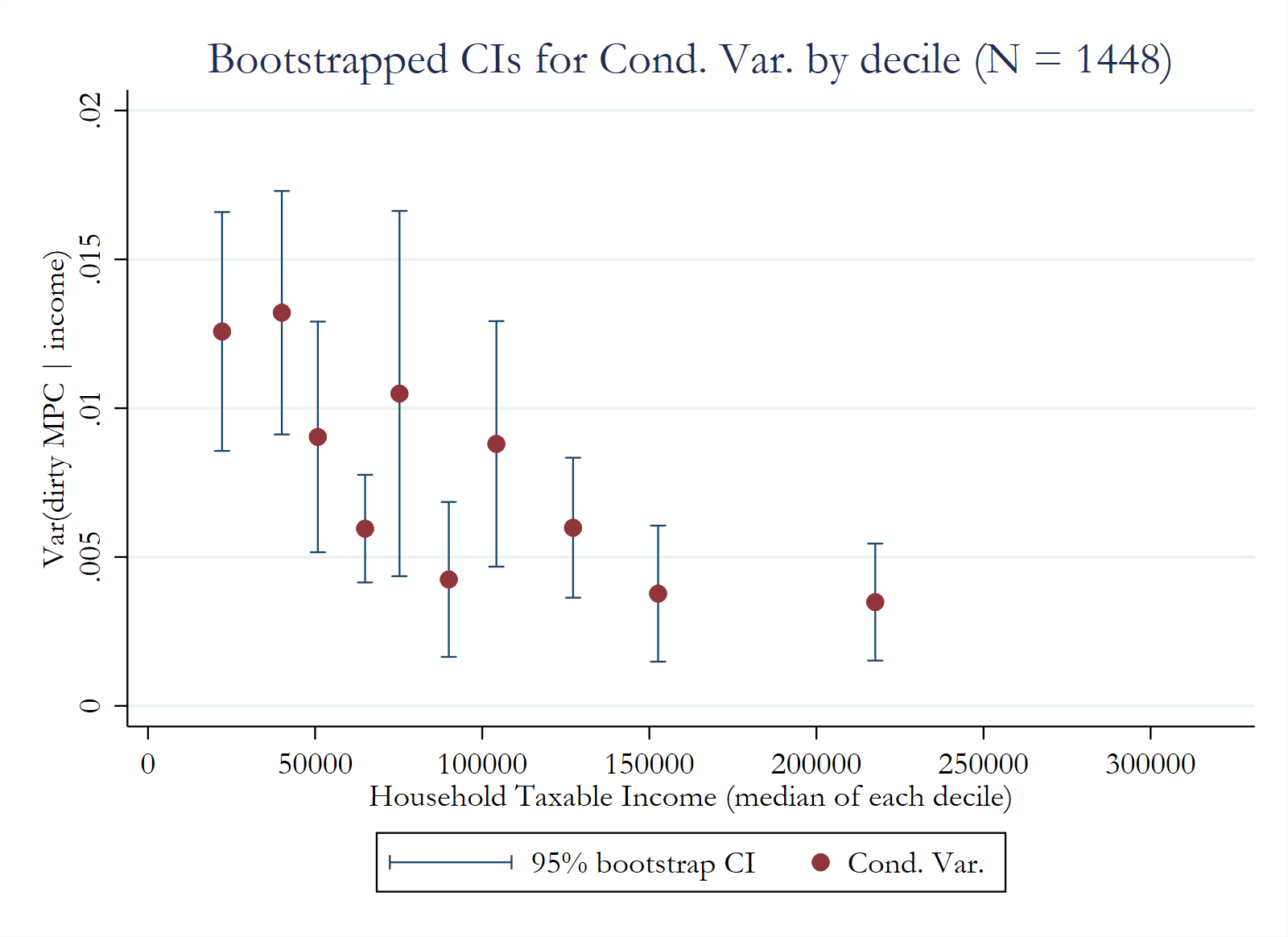}
    \vspace{-0.5em}\caption{Conditional Variance of $\partial x /\partial z$ at Each Income Decile}
    \begin{minipage}{\textwidth}
    \vspace{-0.1em}\footnotesize \textbf{Note:}
    This figure shows the conditional variance of the causal effect of income on consumption of the dirty good, $\text{Var}( \frac{\partial x (w, \theta)}{\partial z}\mid z)$, measured in our hypothetical survey  ($N = 1,448$), and estimated by grouping survey respondents across equal-binned deciles of household taxable income. It is based on each respondent's share of MPC allocated to the aggregate dirty good.\end{minipage}\label{fig:cond_var_causal_effect_of_income_deciles}
\end{figure}

The conditional variance of the causal consumption response to a small increase in taxable income is largest for households in the 2nd decile, earning between \$36,000 and \$50,000. It is otherwise slightly decreasing across the income distribution, from 0.013 to 0.004.  

We assign these conditional variances to all income levels within each decile and smooth out their relationship with taxable income across the income distribution (as described above). This allows us to simulate the optimal linear externality tax using the formula displayed in equation \ref{eq:opt_multidim_tractable}. We re-use most of the empirical quantities from the unidimensional case, now highlighting the conditional average interpretation, e.g. $\eta_{x,z}^{\text{Taste}}\left(z\right) = \mathbb{E}[ \eta_{x,z}^{\text{Taste}}\left( z, \theta \right)\mid z]$ and $\bar x(z) = \mathbb{E}[ x\left(z, \theta\right)\mid z]$. We continue to assume constant values for the elasticities of taxable income and consumption of the carbon-intensive good.

\begin{figure}[h!]
  \centering
  \subfloat[Baseline calibration]{\includegraphics[width=0.72\textwidth, trim=1pt 1pt 1pt 1pt, clip]{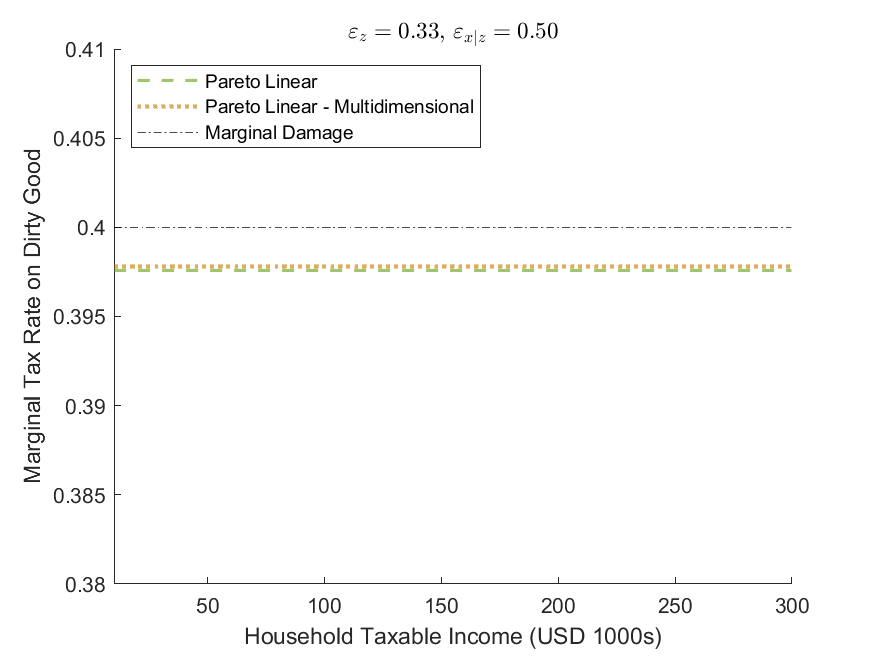}\label{fig:eti0point33_multidim}}
  \\
  \vspace{1.5em}
  \subfloat[Higher $\varepsilon_z$, lower $\varepsilon_{x\mid z}$]{\includegraphics[width=0.72\textwidth, trim=1pt 1pt 1pt 1pt, clip]{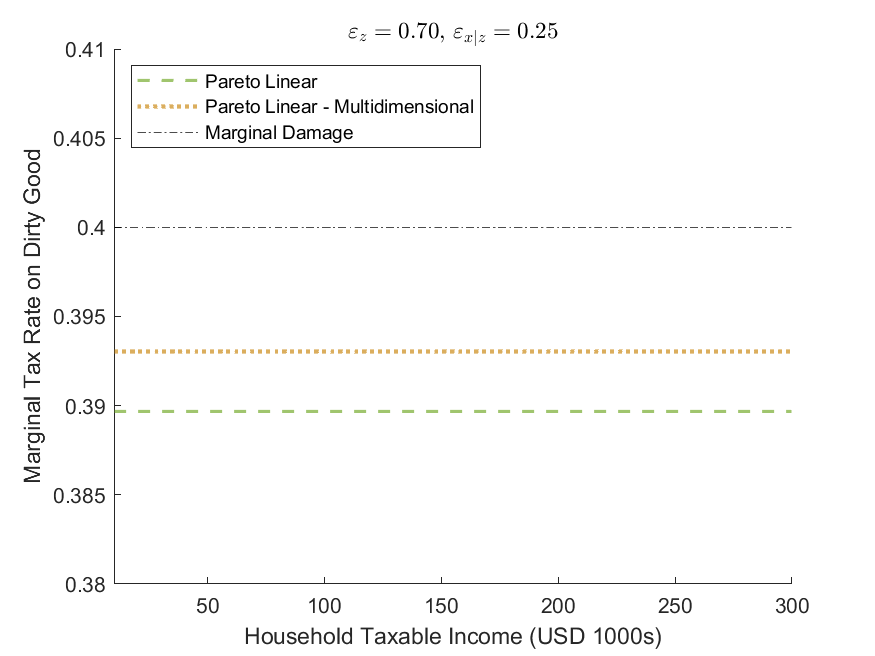}\label{fig:eti0point7_multidim}}
  \caption{Dirty Good Marginal Linear Tax Rate:  Unidimensional vs. Multidimensional Heterogeneity}
    \begin{minipage}{\textwidth}
    \footnotesize \textbf{Note:}
    This figure plots the Pareto efficient (linear) carbon tax in the unidimensional case vs. the quasi Pareto efficient linear tax rate when accounting for multidimensional heterogeneity, for different values of the elasticities of taxable income and consumption of the carbon-intensive good.  \end{minipage}
    \label{fig:multidim_quasi_pareto_tax_schedule_vs_unidim}
\end{figure}

Figure \ref{fig:multidim_quasi_pareto_tax_schedule_vs_unidim} plots the optimal tax rate for the multidimensional case (orange dotted line). For comparison, we also include the the linear tax rate implied by the Pareto efficient linear tax formula in the unidimensional case (green dashed line). Panel (a) assumes an elasticity of taxable income $\varepsilon_{z} = 0.33$, and an elasticity of consumption of the carbon-intensive good $\varepsilon_{x\mid z}=0.5$, while panel (b) assumes $\varepsilon_{z} = 0.7$, and $\varepsilon_{x\mid z}=0.25$. As expected based on our theoretical results, within-income variation of the causal effect of income leads to attenuation and pushes the optimal tax closer to the Pigouvian benchmark.

Comparing across the two panels, we see that the degree of attenuation is larger in the lower panel, where the deviation from Pigouvian taxation was relatively larger as well. This is because the amount of attenuation is proportional to the ratio of elasticities $\varepsilon_z / \varepsilon_{x\mid z}$ (see equation \ref{eq:opt_multidim_tractable}). Recall that this same ratio scales the intensity of the tagging benefit in the first place, which implies that the degree of attenuation from within-income heterogeneity in $\partial x /\partial z$ is greatest when the tagging benefit is largest. As a result, there is only minimal attenuation in Figure \ref{fig:multidim_quasi_pareto_tax_schedule_vs_unidim}(a) where the optimal linear rate in unidimensional case was almost exactly at the Pigouvian level already. There is more attenuation in panel (b), where the higher ratio of elasticities led to a larger deviation.

\vspace{-1.4em}
\section{Conclusion}\label{sec:conclusion}
\vspace{-0.5em}

The use of taxes or subsidies to internalize externalities is widely accepted by economists as an important role of government. When it comes specifically to climate change, economists have forcefully advocated for a carbon tax on these grounds \citep{EconomistsStatement2019}. But this idea has gained limited traction politically.

We formally study one specific critique of carbon taxes, and of corrective taxes more generally: They can be regressive. Our theoretical results show that the fact that poorer individuals are more carbon-intensive in their consumption does not necessarily change the optimal carbon tax. Instead, the key distinction is whether this reflects income effects or preference heterogeneity. We show how to quantify preference heterogeneity empirically, and take our sufficient statistics tax formulas to the data.

The implications of our empirical results for the optimal carbon tax depend on whether nonlinear carbon taxation is feasible. If only linear carbon taxation is possible, then the optimal carbon tax is remarkably close to the classic Pigouvian prescription of setting the tax equal to the marginal damage from emitted carbon. This conclusion for the optimal linear tax is robust to a very wide range of different parameter values.

The conclusion is more nuanced if nonlinear carbon taxation is possible. Our benchmark scenario implies that marginal tax rates are modestly below the Pigouvian level for middle-income households, but above that level for higher-income households. This pattern is amplified if the elasticity of taxable income is higher, or the behavioral response to carbon taxation is smaller: In such cases, adjusting the carbon tax becomes a more attractive tool to redistribute income, relative to income taxation.

We expect that our conceptual and empirical frameworks could be applied to a wide range of settings that involve both distributional concerns and externalities. This could include other environmental externalities, public health problems, or even criminal behavior. Our results for this particular application suggest that preference-driven differences in carbon intensity have a limited impact overall. Our estimate of the optimal carbon tax is therefore close to the Pigouvian benchmark, especially if the tax is constrained to be linear. However, this result may differ for other externalities. In the future, it will also be important to integrate our detailed work on the consumption side of the economy with richer analysis of the production side along the lines of \citet{Bierbrauer2024} and others. 

\pagebreak

\singlespacing
\bibliographystyle{apalike}
\bibliography{references.bib}
\onehalfspacing

\pagebreak

\appendix
\counterwithin{figure}{section}  
\counterwithin{table}{section}   

\renewcommand{\thefigure}{\thesection.\arabic{figure}}
\renewcommand{\thetable}{\thesection.\arabic{table}}

\section{Empirical Appendix}\label{app:empirical}

\subsection{Average Expenditure Shares and CO\textsubscript{2} Emissions across Categories}\label{app:empirical_avg_exp_co2}

\begin{table}[htbp]\centering
\def\sym#1{\ifmmode^{#1}\else\(^{#1}\)\fi}
\caption{Average Expenditure Shares across Categories}
\begin{tabular}{l*{1}{c}}
\toprule
                    &\multicolumn{1}{c}{}\\
                    &        Avg.\\
\midrule
Food and Beverages  &       0.183\\
Housing   (excluding utilities)          &       0.272\\
Electricity, natural gas, and home heating fuels&       0.044\\
Water, sewage, and other public services&       0.013\\
Healthcare          &       0.095\\
Gasoline and other motor fuels&       0.039\\
Flights / Airfare expenditures&       0.005\\
Other Expenditures  &       0.347\\
\midrule
Observations        &      113,108\\
\bottomrule
\end{tabular}

\begin{minipage}{\textwidth}
    \vspace{0.5em}
    \footnotesize \textbf{Note:} This table shows shares allocated to each consumption category among consumers in the Consumer Expenditure Survey (2018--2023). 
    \end{minipage}
\end{table}

\begin{table}[htbp]
    \centering
    \caption{Average CO\textsubscript{2} Emissions (in kg) per Dollar of Final Expenditure in the United States}
    \begin{tabular}{lcccc}
        \toprule
        & \textbf{Unit} & \textbf{Price Per} & \textbf{CO\textsubscript{2} Emissions } & \textbf{CO\textsubscript{2} Emissions} \\
        &  & \textbf{Unit (in \$)} & \textbf{per Unit (in Kg)} & \textbf{ (Kg per \$)} \\
        \midrule
        Gasoline & Gallon (3.78 Litres) & 3.635 & 8.78 & 2.415 \\
        Electricity & KWh & 0.2275 & 0.387 & 1.699 \\
        Natural Gas & 1000 Cubic Feet & 15.39 & 53.035 & 3.446 \\
        \bottomrule
    \end{tabular}  \label{tab:co2_emissions_per_dollar_per_exp_cat}
     \vspace{0.5em}
     \begin{minipage}{\textwidth}
    \footnotesize \textbf{Note:} We use 2023 data from U.S. Energy Information Administration (EIA) for information on prices per unit, as well as 2023 data from the U.S. Environmental Protection Agency (EPA), and the Code of Federal Regulations for information on CO2 emissions per unit. 
    \end{minipage}
\end{table}

\subsection{Share of Carbon-Intensive Goods in Status Quo Expenditures}\label{app:empirical_share_dirty_good_status_quo_custom_survey}

\begin{figure}[H]
  \centering
  \subfloat[Aggregate Dirty Good]{\includegraphics[width=0.64\textwidth, trim=1pt 1pt 1pt 1pt, clip]{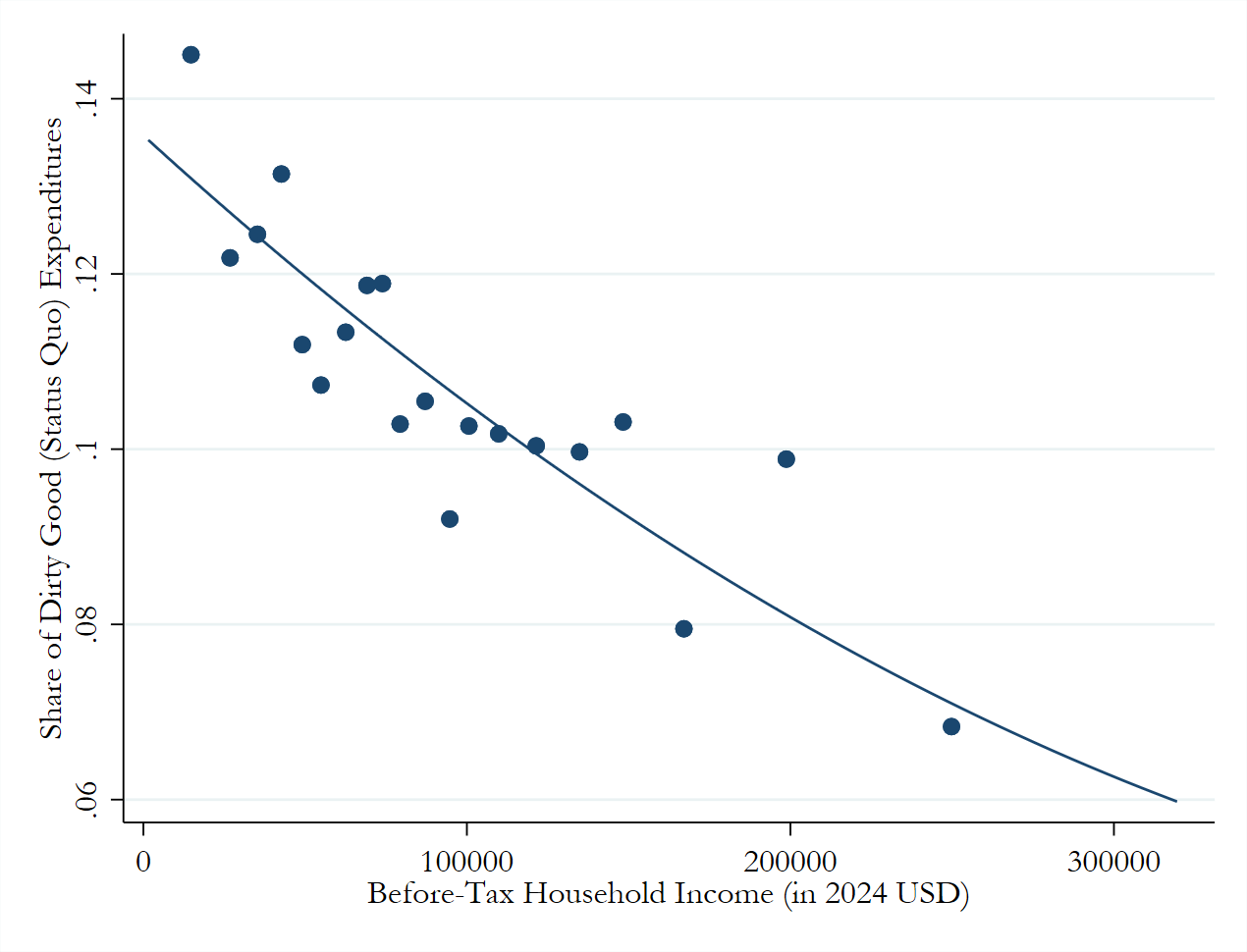}\label{fig:share_exp1surv}}
  \\
  \subfloat[Gasoline]{\includegraphics[width=0.64\textwidth, trim=1pt 1pt 1pt 1pt, clip]{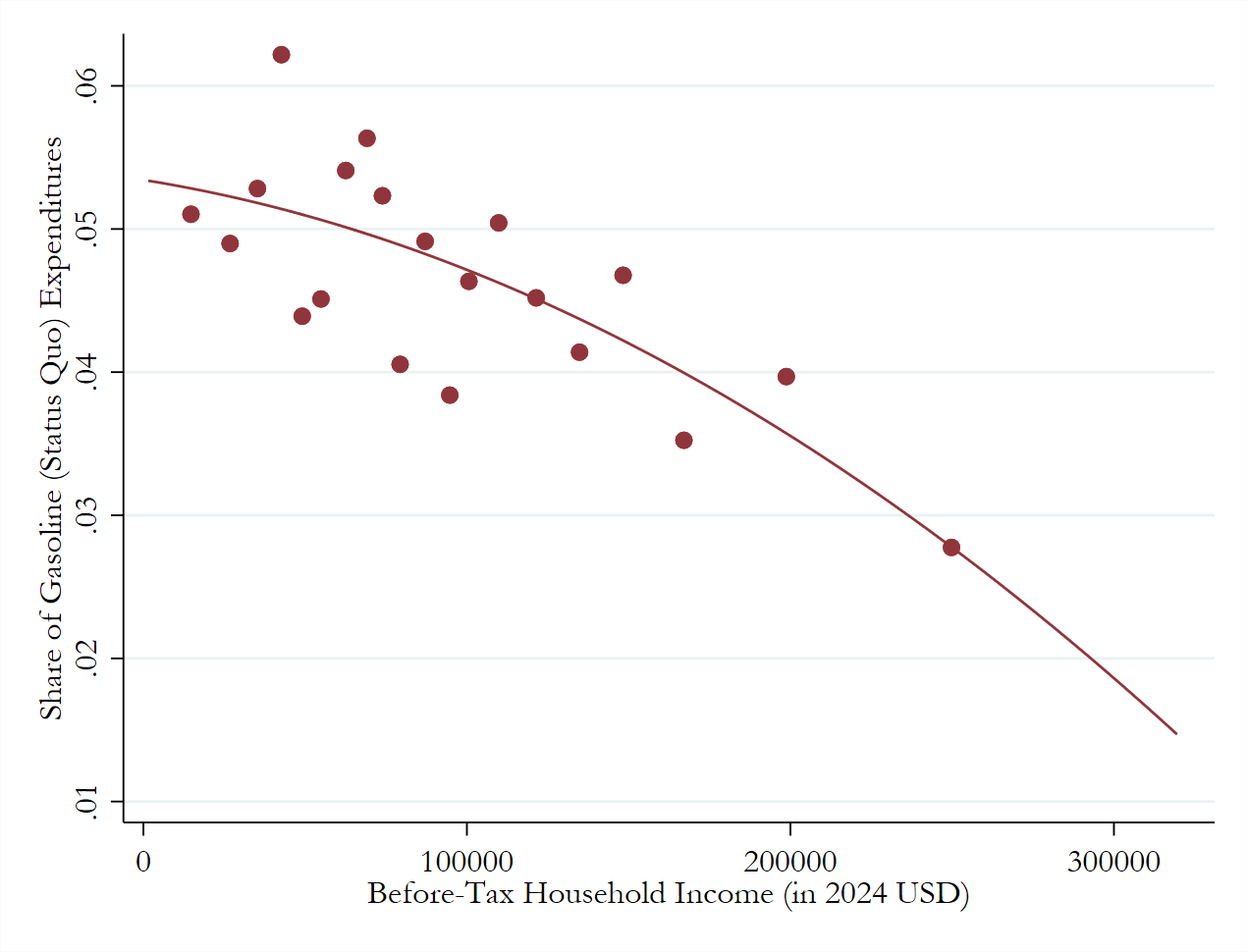}\label{fig:share_exp2surv}}
  \caption{Share of carbon-intensive goods in total (monthly) expenditures in custom survey}
   \begin{minipage}{\textwidth}
    \footnotesize \textbf{Note:} 
    These binscatter plots show (i) the average share of expenditure on polluting goods as a fraction of total expenditures, and (ii) the corresponding average household taxable income, averaged across 20 bins of household taxable income (selected with the command binsreg \citep{cattaneo2024binscatter}), and (iii) a global polynomial of degree 2 estimated to fit the relationship between these two variables, using data from a custom representative survey of the US population ($N=1,448)$. The \enquote{dirty good} expenditure category includes gasoline, electricity, natural gas, and heating fuels. \end{minipage}\label{fig:cross_section_survey_statusquo_expshare_income}
\end{figure}

\newpage 
\subsection{Estimating the Uncompensated Elasticity of Gasoline Consumption}\label{app:empirical_estimation_elas_gaso_consumption}
Following \cite{li2014gasoline}, we estimate a log-log two-way fixed effects regression to recover the elasticity of gasoline consumption with respect to the net-of-tax rate. We use granular data at the household level rather than aggregated data at the state level and estimate the following regression equation:
\begin{equation}\label{eq:log_log_reg_elas_of_cons}
        \log(q_{ist}) = \alpha \log(p_{st}) + \beta \log(1+\frac{\tau_{st}}{p_{st}}) + \Theta X_{ist} + \delta_s + \phi_t + e_{ist}
\end{equation}
where $q_{ist}$ is the household-level consumption of gasoline (in quantities), $p_{st}$ is the tax-exclusive price of gasoline, $\tau_{st}$ is the gasoline excise tax (in dollars); $X_{ist}$ is a vector of household-level controls; and $\delta_s$ and $\phi_t$ are state, and time (yearly or monthly) fixed effects.

The coefficient $-\beta$ in the first row can be interpreted as the uncompensated elasticity of gasoline consumption with respect to the net-of-tax rate.

\begin{table}[H]
    \centering
    \caption{Elasticity of gasoline consumption with respect to tax rate: 1996-2021}   
    \footnotesize
{
\def\sym#1{\ifmmode^{#1}\else\(^{#1}\)\fi}
\begin{tabular}{l*{4}{c}}
\toprule
          &\multicolumn{1}{c}{(1)}         &\multicolumn{1}{c}{(2)}         &\multicolumn{1}{c}{(3)}         &\multicolumn{1}{c}{(4)}         \\
\midrule
$\log(1 + \text{tax ratio})$&   -0.325         &   -0.386         &   -0.539\sym{*}  &   -0.584\sym{**} \\
          &  (0.312)         &  (0.313)         &  (0.288)         &  (0.288)         \\
$\log(\text{gas price})$&   -0.453\sym{***}&   -0.468\sym{***}&   -0.419\sym{***}&   -0.429\sym{***}\\
          &  (0.143)         &  (0.141)         &  (0.136)         &  (0.133)         \\
\midrule
State FE  &        Y         &        Y         &        Y         &        Y         \\
Year FE   &        Y         &                  &        Y         &                  \\
$ \text{Month} \times \text{Year} $ FE&                  &        Y         &                  &        Y         \\
Household Level Controls&        N         &        N         &        Y         &        Y         \\
R$^2$     &    0.089         &    0.109         &    0.198         &    0.218         \\
Observations&  596,822         &  596,822         &  542,574         &  542,574         \\
Clusters (State $\times$  Year)&    1,035         &    1,035         &      991         &      991         \\
\bottomrule
\end{tabular}
}

    \label{tab:elas_of_cons_wrt_tax_rate}\vspace{1em}
    \begin{minipage}{\textwidth}
  \footnotesize {\textbf{Note:}}
  This table shows the coefficient estimates from the regression of equation (\ref{eq:log_log_reg_elas_of_cons}). All regressions control for family size and use annual weights to make the sample representative of the United States each calendar year (see section 6.3.1 of \url{https://www.bls.gov/cex/pumd-getting-started-guide.htm#section3}). Additional household-level control variables include age (of the reference person), urban vs. rural status, taxable income, education level, and family type. Quarterly data from 1996 to 2021 on household gasoline consumption are from the interview module of the Consumer Expenditure Surveys (CEX) Public use microdata files. We can use Month $\times$ Year fixed effects because of the rotating nature of the survey. Gasoline tax data are obtained from the U.S. Department of Transportation, Federal Highway Administration, and gasoline prices are from the U.S. Energy Information Administration (State Energy Data System, SEDS) for the same period. Standard errors are clustered at the State $\times$ Year level. *** p$<$0.01, ** p$<$0.05, * p$<$0.1.
\end{minipage}
\end{table}

\subsection{Screenshot of Main Survey Elicitation Instrument}
\begin{figure}[H]
    \centering
    \includegraphics[width=0.92\linewidth]{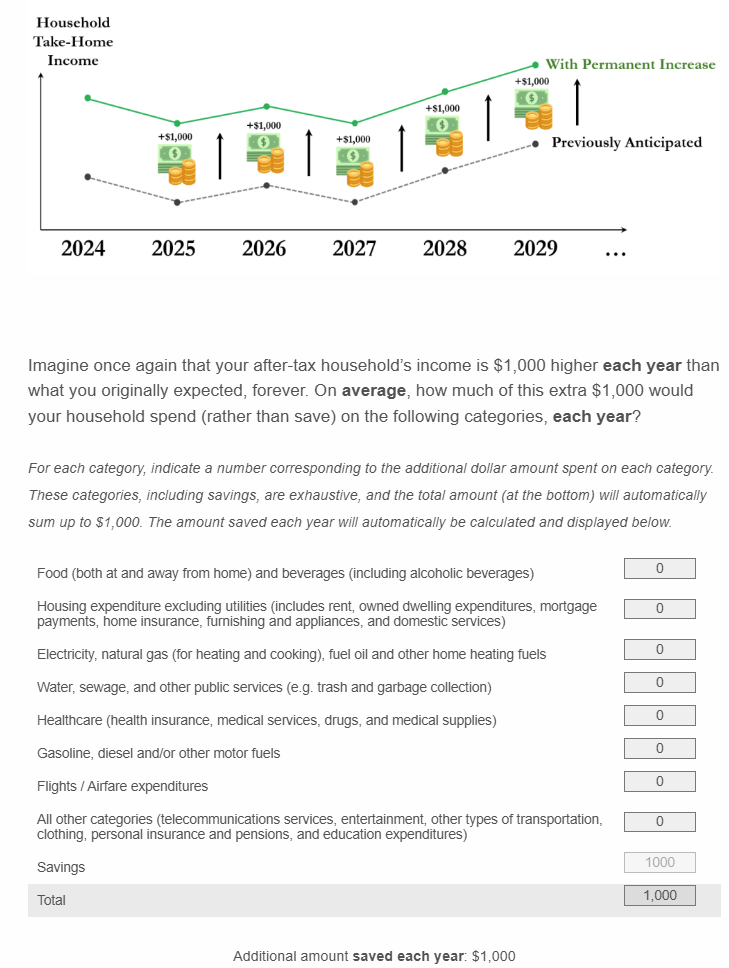}
    \caption{Survey Screen for the Causal Consumption Response elicitation block}
    \label{fig:q_causal_screenshot}
\end{figure}

\subsection{Status Quo Income Marginal Tax Rates}

\begin{figure}[H]
    \centering
    \includegraphics[width=0.7\linewidth]{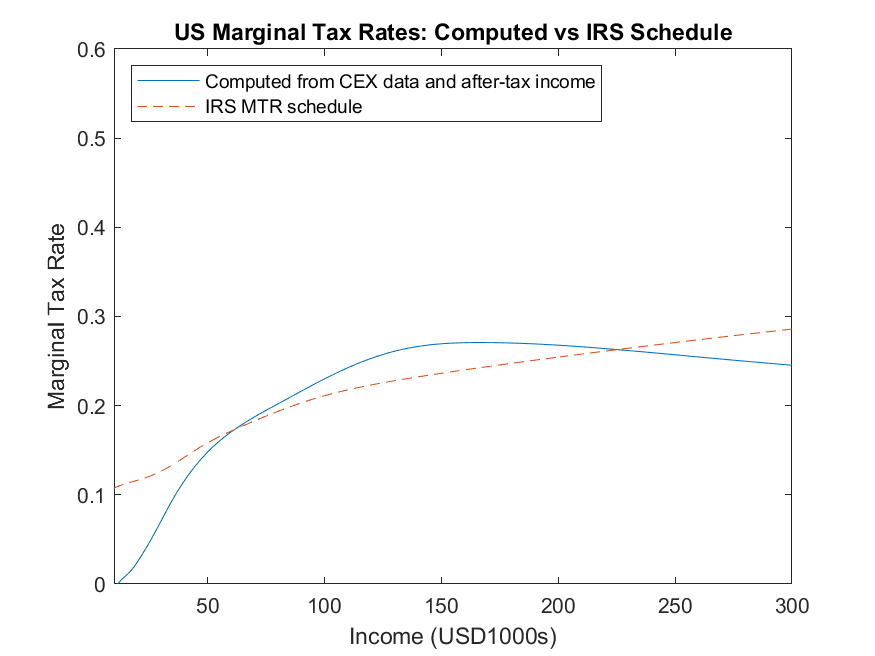}
    \caption{Status quo income marginal tax rates across the income distribution}
    \begin{minipage}{\textwidth}
    \scriptsize \textbf{Note:} 
    This figure shows the status quo income marginal tax rates across the household income distribution. The blue curve shows estimated marginal tax rates used in our analysis, and computed from CEX data (2018–2023, converted to 2024 USD) by fitting a smoothing spline to the relationship between after-tax and before-tax income binned by income percentiles, and computing 1 minus its derivative. The dashed orange line shows the 2024 IRS statutory tax rate schedule (adjusted to account for the share of married couples) for illustrative comparison. \end{minipage}\vspace{-0.5em}
    \label{fig:status_quo_mtr}
\end{figure}

\subsection{Pareto Efficient Marginal Tax Schedules - Gasoline Only}\label{app:empirical_pareto_efficient_tax_gasoline_only}

\begin{figure}[H]
    \centering
    \includegraphics[width=0.5\linewidth]{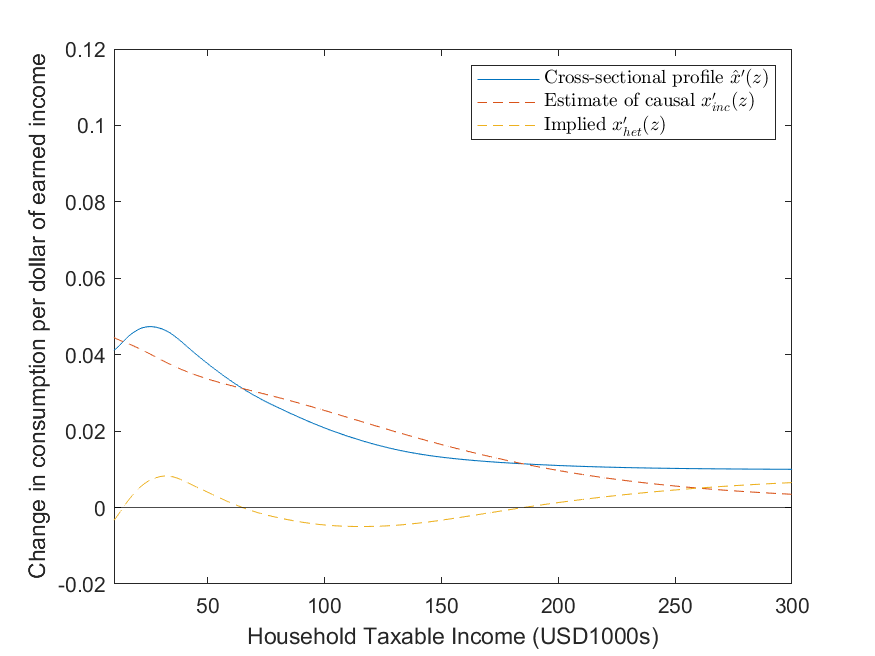}
    \caption{Gasoline: Decomposition of the cross-sectional profile $\hat x'(z)$}
    \begin{minipage}{\textwidth}
    \scriptsize \textbf{Note:} 
    This figure shows how we decompose the cross-sectional relationship between income and gasoline consumption. The blue line is a smoothed estimate of the cross-sectional slope (from the 2018-2023 CEX surveys), $\hat x'(z)$ at each level of income. The dashed red line shows estimates of the causal effect of income on expenditure on carbon-intensive goods (from our custom hypothetical survey). The yellow dashed line subtracts the income effect from the cross-section to obtain the part of the relationship between $x$ and $z$ that comes from preference heterogeneity.\end{minipage}\vspace{-0.5em}
    \label{fig:decomposition_gaso}
\end{figure}

\begin{figure}[H]
  \centering
  \includegraphics[width=0.6\linewidth]{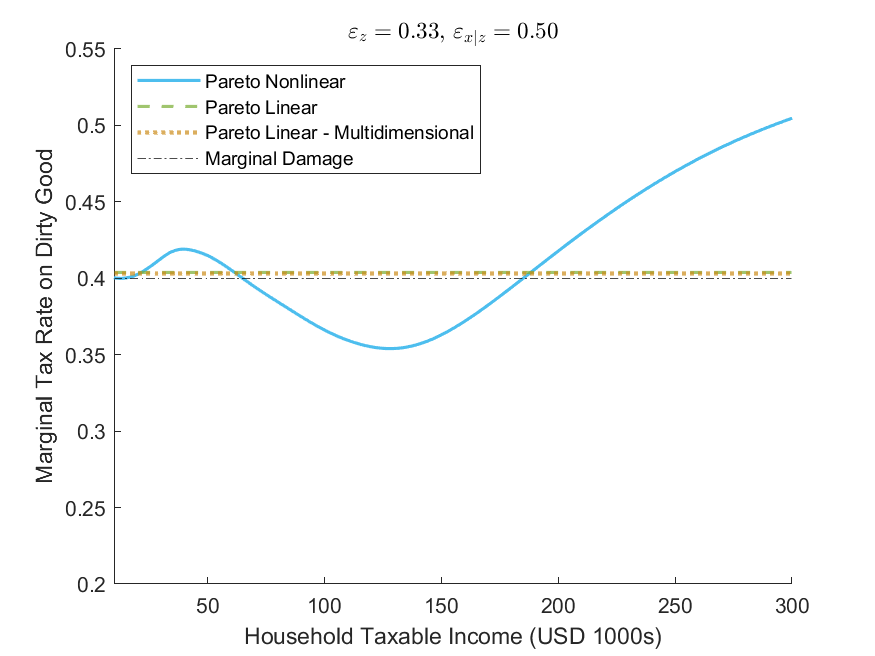}
  \caption{Gasoline Marginal Tax Rates:  Unidimensional vs. Multidimensional Heterogeneity}
    \begin{minipage}{\textwidth}
    \scriptsize \textbf{Note:} 
    This figure plots the Pareto efficient carbon tax for \textbf{gasoline only} in the unidimensional case (Linear and Nonlinear) vs. the quasi Pareto efficient linear tax rate when accounting for multidimensional heterogeneity, when the elasticity of taxable income is $\varepsilon_{z}=0.33$, and the elasticity of consumption of the carbon-intensive good $\varepsilon_{x\mid z}=0.5$.  \end{minipage}
    \label{fig:multidim_quasi_pareto_tax_schedule_vs_unidim_gaso}
\end{figure}

\subsection{Pareto Efficient Marginal Tax Schedules - Electricity, Natural Gas, and Heating Fuels Only}\label{app:empirical_pareto_efficient_tax_electricity_only}

\begin{figure}[H]
    \centering
    \includegraphics[width=0.5\linewidth]{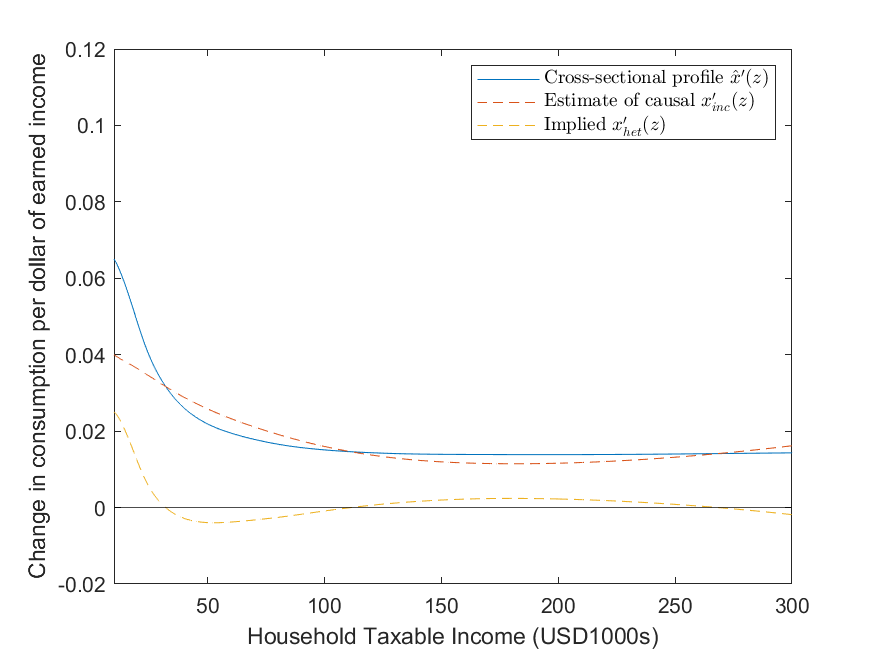}
    \caption{Electricity, Natural Gas, and Heating Fuels: Decomposition of the cross-sectional profile $\hat x'(z)$}
    \begin{minipage}{\textwidth}
    \scriptsize \textbf{Note:} 
    This figure shows how we decompose the cross-sectional relationship between income and electricity, natural gas, and heating fuels consumption. The blue line is a smoothed estimate of the cross-sectional slope (from the 2018-2023 CEX surveys), $\hat x'(z)$ at each level of income. The dashed red line shows estimates of the causal effect of income on expenditure on carbon-intensive goods (from our custom hypothetical survey). The yellow dashed line subtracts the income effect from the cross-section to obtain the part of the relationship between $x$ and $z$ that comes from preference heterogeneity.\end{minipage}\vspace{-0.5em}
    \label{fig:decomposition_elec}
\end{figure}

\begin{figure}[H]
  \centering
  \includegraphics[width=0.6\linewidth]{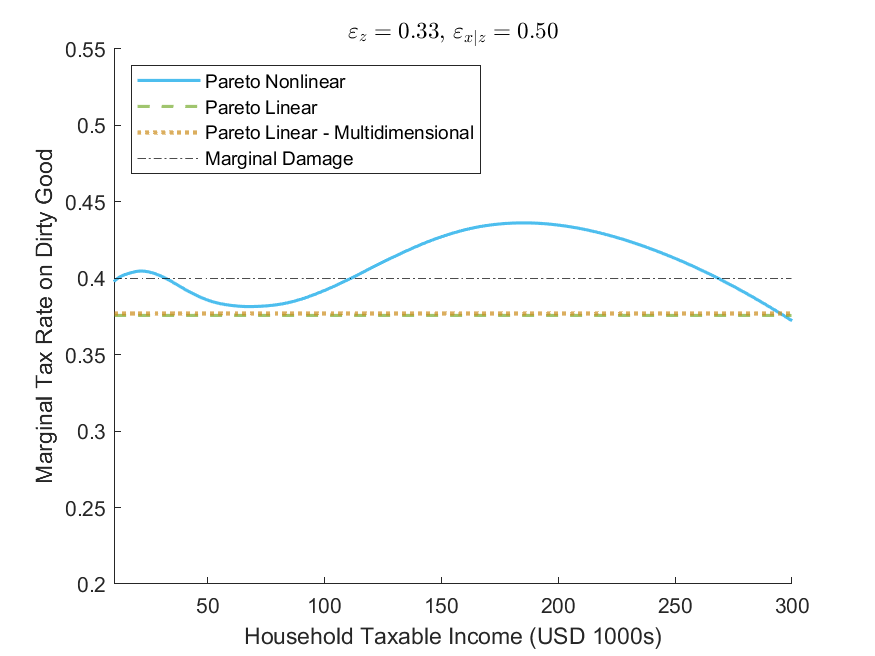}
  \caption{Electricity, Natural Gas, and Heating Fuels Marginal Tax Rates:  Unidimensional vs. Multidimensional Heterogeneity}
    \begin{minipage}{\textwidth}
    \scriptsize \textbf{Note:} 
    This figure plots the Pareto efficient carbon tax for \textbf{electricity, natural gas, and heating fuels consumption only} in the unidimensional case (Linear and Nonlinear) vs. the quasi Pareto efficient linear tax rate when accounting for multidimensional heterogeneity, when the elasticity of taxable income is $\varepsilon_{z}=0.33$, and the elasticity of consumption of the carbon-intensive good $\varepsilon_{x\mid z}=0.5$.  \end{minipage}
    \label{fig:multidim_quasi_pareto_tax_schedule_vs_unidim_elec}
\end{figure}

\subsection{Online Survey Demographics Breakdown \& Sample Definition}

\begin{table}[htbp]
\centering
\footnotesize
\begin{minipage}[t]{0.30\textwidth}
\vspace{0pt}
\centering

\begin{tabular}{lrr|r}
\toprule
Ethnicity & Freq. & \% & (C)\\
\midrule
White & 1097 & 75.8\% & 73.9\%\\
Black & 199 & 13.7\% & 14.3\%\\
Asian & 87 & 6.0\% & 7.5\%\\
Mixed & 43 & 3.0\% & 2.5\%\\
Other & 22 & 1.5\% & 1.7\%\\
\bottomrule
\end{tabular}

\end{minipage}
\hspace{0.05\textwidth}
\begin{minipage}[t]{0.30\textwidth}
\vspace{0pt}
\centering

\begin{tabular}{lrr|r}
\toprule
Age & Freq. & \% & (C)\\
\midrule
25-29 & 263 & 18.2\% & 16.9\%\\
30-34 & 276 & 19.1\% & 18.0\%\\
35-39 & 294 & 20.3\% & 17.3\%\\
40-44 & 196 & 13.5\% & 16.8\%\\
45-49 & 223 & 15.4\% & 15.2\%\\
50-54 & 196 & 13.5\% & 15.9\%\\
\bottomrule
\end{tabular}

\end{minipage}
\hspace{0.01\textwidth}
\begin{minipage}[t]{0.30\textwidth}
\vspace{0pt}
\centering

\begin{tabular}{lrr|r}
\toprule
Sex & Freq. & \% & (C)\\
\midrule
Female & 748 & 51.7\% & 49.7\%\\
Male & 700 & 48.3\% & 50.3\%\\
\bottomrule
\end{tabular}

\end{minipage}
\caption{Demographics Comparison: Online Survey vs 2023 Census Population Estimates}
\label{tab:demographics_comparison_survey_census2023}
\begin{minipage}{\textwidth}
        \scriptsize
        \textbf{Note}: This table shows the breakdown for our custom online survey sample by ethnicity, age, and sex. The rightmost (C) column in each subtable corresponds to the actual demographic distribution of the US population aged 25-54 drawn from 2023 Population Estimates (as of July 1) of the United States Census Bureau constructed with data from the 2020 Census and measures of population change. (See \texttt{https://www.census.gov/data/tables/time-series/demo/popest/2020s-national-detail.html}.) 
    \end{minipage}
\end{table}

\begin{table}[htbp]
\centering
\caption{Online Survey Response Quality: Sample Exclusions}
\label{tab:sample_exclusions}
\begin{tabular}{lcc}
\toprule
 & N & \% \\
\midrule
\textbf{Total \# of Respondents} &        2,219&      100.00\\
\hspace{0.5em}Suspected Bots&          18&        0.81\\
\hspace{0.5em}Speedsters  &         421&       18.97\\
\hspace{0.5em}Failed Attention Check&          44&        1.98\\
\hspace{0.5em}Inconsistent or Outlier Income Response&          55&        2.48\\
\hspace{0.5em}Zero Expenditure, Zero or Inconsistent MPC&         233&       10.50\\
\midrule
\textbf{Final Sample}&        1,448&       65.25\\
\bottomrule
\end{tabular}

\begin{minipage}{\textwidth}
        \scriptsize
        \textbf{Note}: This table reports the sequential exclusion of Prolific respondents from the combined Waves 1 and 2 of our online survey. Exclusions are applied in the order shown, as per our pre-registration, with each category excluding observations from the remaining sample after previous exclusions. Percentages are calculated relative to total observations.
    \end{minipage}
\end{table}

\renewcommand{\thefigure}{\thesection.\arabic{figure}}
\renewcommand{\thetable}{\thesection.\arabic{table}}

\allowdisplaybreaks

\clearpage
\section{Proofs of Theoretical Results}\label{app:proofs}

\begin{proof}[Proof of Proposition \ref{prop:uni-1}]

Consider the effects of an infinitesimal reform to the baseline
tax schedules $\mathcal{T}_{z}\left(\cdot\right)$ and $\mathcal{T}_{x}\left(\cdot\right)$, as described in Section \ref{subsec:optnonlin}.

\smallskip
\underline{Behavioral Responses}

Section \ref{sec:behav} introduces notation for the various behavioral responses, but we define them more formally here. First, the compensated net-of-tax rate elasticity of taxable income is:
\begin{equation}
    \varepsilon_{z}\left(w\right)\equiv-\frac{1-\mathcal{T}_{z}^{\prime}\left(z\left(w\right)\right)}{\tau_{z}^{\prime}\left(z\left(w\right)\right)z\left(w\right)}\left.\frac{\partial z\left(w\right)}{\partial\kappa}\right|_{\kappa=\tau_{z}\left(z\left(w\right)\right)+\tau_{x}\left(x\left(w\right)\right)=\tau_{x}^{\prime}\left(x\left(w\right)\right)=0}
\end{equation}
In words, this is the response of taxable income to a change to the marginal income tax rate without any change in the overall tax burden or the marginal tax rate on the externality-generating good.

Second, holding $z$ constant, the compensated net-of-tax rate elasticity of spending on the externality-generating good is:
\begin{equation}
    \varepsilon_{x|z}\left(w\right)\equiv-\frac{1+\mathcal{T}_{x}^{\prime}\left(x\left(w\right)\right)}{\tau_{x}^{\prime}\left(x\left(w\right)\right)x\left(w\right)}\left.\frac{\partial x\left(w;z\left(w\right)\right)}{\partial\kappa}\right|_{z\left(w\right),\kappa=\tau_{z}\left(z\left(w\right)\right)+\tau_{x}\left(x\left(w\right)\right)=\tau_{z}^{\prime}\left(z\left(w\right)\right)=0}.
\end{equation}
Our assumptions ensure that all of these responses are smooth.

\smallskip
\underline{A Reform To Test Pareto Efficiency}

To characterize the set of Pareto efficient tax systems, consider a distribution neutral tax reform as presented in Section \ref{sec:distneut}. The income tax is changed in a way that precisely offsets the impact of the change to the tax on the dirty good.
\vspace{-0.5em}\begin{equation}\vspace{-0.5em}
 \tau_{z}\left(z\right)=-\tau_{x}\left(\hat{x}\left(z\right)\right).   
\end{equation}
This implies marginal income tax rate changes of
\vspace{-0.5em}\begin{equation}\vspace{-0.5em}
\tau_{z}'\left(z\right)=-\hat{x}'\left(z\right)\tau_{x}'\left(\hat{x}\left(z\right)\right),
\end{equation}

\smallskip
\underline{Welfare Effects of a Tax Reform}

This distribution neutral reform has no mechanical effect on the utility level of any agent at any level of income. We argued in Section \ref{sec:welfeff} that this leaves us with three effects that are local to the level of income at which agents face a change in their marginal tax rate.
\begin{align*}
    \Delta\bm{\mathcal{W}}^{z}(z) &= \mathcal{T}_{z}'(z)\frac{z\varepsilon_{z}(z)}{1-\mathcal{T}_{z}'(z)}x_{\text{het}}^{\prime}(z)\\
    \Delta\bm{\mathcal{W}}^{x\mid z}(z) &= -\left(\mathcal{T}_{x}'(\hat{x}(z))-\frac{\mathcal{D}'(\bar{x})(z)}{\lambda}\right)\frac{\varepsilon_{x\mid z}(z)\,\hat{x}(z)}{1+\mathcal{T}_{x}'(\hat{x}(z))}\\
    \Delta\bm{\mathcal{W}}^{z\rightarrow x}(z) &= \left(\mathcal{T}_{x}'(\hat{x}(z))-\frac{\mathcal{D}'(\bar{x})}{\lambda}\right)x_{\text{inc}}^{\prime}(z)\frac{z\varepsilon_{z}(z)}{1-\mathcal{T}_{z}'(z)}x_{\text{het}}^{\prime}(z)
\end{align*}
where $\lambda$ is the multiplier on the planner's budget constraint, and
$\mathcal{D}^{\prime}\left(\bar{x}\right)/\lambda$ is the
marginal externality cost associated with consumption of $x$.

Integrating over all types in the economy and imposing the value of $\tau_{x}'(\hat{x}(z'))$ for our specific reform, the social welfare impact is:
\begin{equation}
\left.\frac{\partial\mathcal{L}}{\partial\kappa}\right|_{\kappa=0}=\int(\Delta\bm{\mathcal{W}}^{z}(z) + \Delta\bm{\mathcal{W}}^{x\mid z}(z) + \Delta\bm{\mathcal{W}}^{z\rightarrow x}(z))\tau_{x}\left(\hat{x}\left(z\right)\right)h_{z}\left(z\right)\text{d}z,\label{eq:DN_welfEffectApp}    
\end{equation}
where $h_{z}\left(z\right)$ is the distribution of taxable income.

\smallskip
\underline{A Necessary Condition for Optimality}

If the tax system in set optimally, it must be the case that any such infinitesimal tax reform has no first-order effect on welfare. That is to say, for any $\tau_{x}\left(\cdot\right)$
we must have $\left.\frac{\partial\mathcal{L}}{\partial\kappa}\right|_{\kappa=0}=0$.

By the fundamental lemma of the calculus of variations,\footnote{If for all continuous functions $\rho:\mathbb{R}\to\mathbb{R}$ we have $\int_{0}^{\infty}f(z)\rho(z)dz=0$, and if $f$ is continuous, then $f\left(z\right)=0$ for all $z$. Integrating over all types in the economy, the social welfare impact of a distribution-neutral reform is:
\vspace{-0.3em}\begin{equation}\vspace{-0.3em}
\left.\frac{\partial\mathcal{L}}{\partial\kappa}\right|_{\kappa=0}=\int(\Delta\bm{\mathcal{W}}^{z}(z) + \Delta\bm{\mathcal{W}}^{x\mid z}(z) + \Delta\bm{\mathcal{W}}^{z\rightarrow x}(z))\tau_{x}'(\hat{x}(z'))h_{z}\left(z\right)\text{d}z,\label{eq:DN_welfEffect}    
\end{equation}
where $h_{z}\left(z\right)$ is the density of the distribution of income at $z$. Letting $f\left(z\right)\equiv\int(\Delta\bm{\mathcal{W}}^{z}(z) + \Delta\bm{\mathcal{W}}^{x\mid z}(z) + \Delta\bm{\mathcal{W}}^{z\rightarrow x}(z))h_{z}\left(z\right)$
and $\rho\left(z\right)\equiv\tau_{x}'(\hat{x}(z))$. Equation \ref{eq:ParetoTaxSN} follows immediately.} we obtain equation \ref{eq:ParetoTaxSN}, which we restate here for convenience
\begin{equation*}
\frac{\mathcal{T}_{x}^{\prime}\left(\hat{x}\left(z\right)\right)-\frac{\mathcal{D}^\prime(\bar{x})}{\lambda}}{1+\mathcal{T}_{x}^{\prime}\left(\hat{x}\left(z\right)\right)}=\eta_{x,z}^{\text{Taste}}\left(z\right)\frac{\varepsilon_{z}\left(z\right)}{\varepsilon_{x|z}\left(z\right)}\left(\frac{\mathcal{T}_{z}^{\prime}\left(z\right)+\left[\mathcal{T}_{x}^{\prime}\left(\hat{x}\left(z\right)\right)-\frac{\mathcal{D}^\prime(\bar{x})}{\lambda}\right]x^\prime_{\text{inc}}(z)}{1-\mathcal{T}_{z}^{\prime}\left(z\right)}\right)
\end{equation*}
It is straightforward to then re-arrange this to obtain \ref{eq:ParetoTaxSN_r}, also restated:
\begin{equation*}
\mathcal{T}_{x}^{\prime}\left(\hat{x}\left(z\right)\right)-\frac{\mathcal{D}^{\prime}(\bar{x})}{\lambda}=\frac{\text{RR}(z)\mathcal{T}_{z}^{\prime}\left(z\right)}{1-\text{RR}(z)x_{\text{inc}}^{\prime}(z)},
\end{equation*}
where:
\begin{equation*}
\text{RR}(z)\equiv\eta_{x,z}^{\text{Taste}}\left(z\right)\frac{\varepsilon_{z}\left(z\right)}{\varepsilon_{x|z}\left(z\right)}\frac{1+\mathcal{T}_{x}^{\prime}\left(\hat{x}\left(z\right)\right)}{1-\mathcal{T}_{z}^{\prime}\left(z\right)}.
\end{equation*}

\smallskip
\underline{Pareto Inefficiency}

Because this reform holds each agent's total tax liability constant, it has no first-order impact on agents' private utility. Rather, if equation \ref{eq:changewelf} is not equal to zero, the planner can raise revenue by moving in one direction or the other ($\kappa>0$ or $\kappa<0$) with no impact on individual utility. The resulting surplus in government revenue can then be used to provide a lump sum transfer, leading to a Pareto improvement.
\end{proof}

\bigskip
\begin{proof}[Proof of Proposition \ref{prop:optlevels}]

Consider the effects of a more general infinitesimal reform to the baseline
tax schedules $\mathcal{T}_{z}\left(\cdot\right)$ and $\mathcal{T}_{x}\left(\cdot\right)$, as described in Section \ref{subsec:optnonlin}. We start introducing some additional notation, as well as characterizing the expanded set of behavioral responses to the reform, which our assumptions again ensure are smooth.

\smallskip
\underline{Behavioral Responses}

In addition to the behavioral responses defined in our proof of Proposition \ref{prop:uni-1}, we will need notation for several additional responses. First, the income effects on taxable income and consumption of $x$ are as follows:
\begin{align*}
\eta_{z}(w) & \equiv-\frac{\left.\frac{\partial z\left(w\right)}{\partial\kappa}\right|_{\kappa=\tau_{z}^{\prime}\left(z\left(w\right)\right)=\tau_{x}^{\prime}\left(x\left(w\right)\right)=0}}{\tau_{z}\left(z\left(w\right)\right)+\tau_{x}\left(x\left(w\right)\right)} & \eta_{x|z}(w) & \equiv-\frac{\left.\frac{\partial x\left(w;z\left(w\right)\right)}{\partial\kappa}\right|_{z\left(w\right),\kappa=\tau_{z}^{\prime}\left(z\left(w\right)\right)=\tau_{x}^{\prime}\left(x\left(w\right)\right)=0}}{\tau_{z}\left(z\left(w\right)\right)+\tau_{x}\left(x\left(w\right)\right)}.
\end{align*}

Using this notation, type $w$'s overall responses to a reform are
\begin{align}\label{eq:totalIncomeResp}
\left.\frac{\text{d}z\left(w\right)}{\text{d}\kappa}\right|_{\kappa=0}=&-\frac{\varepsilon_{z}\left(w\right)z\left(w\right)\tau_{z}^{\prime}\left(z\left(w\right)\right)}{1-\mathcal{T}_{z}^{\prime}\left(z\left(w\right)\right)}+\chi\left(z\left(w\right)\right)\tau_{x}^{\prime}\left(x\left(w;z\right)\right)\\\notag
&-\eta_{z}\left(w\right)\left[\tau_{z}\left(z\left(w\right)\right)+\tau_{x}\left(x\left(w;z\right)\right)\right]
\end{align}

\vspace{-3em}
\begin{align}\label{eq:totalCommodityResp}
\left.\frac{\text{d}x\left(w;z\left(w\right)\right)}{\text{d}\kappa}\right|_{\kappa=0}=&-\frac{\varepsilon_{x|z}\left(w\right)x\left(w\right)}{1+\mathcal{T}_{x}^{\prime}\left(x\left(w\right)\right)}\tau_{x}^{\prime}\left(x\left(w\right)\right)-\eta_{x|z}\left(w\right)\left[\tau_{z}\left(z\left(w\right)\right)+\tau_{x}\left(x\left(w\right)\right)\right]\\\notag
&+x^\prime_{\text{inc}}(z(w))\left.\frac{\text{d} z\left(w\right)}{\text{d}\kappa}\right|_{\kappa=0}
\end{align}
where $\varepsilon_{z}\left(w\right)$ and $\varepsilon_{x|z}$ are as defined above, and:
\[
\chi\left(w\right)\equiv\frac{1}{\tau_{x}^{\prime}\left(x\left(w\right)\right)}\left.\frac{\partial z\left(w\right)}{\partial\kappa}\right|_{\kappa=\tau_{z}\left(z\left(w\right)\right)+\tau_{x}\left(x\left(w\right)\right)=\tau_{z}^{\prime}\left(z\left(w\right)\right)=0}
\]
is the cross-base response of taxable income to the commodity tax
rate. 

Inserting equation \ref{eq:totalIncomeResp} into equation \ref{eq:totalCommodityResp},
we get:
\begin{align*}
\left.\frac{\text{d}x\left(w;z\left(w\right)\right)}{\text{d}\kappa}\right|_{\kappa=0}=&\left[-\frac{\varepsilon_{x|z}\left(w\right)x\left(w\right)}{1+\mathcal{T}_{x}^{\prime}\left(x\left(w\right)\right)}+x^\prime_{\text{inc}}(z(w))\chi\left(z\left(w\right)\right)\right]\tau_{x}^{\prime}\left(x\left(w\right)\right)\\
&-x^\prime_{\text{inc}}(z(w))\frac{\varepsilon_{z}\left(w\right)z\left(w\right)}{1-\mathcal{T}_{z}^{\prime}\left(z\left(w\right)\right)}\tau_{z}^{\prime}\left(z\left(w\right)\right)\\
&-\left[\eta_{x|z}\left(w\right)+x^\prime_{\text{inc}}(z(w))\eta_{z}\left(w\right)\right]\left[\tau_{z}\left(z\left(w\right)\right)+\tau_{x}\left(x\left(w\right)\right)\right].
\end{align*}
By Slutsky symmetry, we have the following expression
for the partial effect of taxable income on the choice of the externality-generating good:
\begin{equation*}
x^\prime_{\text{inc}}(z(w)) =\frac{x\left(w\right)}{z\left(w\right)}\left[\frac{-\frac{1-\mathcal{T}_{z}^{\prime}\left(z\left(w\right)\right)}{x\left(w\right)}\chi\left(w\right)}{\varepsilon_{z}\left(w\right)}\right]=-\frac{1-\mathcal{T}_{z}^{\prime}\left(z\left(w\right)\right)}{z\left(w\right)}\frac{\chi\left(w\right)}{\varepsilon_{z}\left(w\right)}.
\end{equation*}
Put another way, we can write the cross-base response in terms of
this partial effect:
\[
\chi\left(w\right)=-\frac{z\left(w\right)\varepsilon_{z}\left(w\right)}{1-\mathcal{T}_{z}^{\prime}\left(z\left(w\right)\right)}x^\prime_{\text{inc}}(z(w)).
\]
Thus, the income response (equation \ref{eq:totalIncomeResp}) can be rewritten as
\begin{align}\label{eq:totalIncomeResp-1}
\left.\frac{\text{d}z\left(w\right)}{\text{d}\kappa}\right|_{\kappa=0}=&-\frac{z\left(w\right)\varepsilon_{z}\left(w\right)}{1-\mathcal{T}_{z}^{\prime}\left(z\left(w\right)\right)}\tau_{z}^{\prime}\left(z\left(w\right)\right)-\frac{z\left(w\right)\varepsilon_{z}\left(w\right)}{1-\mathcal{T}_{z}^{\prime}\left(z\left(w\right)\right)}x^\prime_{\text{inc}}(z(w))\tau_{x}^{\prime}\left(x\left(w\right)\right)\\\notag
&-\eta_{z}\left(w\right)\left[\tau_{z}\left(z\left(w\right)\right)+\tau_{x}\left(x\left(w\right)\right)\right],
\end{align}
and the response of $x$ (equation \ref{eq:totalCommodityResp}) can written as:
\begin{align}\label{eq:totalCommodityResp2}
\left.\frac{\text{d}x\left(w;z\left(w\right)\right)}{\text{d}\kappa}\right|_{\kappa=0}=&-\left[\frac{\varepsilon_{x|z}\left(w\right)x\left(w;z\left(w\right)\right)}{1+\mathcal{T}_{x}^{\prime}\left(x\left(w\right)\right)}+\left(x^\prime_{\text{inc}}(z(w))\right)^{2}\frac{z(w)\varepsilon_{z}(w)}{1-\mathcal{T}_{z}^{\prime}\left(z\left(w\right)\right)}\right]\tau_{x}^{\prime}(x(w))\\\notag
&-x^\prime_{\text{inc}}(z(w))\frac{z\left(w\right)\varepsilon_{z}\left(w\right)}{1-\mathcal{T}_{z}^{\prime}\left(z\left(w\right)\right)}\tau_{z}^{\prime}\left(z(w)\right)\\\notag
&-\left[\eta_{x|z}\left(w\right)+x^\prime_{\text{inc}}(z(w))\eta_{z}\left(w\right)\right]\left[\tau_{z}\left(z(w)\right)+\tau_{x}(x(w))\right].
\end{align}

\smallskip
\underline{Welfare Effects of a Tax Reform}

For any given pair of continuously differentiable reform functions $\tau_{z}$ and $\tau_{x}$, an infinitesimal reform away from $\mathcal{T}_{z}$
and $\mathcal{T}_{x}$ has the following effect on social welfare:
{\small{}\allowdisplaybreaks[1] 
\begin{align}\label{eq:welfEffect}
\frac{1}{\lambda}&\left.\frac{\partial\mathcal{W}}{\partial\kappa}\right|_{\kappa=0}+\left.\frac{\partial\mathcal{R}}{\partial\kappa}\right|_{\kappa=0}  \\
&=\int\left(1-g\left(w\right)\right)\left[\tau_{z}\left(z\left(w\right)\right)+\tau_{x}\left(x\left(w\right)\right)\right]f\left(w\right)\text{d}w \tag{A} \\
 & -\int\left(\frac{\mathcal{T}_{z}^{\prime}\left(z\left(w\right)\right)+\left[\mathcal{T}_{x}^{\prime}\left(x\left(w\right)\right)-\frac{\mathcal{D}^{\prime}\left(\bar{x}\right)}{\lambda}\right]x^\prime_{\text{inc}}(z(w))}{1-\mathcal{T}_{z}^{\prime}\left(z\left(w\right)\right)}\right)\varepsilon_{z}\left(w\right)z\left(w\right)\tau_{z}^{\prime}\left(z\left(w\right)\right)f\left(w\right)\text{d}w\tag{B} \\
 & -\int\left\{ \left(\frac{\mathcal{T}_{z}^{\prime}\left(z\left(w\right)\right)+\left[\mathcal{T}_{x}^{\prime}\left(x\left(w\right)\right)-\frac{\mathcal{D}^{\prime}\left(\bar{x}\right)}{\lambda}\right]x^\prime_{\text{inc}}(z(w))}{1-\mathcal{T}_{z}^{\prime}\left(z\left(w\right)\right)}\right)x^\prime_{\text{inc}}(z(w))\varepsilon_{z}\left(w\right)z\left(w\right)\right.\dots\notag \\
 & \qquad\qquad\qquad\qquad\quad\ \dots+\left.\frac{\mathcal{T}_{x}^{\prime}\left(x\left(w\right)\right)-\frac{\mathcal{D}^{\prime}\left(\bar{x}\right)}{\lambda}}{1+\mathcal{T}_{x}^{\prime}\left(x\left(w\right)\right)}\varepsilon_{x|z}\left(w\right)x\left(w\right)\right\} \tau_{x}^{\prime}\left(x\left(w\right)\right)f\left(w\right)\text{d}w\tag{C}
\end{align}
}where $\lambda$ is the multiplier on the planner's budget constraint, and
$\mathcal{D}^{\prime}\left(\bar{x}\right)/\lambda$ is the
marginal externality cost associated with consumption of $x$.

The first line (A) captures the social welfare benefit of any redistribution from the reform. This includes differences in income effects and externality impacts, by the definition of $g(w)$. The second line (B) captures the fiscal externality via the income tax. The third line (C) captures the fiscal externality via the tax on $x$.

\smallskip
\underline{Optimality Condition for $T_z$}

If the tax system in set optimally, it must be the case that any such infinitesimal tax reform has no first-order effect on welfare. That is to say, for any $\left(\tau_{z}\left(\cdot\right),\tau_{x}\left(\cdot\right)\right)$
we must have $\frac{1}{\lambda}\left.\partial\mathcal{W}/\partial\kappa\right|_{\kappa=0}+\left.\partial\mathcal{R}/\partial\kappa\right|_{\kappa=0}=0$. To characterize the optimal income tax, we can therefore consider a reform only to the income tax by supposing that $\tau_{x}\left(x\right)=0$ for all
$x$, and then set equation \ref{eq:welfEffect} equal to zero.

With a change in variables, we can write the resulting optimality condition in terms of the income distribution. For all $\tau_z(\cdot)$:
\begin{align*}
&\int\left(1-g\left(\hat{w}\left(z\right)\right)\right)\tau_{z}\left(z\right)h_{z}\left(z\right)\text{d}z\\
&\quad-\int\left[\frac{\mathcal{T}_{z}^{\prime}\left(z\right)+\left[\mathcal{T}_{x}^{\prime}\left(x\left(\hat{w}\left(z\right)\right)\right)-\frac{\mathcal{D}^{\prime}\left(\bar{x}\right)}{\lambda}\right]x^\prime_{\text{inc}}(z(w))}{1-\mathcal{T}_{z}^{\prime}\left(z\right)}\right]\varepsilon_{z}\left(\hat{w}\left(z\right)\right)z\tau_{z}^{\prime}\left(z\right)h_{z}\left(z\right)\text{d}z=0.
\end{align*}

By the fundamental lemma of the calculus of variations as we did for Proposition \ref{prop:uni-1}, we can then also write the following condition
characterizing optimal marginal rates on income in terms of the distribution
of income along with various other sufficient statistics:\footnote{Here we use a slightly different variant of fundamental lemma of the calculus of variations, which states that, for all continuously differentiable functions $h : \mathbb{R}_+ \to \mathbb{R}$ we have
\begin{equation}
\int_0^{\infty} f(x) h(x) dx + \int_0^{\infty} g(x) h'(x) dx = 0,
\end{equation}
and if $f$ and $g$ are continuous, then $g'(x) = f(x) \quad \forall x \geq 0$, $g$ is continuously differentiable, and $\lim_{x \to \infty} g(x) = \lim_{x \to 0} g(x) = 0$. Our application of it here follows \citet{jacquet2021optimal}.}
\begin{equation}
\frac{\mathcal{T}_{z}^{\prime}\left(z\right)+\left[\mathcal{T}_{x}^{\prime}\left(x\left(\hat{w}\left(z\right)\right)\right)-\frac{\mathcal{D}^{\prime}\left(\bar{x}\right)}{\lambda}\right]x^\prime_{\text{inc}}(z(w))}{1-\mathcal{T}_{z}^{\prime}\left(z\right)}=\frac{\int_{z}^{\infty}\left(1-g\left(\hat{w}\left(s\right)\right)\right)h_{z}\left(s\right)\text{d}s}{\varepsilon_{z}\left(\hat{w}\left(z\right)\right)zh_{z}\left(z\right)}.\label{eq:eq:optCondForT_z}
\end{equation}
Defining the average welfare weight for agents with income greater than $z$ as:
\begin{equation}
    1-\bar{g}_{+}\left(z\right) = \frac{\int_{z}^{\infty}\left(1-g\left(\hat{w}\left(s\right)\right)\right)h_{z}\left(s\right)\text{d}s}{1-H_z(z)},
\end{equation}
we obtain the optimality condition stated in Proposition \ref{prop:optlevels}.

\smallskip
\underline{Optimality Condition for $T_x$}

Next, to characterize the optimal tax on the externality-generating good, we consider a reform only to the tax on $x$ by supposing that $\tau_{z}\left(z\right)=0$ for all
$z$, and set equation \ref{eq:welfEffect} to zero. Letting $H_x(x)$ be the distribution of $x$, with density $h_x(x)$,  we can write the resulting optimality condition in terms of the distribution of consumption of $x$. For all $\tau_x(\cdot)$:
\begin{align*}
&\int\left(1-g(\tilde{w}(x))\right)\tau_{x}\left(x\right)h_{x}\left(x\right)\text{d}x-\int\left[\frac{\mathcal{T}_{x}^{\prime}\left(x\right)-\frac{\mathcal{D}^{\prime}\left(\bar{x}\right)}{\lambda}}{1+\mathcal{T}_{x}^{\prime}\left(s\right)}\varepsilon_{x|z}\left(\tilde{w}\left(x\right)\right)x\right.\dots\\
&+\hspace{-0.3em}\left.\left(\frac{\mathcal{T}_{z}^{\prime}(\tilde{z}(x))+\left[\mathcal{T}_{x}^{\prime}(x)-\frac{\mathcal{D}^{\prime}(\bar{x})}{\lambda}\right]x^\prime_{\text{inc}}(z(w))}{1-\mathcal{T}_{z}^{\prime}\left(z\right)}\right)x^\prime_{\text{inc}}(z(w))\varepsilon_{z}(\tilde{w}(x))\tilde{z}(x)\right]\tau_{x}^{\prime}(x)h_{x}(x)\text{d}x=0,
\end{align*}
where $\tilde{w}\left(x\right)$ and $\tilde{z}\left(x\right)$ are the values of $w$ and $z$ which correspond to each value of $x$.

Similar to the case of the income tax, the fundamental lemma of the calculus of variations then lets us write:
\begin{align}\label{eq:optCondForT_x}
\frac{\mathcal{T}_{x}^{\prime}\left(x\right)-\frac{\mathcal{D}^{\prime}\left(\bar{x}\right)}{\lambda}}{1+\mathcal{T}_{x}^{\prime}\left(x\right)}&\varepsilon_{x|z}\left(\tilde{w}\left(x\right)\right)xh_{x}\left(x\right)\\\notag
&+x^\prime_{\text{inc}}(z(w))\left(\frac{\mathcal{T}_{z}^{\prime}(\tilde{z}(x))+\left[\mathcal{T}_{x}^{\prime}(x)-\frac{\mathcal{D}^{\prime}(\bar{x})}{\lambda}\right]x^\prime_{\text{inc}}(z(w))}{1-\mathcal{T}_{z}^{\prime}\left(z\right)}\right)\varepsilon_{z}({w}(x))\tilde{z}(x)h_{x}(x)\\
&=\int_{x}^{\infty}\left(1-g\left(\tilde{w}\left(y\right)\right)\right)h_{x}\left(y\right)d\text{y}.\notag
\end{align}

We can then write equation \ref{eq:optCondForT_x} in terms of the distribution
of taxable income. To do so, we use the fact that:
\begin{align*}
h_{z}\left(z\right) & =\frac{\text{d}x\left(\hat{w}\left(z\right);z\right)}{\text{d}z}h_{x}\left(x\left(\hat{w}\left(z\right);z\right)\right)\\
 & =\left[x^\prime_{\text{inc}}(z(w))+x^\prime_{\text{het}}(z(w))\right]h_{x}\left(x\left(\hat{w}\left(z\right);z\right)\right).
\end{align*}
Using this result, multiplying both sides of equation \ref{eq:optCondForT_x} by $x^\prime_{\text{inc}}(z(w))+x^\prime_{\text{het}}(z(w))$,
and inserting the optimality condition for the income tax (equation \ref{eq:eq:optCondForT_z}) we can rewrite equation \ref{eq:optCondForT_x} as:
\begin{multline*}
\frac{\mathcal{T}_{x}^{\prime}\left(\hat{x}\left(z\right)\right)-\frac{\mathcal{D}^{\prime}\left(\bar{x}\right)}{\lambda}}{1+\mathcal{T}_{x}^{\prime}\left(\hat{x}\left(z\right)\right)}\varepsilon_{x|z}\left(\hat{w}\left(z\right)\right)\hat{x}\left(z\right)h_{z}\left(z\right)
+x^\prime_{\text{inc}}(z(w))\int_{z}^{\infty}\left(1-g\left(\hat{w}\left(s\right)\right)\right)h_{z}\left(s\right)\text{d}s\\
=\left(x^\prime_{\text{inc}}(z(w))+x^\prime_{\text{het}}(z(w))\right)\int_{z}^{\infty}\left(1-g\left(\hat{w}\left(s\right)\right)\right)h_{z}\left(s\right)\text{d}s,
\end{multline*}
which further simplifies to
\begin{equation}
\frac{\mathcal{T}_{x}^{\prime}\left(\hat{x}\left(z\right)\right)-\frac{\mathcal{D}^{\prime}\left(\bar{x}\right)}{\lambda}}{1+\mathcal{T}_{x}^{\prime}\left(\hat{x}\left(z\right)\right)}=x^\prime_{\text{het}}(z(w))\frac{\int_{z}^{\infty}\left(1-g\left(\hat{w}\left(s\right)\right)\right)h_{z}\left(s\right)\text{d}s}{\varepsilon_{x|z}\left(\hat{w}\left(z\right)\right)\hat{x}\left(z\right)h_{z}\left(z\right)}.\label{eq:FLT_optTaxCond}
\end{equation}
Using the definitions of $\eta_{x,z}^{\text{Taste}}\left(z\right)$ and $\bar{g}_{+}\left(z\right)$, we obtain the optimality condition stated in the lemma.\end{proof}

\medskip
\begin{proof}[Proof of Proposition \ref{prop:lin}]
The proof here proceeds in a very similar manner to fully non-linear case. Each agent now maximizes her utility subject to her slightly modified budget constraint, producing indirect utility $v(w)$. 
\begin{align}
    v(w) &= \max_{z,x,c} u(c,x,z \mid w)\\
         &\text{s.t. } \quad c+x+t_x x = z - {T}_z(z) \notag
\end{align}
Here, $t_x$ is a linear tax on $x$.

A welfarist social planner chooses $t_x$ and a twice-differentiable $\mathcal{T}_z$ to maximize social welfare while raising enough revenue to cover an exogenous revenue requirement, $\mathcal{R}$.
\begin{align}
    \max \bm{\mathcal{W}} &= \int \gamma(w) v(w)dF - \mathcal{D}(\overline{x})\\
    &\text{s.t. } \quad \mathcal{R} = \int \left({T}_z(z(w))+t_x x\right)dF(w)\\
    &\text{and } \quad \overline{x} = \int x(w) dF \notag
\end{align}
Mathematically, everything remains identical to the proof of Proposition \ref{prop:uni-1} up until the expression for the general impact on welfare of a tax reform (equation \ref{eq:welfEffect}).

Imposing that the initial tax on $x$ is linear yields equation \ref{eq:welfEffectLin}.
{\allowdisplaybreaks[1] 
\begin{align}\label{eq:welfEffectLin}
\frac{1}{\lambda}&\left.\frac{\partial\mathcal{W}}{\partial\kappa}\right|_{\kappa=0}+\left.\frac{\partial\mathcal{R}}{\partial\kappa}\right|_{\kappa=0}  \\
&=\int\left(1-g\left(w\right)\right)\left[\tau_{z}\left(z\left(w\right)\right)+\tau_{x}\left(x\left(w\right)\right)\right]f\left(w\right)\text{d}w \tag{A} \\
 & -\int\left(\frac{\mathcal{T}_{z}^{\prime}\left(z\left(w\right)\right)+\left[t_x-\frac{\mathcal{D}^{\prime}\left(\bar{x}\right)}{\lambda}\right]x^\prime_{\text{inc}}(z(w))}{1-\mathcal{T}_{z}^{\prime}\left(z\left(w\right)\right)}\right)\varepsilon_{z}\left(w\right)z\left(w\right)\tau_{z}^{\prime}\left(z\left(w\right)\right)f\left(w\right)\text{d}w\tag{B} \\
 & -\int\left\{ \left(\frac{\mathcal{T}_{z}^{\prime}\left(z\left(w\right)\right)+\left[t_x-\frac{\mathcal{D}^{\prime}\left(\bar{x}\right)}{\lambda}\right]x^\prime_{\text{inc}}(z(w))}{1-\mathcal{T}_{z}^{\prime}\left(z\left(w\right)\right)}\right)x^\prime_{\text{inc}}(z(w))\varepsilon_{z}\left(w\right)z\left(w\right)\right.\dots\notag \\
 & \qquad\qquad\qquad\qquad\quad\ \dots+\left.\frac{t_x-\frac{\mathcal{D}^{\prime}\left(\bar{x}\right)}{\lambda}}{1+t_x}\varepsilon_{x|z}\left(w\right)x\left(w\right)\right\} \tau_{x}^{\prime}\left(x\left(w\right)\right)f\left(w\right)\text{d}w\tag{C}
\end{align}
From here, we proceed similarly to the proof of Proposition \ref{prop:uni-1}.}

However, the reform we consider must now be linear as well. Specifically, consider a joint reform where: $\tau_{x}\left(x\left(\hat{w}\left(z\right);z\right)\right)=\rho \times x(z)$ for some constant $\rho$; and $\tau_{z}\left(z\right)=-\rho \times x(z)$. For such a reform, we have $\tau_{x}^{\prime}\left(x\left(\hat{w}\left(z\right);z\right)\right)=\rho$
and:
\begin{equation}
\tau_{z}^{\prime}\left(z\right)=-\rho \left[x^\prime_{\text{inc}}(z(w))+x^\prime_{\text{het}}(z(w))\right].\label{eq:distNeutReformLin}
\end{equation}

Combining that with equation \ref{eq:welfEffectLin}, the impact on welfare is proportional to:
{\allowdisplaybreaks[1] 
\begin{align}
 &\int\left(\frac{\mathcal{T}_{z}^{\prime}\left(z\left(w\right)\right)+\left[t_x-\frac{\mathcal{D}^{\prime}\left(\bar{x}\right)}{\lambda}\right]x^\prime_{\text{inc}}(z(w))}{1-\mathcal{T}_{z}^{\prime}\left(z\left(w\right)\right)}\right)\varepsilon_{z}\left(w\right)z\left(w\right) [x^\prime_{\text{inc}}(z(w))+x^\prime_{\text{het}}(z(w))]f\left(w\right)\text{d}w\notag \\
 & -\int\left\{ \left(\frac{\mathcal{T}_{z}^{\prime}\left(z\left(w\right)\right)+\left[t_x-\frac{\mathcal{D}^{\prime}\left(\bar{x}\right)}{\lambda}\right]x^\prime_{\text{inc}}(z(w))}{1-\mathcal{T}_{z}^{\prime}\left(z\left(w\right)\right)}\right)x^\prime_{\text{inc}}(z(w))\varepsilon_{z}\left(w\right)z\left(w\right)\right.\dots\notag \\
 & \qquad\qquad\qquad\qquad\quad\ \dots+\left.\frac{t_x-\frac{\mathcal{D}^{\prime}\left(\bar{x}\right)}{\lambda}}{1+t_x}\varepsilon_{x|z}\left(w\right)x\left(w\right)\right\}  f\left(w\right)\text{d}w.
\end{align}
In this case, this simplifies to:
\begin{align}
 &\int\left(\frac{\mathcal{T}_{z}^{\prime}\left(z\left(w\right)\right)+\left[t_x-\frac{\mathcal{D}^{\prime}\left(\bar{x}\right)}{\lambda}\right]x^\prime_{\text{inc}}(z(w))}{1-\mathcal{T}_{z}^{\prime}\left(z\left(w\right)\right)}\right)\varepsilon_{z}\left(w\right)z\left(w\right) x^\prime_{\text{het}}(z(w))\cdots\notag \\
 & \quad\quad\quad\quad\quad\quad\quad\quad\quad \dots -\left\{  \frac{t_x-\frac{\mathcal{D}^{\prime}\left(\bar{x}\right)}{\lambda}}{1+t_x}\varepsilon_{x|z}\left(w\right)x\left(w\right)\right\}  f\left(w\right)\text{d}w.
\end{align}
To obtain the equation in the proposition, simply set this to zero to yield an optimality condition. Then replace the integrals with expectations over the income distribution, and impose the definition of $\text{RR}(z)$.}

By the same logic as before, there is no mechanical redistribution from this reform. It only impacts social welfare by changing government revenue or the externality. Thus, if this reform has a positive impact on social welfare, it is a Pareto improvement.
\end{proof}

\begin{proof}[Proof of Proposition \ref{prop:multidim}]

\smallskip
\underline{Setup and Behavioral Responses}

Consider agents characterized by type $(w,\theta)$ where $w \in W \subseteq \mathbb{R}_{++}$ captures productivity and $\theta \in \Theta$ captures other preference heterogeneity. Each agent solves:
\begin{align*}
\max_{z,x,c} \quad & u(c,x,z;w,\theta) \\
\text{s.t.} \quad & c + x + \mathcal{T}_x(x) = z - \mathcal{T}_z(z)
\end{align*}
Define indirect utility $v(w,\theta)$ and choices $z(w,\theta)$, $x(w,\theta)$. The first-order conditions yield:
\begin{align*}
1 + \mathcal{T}_x'(x) &= \text{MRS}_x(c,x,z;w,\theta), \\
1 - \mathcal{T}_z'(z) &= \text{MRS}_z(c,x,z;w,\theta).
\end{align*}
For a reform parameterized by $\kappa$ with $T_z(z;\kappa) = \mathcal{T}_z(z) + \kappa\tau_z(z)$ and $T_x(x;\kappa) = \mathcal{T}_x(x) + \kappa\tau_x(x)$, implicit differentiation of the first-order conditions yields the total behavioral responses. Decomposing into substitution and income effects:
\begin{align}
\left.\frac{dz(w,\theta)}{d\kappa}\right|_{\kappa=0} &= -\frac{\varepsilon_z(w,\theta)z(w,\theta)\tau_z'(z(w,\theta))}{1-\mathcal{T}_z'(z(w,\theta))} + \chi(w,\theta)\tau_x'(x(w,\theta)) \notag \\
& \quad - \eta_z(w,\theta)[\tau_z(z(w,\theta)) + \tau_x(x(w,\theta))], \label{eq:proof_z_response}
\end{align}
where $\varepsilon_z(w,\theta)$ is the compensated elasticity of taxable income, $\chi(w,\theta)$ is the cross-base response, and $\eta_z(w,\theta)$ is the income effect.

\smallskip
\underline{Slutsky Symmetry}

By the symmetry of the Slutsky matrix, the cross-price effect of the commodity tax on labor supply equals the cross-price effect of the wage on commodity demand. This yields:
\begin{equation}
\chi(w,\theta) = -\frac{z(w,\theta)\varepsilon_z(w,\theta)}{1-\mathcal{T}_z'(z(w,\theta))} \frac{\partial x(w,\theta;z(w,\theta))}{\partial z}. \label{eq:proof_slutsky}
\end{equation}
Substituting \eqref{eq:proof_slutsky} into \eqref{eq:proof_z_response} and similarly deriving $dx/d\kappa$:
\begin{align*}
\left.\frac{dx(w,\theta)}{d\kappa}\right|_{\kappa=0} &= -\frac{\varepsilon_{x|z}(w,\theta)x(w,\theta)}{1+\mathcal{T}_x'(x(w,\theta))}\tau_x'(x(w,\theta)) \\
& \quad - \left[\eta_{x|z}(w,\theta) + \frac{\partial x}{\partial z}\eta_z(w,\theta)\right][\tau_z(z) + \tau_x(x)] + \frac{\partial x}{\partial z}\left.\frac{dz}{d\kappa}\right|_{\kappa=0}.
\end{align*}

\smallskip
\underline{Welfare Effects of a Reform}

Social welfare is $\mathcal{W} = \int_\Theta \int_W \gamma(w,\theta)v(w,\theta)f_w(w|\theta)dw\,\mu(d\theta)$. By the envelope theorem:
\[
\left.\frac{dv(w,\theta)}{d\kappa}\right|_{\kappa=0} = -\frac{\partial v}{\partial c}[\tau_z(z(w,\theta)) + \tau_x(x(w,\theta))].
\]
Total revenue is $\mathcal{R} = \int_\Theta \int_W [T_z(z;\kappa) + T_x(x;\kappa)]f_w(w|\theta)dw\,\mu(d\theta)$. The Lagrangian $\mathcal{L} = \mathcal{W}/\lambda + \mathcal{R}$ responds to the reform as:
\begin{align*}
\left.\frac{d\mathcal{L}}{d\kappa}\right|_{\kappa=0} &= \underbrace{\int_\Theta \int_W (1-g(w,\theta))[\tau_z(z) + \tau_x(x)]f_w dw\,\mu(d\theta)}_{\text{Mechanical effect net of welfare loss}} \\
& \quad + \underbrace{\int_\Theta \int_W \mathcal{T}_z'(z)\left.\frac{dz}{d\kappa}\right|_{\kappa=0} f_w dw\,\mu(d\theta)}_{\text{Behavioral effect on income tax revenue}} \\
& \quad + \underbrace{\int_\Theta \int_W \mathcal{T}_x'(x)\left.\frac{dx}{d\kappa}\right|_{\kappa=0} f_w dw\,\mu(d\theta)}_{\text{Behavioral effect on commodity tax revenue}},
\end{align*}
where the social marginal utility of income is:
\[
g(w,\theta) \equiv \frac{\gamma(w,\theta)}{\lambda}\frac{\partial v}{\partial c} + \mathcal{T}_z'(z)\eta_z(w,\theta) + \mathcal{T}_x'(x)\left[\eta_{x|z}(w,\theta) + \frac{\partial x}{\partial z}\eta_z(w,\theta)\right].
\]

\smallskip
\underline{Change of Variables to $(z,x)$}

Assume that, conditional on $\theta$, $z(w,\theta)$ is continuous and monotonic in $w$. This ensures a continuous marginal distribution of $z$ with density $h_z(z)$. However, the conditional distribution of $x$ given $z$ may have mass points, so we work with the CDF $H_x(x|z) = P(x(w,\theta) \leq x \mid z(w,\theta) = z)$ and corresponding measure $\nu_x(A|z) = \int_{x \in A} dH_x(x|z)$.

Define $\bar{\ell}(z,x) \equiv \mathbb{E}[\ell(w,\theta) \mid z(w,\theta)=z, x(w,\theta)=x]$ for $\ell \in \{g, \varepsilon_z, \varepsilon_{x|z}\}$. The welfare effect becomes:
\begin{align*}
\left.\frac{d\mathcal{L}}{d\kappa}\right|_{\kappa=0} &= \int_{\mathbb{R}_{++}} \int_{\mathbb{R}_{++}} (1-\bar{g}(z,x))[\tau_z(z) + \tau_x(x)]\,dH_x(x|z)\,h_z(z)\,dz \\
& \quad - \int_{\mathbb{R}_{++}} \int_{\mathbb{R}_{++}} \frac{\mathcal{T}_z'(z)\bar{\varepsilon}_z(z,x) + \mathcal{T}_x'(x)\mathbb{E}[x'_{\text{inc}}\varepsilon_z|z,x]}{1-\mathcal{T}_z'(z)} z\tau_z'(z)\,dH_x(x|z)\,h_z(z)\,dz \\
& \quad - \int_{\mathbb{R}_{++}} \int_{\mathbb{R}_{++}} \left\{ \frac{\mathcal{T}_z'(z)\mathbb{E}[x'_{\text{inc}}\varepsilon_z|z,x] + \mathcal{T}_x'(x)\mathbb{E}[(x'_{\text{inc}})^2\varepsilon_z|z,x]}{1-\mathcal{T}_z'(z)} z \right. \\
& \qquad\qquad\qquad + \left. \frac{\mathcal{T}_x'(x)\bar{\varepsilon}_{x|z}(z,x)x}{1+\mathcal{T}_x'(x)} \right\} \tau_x'(x)\,dH_x(x|z)\,h_z(z)\,dz.
\end{align*}

\smallskip
\underline{Vertically Neutral Reform}

Consider a linear tax reform $\tau_x(x) = x$ with compensating income tax changes $\tau_z(z) = -\mathbb{E}[\tau_x(X)|Z=z] = -\bar{x}(z)$, which holds average tax payments constant at each income level. This yields $\tau_z'(z) = -\bar{x}'(z)$.

\smallskip
\underline{Preference Heterogeneity Sufficient Statistic}

Define the preference heterogeneity sufficient statistic:
\[
x_{het}^{\prime}(z;\theta) \equiv \bar{x}'(z) - x'_{\text{inc}}(z;\theta).
\]
This is the residual of the cross-sectional relationship between $x$ and $z$ after removing the causal income effect for type $\theta$ agents.

\smallskip
\underline{Optimal Tax Condition}

For a linear tax $\mathcal{T}_x(x) = t_x x$, substituting the vertically neutral reform into the welfare effect and setting $\left.\frac{d\mathcal{L}}{d\kappa}\right|_{\kappa=0} = 0$ yields:
\begin{align*}0 & =-\int_{\mathbb{R}_{++}}\mathbb{C}[g(z;\theta),x(z;\theta)|z]h_{z}(z)\,dz\\  & \quad+\int_{\mathbb{R}_{++}}\frac{\mathcal{T}_{z}^{\prime}(z)}{1-\mathcal{T}_{z}^{\prime}(z)}\mathbb{E}[x_{het}^{\prime}(z;\theta)\varepsilon_{z}(z;\theta)|z]zh_{z}(z)\,dz\\  & \quad+t_{x}\int_{\mathbb{R}_{++}}\frac{1}{1-\mathcal{T}_{z}^{\prime}(z)}\mathbb{E}[x_{het}^{\prime}(z;\theta)x_{inc}^{\prime}(z;\theta)\varepsilon_{z}(z;\theta)|z]zh_{z}(z)\,dz\\  & \quad-\frac{t_{x}}{1+t_{x}}\mathbb{E}[x(z;\theta)\varepsilon_{x|z}(z;\theta)].\end{align*}

Under the assumption that $\mathbb{E}[\mathbb{C}(g(w,\theta), x(w,\theta)|z)] = 0$, the optimal tax satisfies equation \eqref{eq:opt_multidim_general}:
\[
t_x - \frac{\mathcal{D}'(\bar{x})}{\lambda} = \frac{\mathbb{E}_{z,\theta}[\text{RR}(z;\theta)\mathcal{T}_z'(z)]}{1-\mathbb{E}_{z,\theta}[\text{RR}(z;\theta)x'_{\text{inc}}(z;\theta)]},
\]
where $\text{RR}(z;\theta) = \eta_{x,z}^{\text{Taste}}(z;\theta) \left(\frac{\varepsilon_z(z;\theta)\bar{x}(z)}{\mathbb{E}[\varepsilon_{x|z}(z;\theta)x(z;\theta)]}\right)  \left(\frac{1+t_x}{1-\mathcal{T}_z'(z)}\right)$.

\end{proof}

\begin{proof}[Proof of Corollary \ref{prop:multidim_tractable}]
Suppose $\varepsilon_z(z;\theta) = \bar{\varepsilon}_z(z)$ for all $\theta$ at each $z$. Then the covariance terms in the numerator and denominator of the general optimal tax condition simplify.

For the numerator, $\mathbb{C}[\varepsilon_z(z;\theta), x'_{inc}(z;\theta)|z] = 0$ by assumption, leaving only the average terms.

For the denominator, using $x'_{het} = \bar{x}' - x'_{inc}$:
\begin{align*}
\mathbb{E}[x'_{\text{het}} x'_{\text{inc}}  \varepsilon_z | z] &= \bar{\varepsilon}_z(z) \mathbb{E}[(\bar{x}'(z) - x'_{\text{inc}}(z;\theta)) x'_{\text{inc}}(z;\theta) | z] \\
&= \bar{\varepsilon}_z(z) \left[\bar{x}'(z) \bar{x}'_{\text{inc}}(z) - \mathbb{E}[(x'_{\text{inc}})^2|z]\right] \\
&= \bar{\varepsilon}_z(z) \left[\bar{x}'(z) \bar{x}'_{\text{inc}}(z) - (\bar{x}'_{\text{inc}}(z))^2 - \mathbb{V}[x'_{\text{inc}}|z]\right] \\
&= \bar{\varepsilon}_z(z) \bar{x}'_{\text{het}}(z) \bar{x}'_{\text{inc}}(z) - \bar{\varepsilon}_z(z) \mathbb{V}[x'_{\text{inc}}|z].
\end{align*}

Substituting into the general formula and collecting terms yields equation \eqref{eq:opt_multidim_tractable}:
\[
t_x - \frac{\mathcal{D}'(\bar{x})}{\lambda} = \frac{\mathbb{E}_z[\overline{\text{RR}}(z)\mathcal{T}_z'(z)]}{1 - \mathbb{E}_z[\overline{\text{RR}}(z)\bar{x}'_{\text{inc}}(z)] + \mathbb{E}_z[\widetilde{\text{RR}}(z)\mathbb{V}(x'_{\text{inc}}(z;\theta)|z)]},
\]
where:
\[
\overline{\text{RR}}(z) = \bar{\eta}_{x,z}^{\text{Taste}}(z) \frac{\bar{\varepsilon}_z(z)\bar{x}(z)}{\mathbb{E}[\varepsilon_{x|z}(z;\theta)x(z;\theta)]} \frac{1+t_x}{1-\mathcal{T}_z'(z)},
\]
\[
\widetilde{\text{RR}}(z) = \frac{\bar{\varepsilon}_z(z) z}{\mathbb{E}[\varepsilon_{x|z}(z;\theta)x(z;\theta)]} \frac{1+t_x}{1-\mathcal{T}_z'(z)}.
\]

The variance term $\mathbb{V}[x'_{\text{inc}}|z]$ appears in the denominator with a positive sign, which attenuates any deviation from Pigouvian taxation. Only when $\mathbb{V}[x'_{\text{inc}}|z] = 0$ for all $z$ does the formula reduce to the unidimensional case.
\end{proof}

\end{document}